\documentclass{article}
\usepackage{listings}
\usepackage{cancel}
\usepackage{bbm}
\usepackage{graphicx}
\usepackage{subcaption}
\usepackage{mathtools}
\usepackage{amsmath}
\usepackage{amsfonts}
\usepackage{xr}
\usepackage{varioref}
\usepackage{xr-hyper}
\usepackage{hyperref}
\usepackage{multirow}
\usepackage{booktabs}
\usepackage{footnote}
\usepackage{float}  

\usepackage{amsthm} 
\usepackage[toc,page]{appendix}
\usepackage[font=small,labelfont=bf, skip=0pt]{caption}

\usepackage{array}
\newcolumntype{H}{>{\setbox0=\hbox\bgroup}c<{\egroup}@{}}

\usepackage[a4paper,margin=1in,footskip=0.25in]{geometry}

\usepackage[utf8]{inputenc}
\usepackage[style=authoryear-comp,firstinits=true,citestyle=authoryear,natbib=true,backend=bibtex,doi=false,isbn=false,url=true,maxcitenames=2,eprint=false,maxbibnames=99]{biblatex}
\AtBeginBibliography{}

\DeclareNameAlias{sortname}{last-first}

\renewbibmacro{in:}{}
\renewbibmacro*{volume+number+eid}{%
	\printfield{volume}%
	\setunit*{\addnbspace}
	\printfield{number}%
	\setunit{\addcomma\space}%
	\printfield{eid}}
\DeclareFieldFormat[article]{number}{\mkbibparens{#1}}

\bibliography{biblio}

\makeatletter
\def\thmhead@plain#1#2#3{%
	\thm@notefont{}
	\thmname{#1}\thmnumber{\@ifnotempty{#1}{ }\@upn{#2}}%
	\thmnote{ {\the\thm@notefont#3}}}
\let\thmhead\thmhead@plain
\itshape 
\makeatother

\newtheorem*{definition*}{Definition}

\newtheorem*{assumption*}{Assumption}
\newtheorem*{lemma*}{Lemma}
\newtheorem{lemma}{Lemma}
\newtheorem*{proposition*}{Proposition}
\newtheorem{proposition}{Proposition}
\newtheorem*{conjecture*}{Conjecture}
\newtheorem*{theorem*}{Theorem}
\newtheorem{theorem}{Theorem}
\newtheorem*{corollary*}{Corollary}
\newtheorem{corollary}{Corollary}

\newcommand{\indep}{\perp\!\!\!\!\perp} 

\usepackage{titlesec}
\titleformat*{\subsubsection}{\large\bfseries}

\graphicspath{ {images/} }

\usepackage[many]{tcolorbox}
\newsavebox{\fmbox}

\usepackage[onehalfspacing]{setspace}

\title{Identifying causal effects with subjective ordinal outcomes}
\author{Leonard Goff\thanks{\protect\linespread{1}\protect\selectfont Department of Economics, University of Calgary. For useful conversations, I thank Christopher Barrington-Leigh, Carol Caetano, Andrew Clark, Ben Crost, John Helliwell, Peter Hull, Caspar Kaiser, Louise Laage, Jean-William Lalibert\'e, Simon Lee, Erzo Luttmer, Guy Mayraz, Max Norton, Bernard Salani\'e, Adam Rosen, Kevin Song, Takuya Ura and Sam Viavant. I thank Erzo Luttmer and Social Science Research Services at the University of Wisconsin for help with data access.}}
\date{}

\begin{document}


\maketitle

\begin{abstract}
\large 
Survey questions often ask respondents to select from ordered scales where the meanings of the categories are subjective, leaving each individual free to apply their own definitions in answering. This paper studies the use of these responses as an outcome variable in causal inference, accounting for variation in interpretation of the categories across individuals. I find that when a continuous treatment variable is statistically independent of both i) potential outcomes; and ii) heterogeneity in reporting styles, a nonparametric regression of response category number on that treatment variable recovers a quantity proportional to an average causal effect among individuals who are on the margin between successive response categories. The magnitude of a given regression coefficient is not meaningful on its own, but the \textit{ratio} of local regression derivatives with respect to two such treatment variables identifies the relative magnitudes of convex averages of their effects. These results can be seen as limiting cases of analogous results for binary treatment variables, though comparisons of magnitude involving discrete treatments are not as readily interpretable outside of the limit. I obtain a partial identification result for comparisons involving discrete treatments under further assumptions. An empirical application illustrates the results by revisiting the effects of income comparisons on subjective well-being, without assuming cardinality or interpersonal comparability of responses.
\end{abstract}

\newpage

\large

\section{Introduction}
Many survey questions ask respondents to choose from a set of two or more ordered categories that lack clear definitions, leaving the interpretation of those categories to the respondent. Examples include self-reported health status (SRHS), product or service ratings, job satisfaction, and questions gauging satisfaction with life overall. Individuals' responses are then often used as an outcome variable in research, frequently as a proxy for some underlying latent variable of interest (e.g. true health in the case of SRHS).\footnote{A broad class of this type of survey questions that use so-called \textit{Likert scales}: e.g. allowing responses such as ``strongly agree'', ``agree'' \dots ``strongly disagree'' to indicate agreement with a given statement, or to categorize quantities such as frequencies (``often'', ``sometimes'', \dots ``almost never''). \citet{hamermesh2004} discusses the use of such outcomes in economics.}

A key question for this practice is how ``reporting functions''---the way that individuals map that latent variable into one of the available response categories---impact conclusions drawn from the data.\footnotemark   \hspace{.1cm} \citet{bondandlang} influentially show that even if individuals share a common reporting function (but it is not ex-ante known to the researcher), averages of the latent variable cannot be meaningfully compared between groups using their responses, absent strong restrictions on the latent variable's unobserved distribution. More fundamentally, if the response categories lack objective definitions, reporting functions might vary between individuals, potentially confounding any attempt to study relationships between explanatory variables and the latent variable. 

This paper shows that the observed categorical responses can nevertheless be informative about causal relationships in which this latent variable is the outcome, despite the dual threats of reporting functions being both i) unknown to the researcher and ii) heterogeneous across respondents. Taking the perspective of \citet{bondandlang} that the latent variable driving individuals' responses is the researcher's ultimate outcome of interest, I decompose differences in the observed joint distribution of responses and covariates to the causal effects of those covariates on the latent variable. I do so by strengthening the familiar selection-on-observables assumption that one or more explanatory variables are statistically independent of potential outcomes, adding to it that explanatory variables are also independent of heterogeneity in reporting functions (with both independence assumptions made conditional on observed control variables). Under this assumption I show how the estimand that arises from the common practice of regressing categorical response numbers on explanatory variables can be interpreted in terms of the causal effects of those regressors on the latent variable of interest. \footnotetext{The use of the term ``reporting function'' for subjective data appears to have first appeared in the economics literature in \citet{oswaldletters}, though the general concept predates its discussion in economics (e.g. \citealt{psycho}).} 

Concretely, I consider a general model of ordered response taking the form:
\begin{align}
	R_i &= r_i(H_i) = r(H_i, V_i) \nonumber \\
	H_i &= h_i(X_i) = h(X_i, U_i) \label{eq:modelintro}
\end{align}
where $H_i \in \mathbb{R}^K$ reflects a set of unobserved latent variables, and $R_i$ an observed response mapped to a real number in some set $\mathcal{R}$. For example, $\mathcal{R} = \{0,1\}$ for a binary yes/no question, or $\mathcal{R} = \{0,1,2,3,4\}$  for a question with five ordered response categories. I focus primarily on the case of a scalar latent variable $H \in \mathbbm{R}$, and later generalize to $K > 1$.

The function $h_i(x)$ in \eqref{eq:modelintro} denotes the potential outcomes of the latent variable for individual $i$, indicating the value of $H$ that would occur if a vector of observed explanatory variables $X$ took each counterfactual value $x$. The function $r_i(h)$ represents individual $i$'s reporting function, which I assume to be weakly increasing in $h$ for each $i$. The random vectors $U_i$ and $V_i$ parameterize heterogeneity across individuals in potential outcomes and reporting functions, respectively. The main statistical assumption of the model is that $X_i \indep (U_i,V_i)$, which I relax to \textit{conditional} independence given control variables. The researcher's objective is to learn how $h_i(x)$ varies with $x$, observing only $R_i$ and $X_i$.

One of the key implications of my results is that if $X_{1i}$ and $X_{2i}$ reflect two continuously distributed components of the vector $X_i$, and $\mathcal{R}$ is associated with a set of integers, then
\begin{equation} \label{eq:compare}
	\frac{\mathbbm{E}[\partial_{x_2}\mathbbm{E}[R_{i}|X_i]]}{\mathbbm{E}[\partial_{x_1}\mathbbm{E}[R_{i}|X_i]]} = \frac{\tilde{\beta}_2}{\tilde{\beta}_1}
\end{equation}
where $\tilde{\beta}_j$ reflects a convex weighted average across individuals of the causal effect of a small change in the $j^{th}$ component of $X$ on $H$. In particular, $\tilde{\beta}_j$ averages the causal partial derivative $\partial_{x_j} h(X_i,U_i)$ over individuals $i$ who are on the margin between two response categories $r-1$ and $r$ for any $r \in \mathcal{R}$.\footnote{$\partial_{x_j} h(X_i,U_i)$ denotes $\partial_{x_j} h(x,U_i)$ with $x_j$ the $j^{th}$ component of $x$, evaluated at $x=X_i$ (and similarly for $\partial_{x_j}\mathbbm{E}[R_{i}|X_i]$)} If the conditional expectation $\mathbbm{E}[R_i|X_i]$ happens to be linear, then the average derivative quantity $\mathbbm{E}[\partial_{x_j}\mathbbm{E}[R_{i}|X_i]]$ in the LHS of \eqref{eq:compare} is simply the coefficient on $X_{j}$ in a linear regression of $R$ on $X$. In this case, Eq. \eqref{eq:compare} affords a causal interpretation to the ratio of OLS regression coefficients for two continuous treatments.

Throughout the paper, I discuss results through an application to survey questions that ask respondents about their overall satisfaction with life, and for ease of exposition refer to the latent variable $H$ as ``happiness''.\footnote{This simplified language ignores e.g. distinctions between hedonic, affective and evaluative notions of well-being \citep{DEATON201818,helliwellcpbl}.}  For example, the popular Cantril Ladder question asks individuals to describe their satisfaction with life on an eleven point scale from $0$ to $10$.\footnote{A popular version of the Cantril ladder question asks: \textit{Please imagine a ladder with steps numbered from zero at the bottom to ten at the top. Suppose we say that the top of the ladder represents the best possible life for you and the bottom of the ladder represents the worst possible life for you. If the top step is 10 and the bottom step is 0, on which step of the ladder do you feel you personally stand at the present time?} \citep{gallup}.} Questions like this about general well-being motivate treating the latent variable $H$ as an outcome of normative interest, drawing on the notion of cardinal utility as a measure of welfare \citep{Fleming1952,Harsanyi1955}. With this interpretation, the marginal rates of substitution between treatment variables are a key input for normative analysis, suggesting trade-offs that would be welfare improving for individuals. However, my results are also applicable to other outcomes elicited on ordered scales, e.g. general or mental health status, job satisfaction, product or service ratings, and other settings in which ordered response models might be employed with individual-specific heterogeneity in the thresholds between response categories. 

Despite a growing trend in papers leveraging natural experiments with subjective outcome data,\footnote{Some prominent examples include \citet{cardetal,aermrs,lindqvistetal2020,aerperez,pnasdwyerdunn}.} empiricists have lacked formal results such as Eq. \eqref{eq:compare} to interpret precisely what is estimated by regressions in which subjectively-defined ordinal responses $R$ are used as the dependent variable. This paper helps to fill the gap by showing that when the selection-on-observables research design is extended to include reporting-function heterogeneity, derivatives of the conditional expectation function of integer category numbers on $X$ reveal positive aggregations of the local causal effects of $X$ on $H$.\footnote{I also show that when the researcher is interested in establishing correlations rather than causation, the same results capture changes to the conditional quantile function of the underlying latent variable, rather than causal effects.} The weights in this aggregation have an intuitive form but are not under the researcher's control. This illuminates the limits for identification of overall unweighted means of causal effects, which correspond to the parameter analyzed by \citet{bondandlang}. My results show that mean regression can nonetheless remain a useful tool for analyzing more general weighted averages of effects, without assuming cardinality or interpersonal comparability of $H$.

I apply my formal results to revisit the influential study of \citet{luttmer2005}, who considers the effects of household income as well as the incomes of one's neighbors on satisfaction with life. Using a selection-on-observables strategy and linear regression adjustment, \citet{luttmer2005} finds a positive coefficient on own-income along with a negative coefficient on neighbor income, suggesting that relative income comparisons are important for subjective well-being. My nonparametric identification results corroborate this interpretation under the maintained exogeneity assumptions, but without assuming cardinality or interpersonal comparability of individuals' responses to the well-being question. Empirically, I first report distributional regressions of $\mathbbm{1}(R_i \le r)$ on $X$ for each $r$. The patterns suggest that differences in the coefficients across $r$ are driven by the unknown distribution of the underlying latent variable, underscoring the theoretical observation that coefficients must be compared between variables to be quantitatively meaningful. The heterogeneity in coefficients across $r$ in fact cancels in the ratio, and I cannot reject equality across $r$ of the local marginal rates of substitution between own and neighbor income (among respondents on the threshold between $r$ and $r+1$). I also estimate these ``marginal'' respondents to be similar to inframarginal respondents in terms of gender and education. The empirical results overall are consistent with simple models of heterogeneity in potential outcomes and/or response functions that render the effects for marginal respondents somewhat typical of the population, in this particular setting.

When the treatment variables of interest are \textit{discrete}, rather than continuous as above, I find that comparisons of magnitude become more complicated. First, I show that when one compares the mean of $R$ between two fixed values $x$ and $x'$ of the vector $X$: \begin{equation} \label{eq:discreteintro}
	\mathbbm{E}[R_i|X_i=x']-\mathbbm{E}[R_i|X_i=x]=\mathbbm{E}\left[\bar{f}(\Delta_i,V_i,x)\cdot \Delta_i\right],
\end{equation}
where $\Delta_i = h(x',U_i)-h(x,U_i)$ is the treatment effect of changing $X$ from $x$ to $x'$ on outcome $H$ for individual $i$. The ``weight'' $\bar{f}(\Delta_i,v_i,x_i)$ is unknown but positive for all $i$, and Eq. (\ref{eq:discreteintro}) thus implies that if the sign of the treatment effect $\Delta_i$ is the same for all individuals, then the sign of $\mathbbm{E}[R_i|X_i=x']-\mathbbm{E}[R_i|X_i=x]$ will be the same as that of the causal effect. However, the magnitude of $\mathbbm{E}\left[\bar{f}(\Delta_i,V_i,x)\right]$ can in general depend on the values $x$ and $x'$ being compared, and quantitative comparisons of regression coefficients can be misleading (even if the regression is correctly specified) if one or more of the treatment variables being considered is discrete and treatment effects are not small.\footnote{The function $\bar{f}$ is defined in Sec. \ref{secdiscrete}, and no longer depends upon $\Delta$ as $x' \rightarrow x$ and the difference becomes a derivative.} 

Eq. \eqref{eq:discreteintro} provides a new perspective on the key point made by \citet{bondandlang}, who argue that the conditional distributions $R_i|X_i=x'$ and $R_i|X_i=x$ are generally uninformative about the sign of $\mathbbm{E}[H_i|X_i=x']-\mathbbm{E}[H_i|X_i=x]$, even if it is assumed that all individuals share a common reporting function. Knowing the sign of the difference in means of $H_i$ between groups $x$ and $x'$ from observations of $R_i$ generally requires that the quantile functions of $H_i|X_i=x'$ and $H_i|X_i=x$ do not cross, i.e. that the latent variable distribution for one group stochastically dominates that of the other. This assumption cannot be verified from the data $(R,X)$, and may be implausible if $x$ and $x'$ are two populations (men vs. women, two countries, etc.), that are each quite heterogeneous in themself and differ from one another across many dimensions. However, when $x'$ and $x$ differ in a single component representing \textit{treatment} in e.g. a quasi-experimental setting, it may be possible to argue that treatment effects $\Delta_i$ are not too heterogeneous.\footnote{Indeed, the much stronger assumption of complete homogeneity in treatment effects is often made implicitly to motivate a causal interpretation of regression models. For example, $h(x,u) = g(x)+u$ yields the regression $H_i = g(X_i)+U_i$ but implies that $\Delta_i = g(x')-g(x)$ for all $i$.} If $\Delta_i$ has the same sign for all $i$, that sign is equal to that of $\mathbbm{E}[R_i|X_i=x']-\mathbbm{E}[R_i|X_i=x]$. Thus while the argument made by \citet{bondandlang} is compelling for generic comparisons between two groups, it may have less bearing on settings where a clear research design is leveraged to interpret differences in $R_i$ causally.

Implications of my results for regression analysis using subjective ordinal outcomes are threefold. First, the focus on finding natural experiments popular in modern applied work yields a previously unrecognized benefit for subjective outcomes: reporting functions may become uncorrelated with treatment variables of interest $X$, affording inference on the direction of causal effects of $X$ on unobserved $H$. Second, researchers can move beyond interpretations of the sign of average effects and consider magnitudes only when multiple valid treatment variables are available. Third, such comparisons of magnitude are most informative when the two variables being compared are continuous rather than discrete. An implication is that identification in experimental work with subjective outcome variables would benefit from randomizing the quantitative ``doses'' of multiple treatments.\\ 

\noindent \textit{Outline of paper:} In Section \ref{secmodel} I propose a general nonparametric model of ordered response with nonseparable heterogeneity: it allows each respondent to have their own response function, but takes treatment variables $X$ to be conditionally independent of all unobserved heterogeneity. Section \ref{seccontiniousvariation} establishes my main identification result when there is continuous variation in $X$, which provides a generalization of Eq. \eqref{eq:compare}. I outline assumptions under which Eq. \eqref{eq:compare} in turn reveals a local average marginal rate of substitution between two continuous treatment variables. Section \ref{sec:empirical} applies these results to revisit the findings of \citet{luttmer2005} relating the effect of one's own income and one's neighbors' incomes on life satisfaction. 

In Section \ref{secdiscrete}, I turn to identification with a discrete treatment variable. After showing that ratios of regression coefficients involving one or more discrete regressors lack the guarantee of a simple quantitative interpretation like Eq. \eqref{eq:compare}, I describe how one can obtain bounds on the ratio of the total weight that the conditional expectation function applies to causal effects when comparing continuous to discrete variation in $X$. The analytic results suggest that when there are many response categories and individual reporting functions are approximately linear, discrete contrasts will tend to \textit{overstate} causal effects relative to regression derivatives, by a factor that is upper bounded by two. I assess this implication through simulations with a variety of assumed distributions of the latent variable, and only find evidence of appreciable distortion when treatment effects are made implausibly large in the DGP.

Appendix \ref{sec:bl} provides an extended discussion of how my results relate to \citet{bondandlang}. Appendix \ref{sec:orderedresponse} relates my general model of ordered response to ones previously considered in the literature. Appendix \ref{sec:extensions} considers several extensions to my baseline model, such as using instrumental variables rather than selection-on-observables for identification, or allowing for a multivariate latent variable. Appendices \ref{sec:additionalcont} and \ref{seccontinuousreporting} develop some supporting theoretical results for the paper. Appendix \ref{sec:empiricalmore} provides additional results for the empirical application, while Appendix \ref{sec:regression} expands on the implications of my results for practical regression analysis and presents a numerical illustration. 

\section{Model} \label{secmodel}


Suppose that there exists a meaningful latent value $H_i$ for each individual which the researcher is ultimately interested in as an outcome. With the life satisfaction example in mind, I will often refer to $H_i$ as $i's$ underlying ``happiness'', which the researcher aims to learn about given those individuals' responses $R_i$.\footnote{The model extends naturally to a setting in which the definition of ``$H$'' is itself subjective, in the sense that different individuals use different latent variables when constructing their responses. The key requirement is that these subjectively defined latent variables in turn reflect increasing transformations of an objective variable of interest. See Appendix \ref{sec:extended}.} Section \ref{sec:params} discusses the interpretation of $H_i$ as a measure of utility. In the body of this paper I take $H_i$ to be a scalar, but Appendix \ref{sec:multivariate} extends results to the vector case.

The researcher observes a sample of $(R_i,X_i,W_i)$ across individuals $i$ generated as:
\begin{align}
	R_i &= r_i(H_i) = r(H_i, V_i) \label{modelr}\\
	H_i &= h_i(X_i) = h(X_i, U_i) \label{modelh}
\end{align}
where $r_i(h)$ is in individual-specific function mapping happiness $h$ to the space of possible responses $\mathcal{R}$. The above model indexes heterogeneity in $r_i(\cdot)$ by a heterogeneity parameter $V_i \in \mathcal{V} \subseteq \mathbbm{R}^{d_V}$. Since no constraints are placed on $d_V$, this is without loss of generality and the model is compatible with each individual having their own reporting function $r_i(h)$. Figure \ref{fighonest} depicts two examples of reporting functions when $\mathcal{R} = \{0,1,2\}$.

For each individual there is a function $h_i(\cdot)$ mapping values of a vector of $J$ explanatory variables $X$ into a value of $H$ via (\ref{modelh}), where heterogeneity in the function $h_i(\cdot)$ is represented by parameter $U_i \in \mathcal{U} \subseteq \mathbbm{R}^{d_U}$. The primary interpretation of the function $h_i(x)$ is that it denotes potential outcomes for individual $i$ as a function of counterfactual values of $x$, in some set of possible treatments $\mathcal{X} \subseteq \mathbbm{R}^{J}$.\footnote{An alternative interpretation of $h(x,u)$ is always also available and requires no causal assumptions, which is that $h$ represents the conditional quantile function of $H_i$ given $X_i$, with $U_i \in [0,1]$ a scalar indicating $i$'s rank in a distribution of their peers. \label{eqcondquantile} In particular, let $\theta_i:=F_{H|XVW}(H_i|X_i,V_i,W_i)$ be $i$'s ``rank'' in the conditional happiness distribution of individuals sharing their value of $X,V$ and $W$, where $F_{H|XVW}$ denotes a cumulative distribution function of $H$. Now let $U_i = (\theta_i, V_i,W_i)^T$, and define $h(x,u):=Q_{H|XVW}(\theta|x,v,w)$ for any $u=(\theta,v,w)^T$, where $Q_{H|XVW}$ denotes the conditional quantile function of $H$ given $X,V$ and $W$. Eq. (\ref{modelh}) now follows from these definitions. See Appendix \ref{sec:idio} for details. This representation is helpful when causal effects are not the target, and the researcher is instead interested in uncovering statistical features of the joint distribution between $H_i$ and $X_i$.} Since the dimension $d_U$ is again left unrestricted, the above model places no restriction on heterogeneity in potential outcomes and causal effects across individuals.

Finally, $W_i \in \mathcal{W} \subseteq \mathbbm{R}^{d_W}$ is a vector of additional observed variables to be used as control variables in the analysis. These $W_i$ can be thought of as variables that matter for happiness but are not necessarily manipulable (e.g. race), as components of $U_i$ that are potentially correlated with $X_i$ but are observable, or as correlates of $X_i$ that proxy for reporting function heterogeneity $V_i$. In settings with stratified randomization, $W_i$ isolates the experimental strata.

\subsection{Causal parameters of interest} \label{sec:params}
The function $h(x,u)$ is our main object of interest: how it varies with $x$ holding $u$ fixed yields the causal effect of that change on $H$. For example, $h(x',U_i)-h(x,U_i)$ denotes the ``treatment effect'' for unit $i$ of moving between two counterfactual values $x$ and $x'$ of the vector $X$. I consider the identification of such discrete treatment effects in Section \ref{secdiscrete}.

For most of the analysis, I will consider small changes in one or more components of $x$ that are continuously distributed. Letting $\partial_{x_j}$ denote a partial derivative with respect to $x_j$, the function $\partial_{x_j} h(x,U_i)$ for a given individual $i$ characterizes the effect of a small change in the $j^{th}$ component of $X$ on $H$, when $X=x$. An average of the value of this derivative across individuals $i$ provides a summary of the marginal effect of $x_j$ on $H$ when $X=x$. More generally, such averages can employ weights $\rho_i$ that depend on the individual-level observables $(X_i,W_i)$ and unobserved heterogeneity parameters $(U_i,V_i)$. For example, for a given function $\rho(u,v,x,w)$, we might consider a weighted average of the form:
\begin{equation} \label{eq:genavg}
	\tilde{\beta}_j = \mathbbm{E}[\rho_i \cdot \partial_{x_j} h(X_i,U_i)]
\end{equation}
where $\rho$ is chosen such that $\rho_i:=\rho(U_i,V_i,X_i,W_i)$ is positive with probability one and satisfies $\mathbbm{E}[\rho_i] =1$. 
Many results of this paper represent, intuitively, limits of parameters of the form $\tilde{\beta}_j$ for a sequence of such weighting functions $\rho(\cdot)$.\footnote{For example, the average derivative $\mathbbm{E}[\partial_{x_j} h(x,U_i)|h(x,U_i)=h]$ that conditions on a single value $h$ for $h_i(x)$ represents the limit of $\tilde{\beta}_j$ for the function $\rho(U_i,V_i)=\frac{\mathbbm{1}(h(x,U_i) \in [h, h+\epsilon])}{\mathbbm{E}[\mathbbm{1}(h(x,U_i) \in [h, h+\epsilon])]}$, as $\epsilon \rightarrow 0$. See also discussion after proof of Theorem \ref{propflow}.}

If we interpret $H$ as a measure of ``utility'', then $h_i(\cdot)=h_i(\cdot,U_i)$ can be thought of as $i$'s utility function, and $H_i=h_i(X_i)$ as their realized utility (evaluated at $i$'s actual $X_i$). Under this interpretation the \textit{ratio} of two derivatives of $h_i(x)$ represents a local marginal rate of substitution of $X_1$ for $X_2$ when $X=x$, for individual $i$, e.g. 
$$MRS_i(x):=\frac{\partial_{x_2}h(x,U_i)}{\partial_{x_1} h(x,U_i)}$$
Note that this interpretation only requires $h_i(\cdot)$ to represent utility in an ordinal sense: $MRS_i(x)$ yields the slope of the indifference curve for $i$ that passes through the point $x$. Among individuals for whom $X_i=x$, the quantity $MRS_i(x)$ yields a marginal rate of substation at their actual value of $X_i$. Weighted averages of this realized marginal rate of substitution across individuals take the form $\widetilde{MRS}:=\mathbbm{E}\left[\rho_i\cdot MRS_i(X_i)\right]$ for $\rho_i=\rho(U_i,W_i,X_i,W_i)$ defined as following Eq. \eqref{eq:genavg}, or a limit of $\widetilde{MRS}$ for a sequence of such functions.

Finally, this paper will consider weighted averages of discrete \textit{treatment effects} between two fixed values of $X$, i.e. $\Delta_i:=h(x',U_i)-h(x,U_i)$ for some $x,x' \in \mathcal{X}$. Weighted averages of treatment effects take the form:
$$\widetilde{\Delta}:=\mathbbm{E}\left[\rho(U_i,W_i,X_i,W_i)\cdot \Delta_i\right]$$
with $\rho_i:=\rho(U_i,V_i,X_i,W_i)$ as above, or the limit of $\widetilde{\Delta}$ for a sequence of such functions $\rho$.

\subsection{Model assumptions}
Note that model (\ref{modelr})-(\ref{modelh}) embeds an exclusion restriction: $X$ does not directly enter in the equation for $R$, and only affects reports through $H$. This is important for drawing inferences about the relationship between $H$ and $X$ from the observable joint distribution of $R$ and $X$. The model can be generalized slightly to allow reporting behavior to depend directly on observables, as described in Appendix \ref{sec:testing}.

The following two subsections introduce the two key identifying assumptions of the model: first, that reporting functions are weakly increasing in $h$; and second, that the researcher as exogenous variation in some components of $X_i$. These assumptions are, under suitable regularity conditions, sufficient for the main results of this paper. The basic model is therefore more general than existing models of ordered response, which typically couple parametric assumptions with an assumption that there is no heterogeneity in $v$. Appendix \ref{sec:orderedresponse} shows how the model nests models previously considered in the literature.

\subsubsection{First assumption: reporting functions are weakly increasing} \label{sec:reporting}

The main assumption that I make about the reporting functions $r_i(\cdot)$ themselves is that they are \textit{increasing} in $H_i$:
\begin{assumption*}[MONO (weakly increasing reporting functions)]
	$r(h,v)$ is weakly increasing and left-continuous in $h$ for all $v \in \mathcal{V}$
\end{assumption*}
\noindent Appendix \ref{sec:multivariate} extends Assumption MONO to the case in which $H_i$ is a random vector, assuming that $r_i(\cdot)$ is weakly increasing in each component of $H_i$. Note that MONO does \textit{not} assume the effect of $X$ on $H$ to be monotonic or uniform across individuals.

The first part of Assumption MONO rules out cases in which individuals would report a lower value of $R$ if $H$ were increased. The left-continuity assumption of MONO is essentially a normalization, since any weakly increasing function of bounded variation is continuous except at isolated points within its support.\footnote{Hence a reporting function that is, say, right continuous rather than left continuous could be made left continuous by modifying the function on a set of Lebesque measure zero.}

\begin{figure}[H]
	\begin{center} \vspace{.2cm}
		\includegraphics[height=2in]{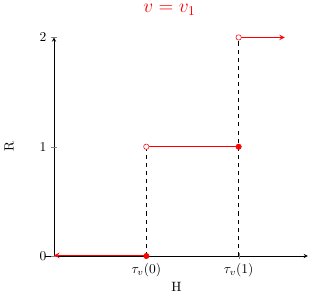}	\quad \quad \quad \quad \quad
		\includegraphics[height=2in]{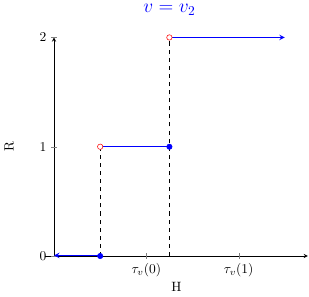}
		\caption{Examples of two different reporting functions, in a case with three categories: $\mathcal{R}=\{0,1,2\}$. The reporting function depicted in the right panel is more ``optimistic'' than the one in the left panel, as the threshold value of $H$ for $R=1$ and $R=2$ are both lower than for the reporting function on the left (see Lemma \ref{prophonest}).} \label{fighonest} 
	\end{center}
\end{figure}

The following lemma shows that Assumption MONO is equivalent to there being a set of ``thresholds'' $\tau_v(r)$ that separate the ordered categories in $\mathcal{R}$. This characterization is useful in developing the formal results to come.
\begin{lemma} \label{prophonest}
	MONO holds iff for all $v \in \mathcal{V}, r \in \mathcal{R}$ and $h \in \mathcal{H}$: 
	\begin{equation} \label{eq:iff}
		r(h,v) \le r \iff h \le \tau_{v}(r)
	\end{equation}
	where $\tau_{v}(r) = \sup\{h \in \mathcal{H}: r(h,v) \le r\}$ or $\tau_v(r):=\infty$ if the supremum does not exist.
\end{lemma}
\begin{proof}
	All proofs are given in Appendix \ref{sec:proofs}.
\end{proof}
\noindent As an illustration of Lemma \ref{prophonest}, suppose that $\mathcal{R} = \{0,1,\dots \bar{R}\}$ for some integer $\bar{R}$. Then Lemma \ref{prophonest} implies that any given reporting function $r(h,v)$ can be written as:
\begin{equation}
	r(h,v) = 
	\begin{cases}
		0 & \textrm{ if } h \le \tau_v(0)\\
		1 & \textrm{ if } \tau_v(0) < h \le \tau_v(1)\\
		2 & \textrm{ if } \tau_v(1) < h \le \tau_v(2)\\
		\vdots &\\
		\bar{R} & \textrm{ if } h > \tau_v(\bar{R}-1)
	\end{cases}
\end{equation}

\noindent \textit{Remark: } Assumption MONO does not require that respondents are motivated only by ``honesty'' when choosing $R_i$. Instead, they may have direct preferences for certain response categories. Consider a utility maximization model in which $ r(h,v) = \textrm{argmax}_{r \in \mathcal{R}} \textrm{ }u(r,h,v)$, with utility $u$ depending not only on happiness $h$, but also directly on the response category $r$. As an example, let us further assume that the utility function takes the form $u(r,h,v) = \phi_v(r)-|h^*_v(r)-h|$ where individuals of type $v$ obtain utility $\phi_v(r)$ from giving a response of $r$, but also value giving an answer close to a value $h^*_v(r)$ they perceive to correspond to response $r$. Provided that $h^*_v(r)$ is strictly increasing in $r$ (i.e. higher responses are subjectively associated with higher values of happiness), then $u$ satisfies the property of \textit{increasing differences} (cf. \citealt{shannonmilgrom}) in $(r,h)$, which in turn implies MONO.\footnote{Note that heterogeneity $v$ in this form for utility need not be additively separable from quantities that depend on $x$ (i.e. $h$). Such separability is shown by \citet{allenrehbeck} to admit important identification results for latent utility.}\footnote{MONO also allows there to be individuals with preferences that \textit{only} depend on $r$, giving the same response regardless of their $H_i$. Such individuals will not contribute to regression derivatives and differences of $R$ on $X$ under EXOG.}

\subsubsection{Second assumption: random variation in treatment variables} \label{sec:ci}
The final piece of the model is a conditional independence assumption for variation in $X$. In particular, I suppose that conditional on $W$, the treatments $X$ are as-good-as-randomly assigned in the following sense:
\begin{assumption*}[EXOG (conditionally exogenous components of $X$)] i) $\{X_{i} \indep V_i \} \textrm{ }| \textrm{ } W_{i}$; and ii) $\{X_{i} \indep U_i \} \textrm{ }| \textrm{ } (W_{i},V_i)$
\end{assumption*}
\noindent A sufficient condition for Assumption EXOG is that:
\begin{equation} \label{eq:idx}
	\{X_{i} \indep (U_i, V_i) \} \textrm{ }| \textrm{ } W_{i}
\end{equation} 
Eq. (\ref{eq:idx}) provides a natural foundation for EXOG and is simpler to motivate, but is technically stronger than the results require.\footnote{Eq. (\ref{eq:idx}) can equivalently be expressed as $\{(U_i,V_i) \indep X_{j,i}\}| (X_{-j,i},W_i)$ for all $j$, where $X_{-j,i}$ denotes the elements of $X_i$ apart from $X_{j,i}$. Assumption EXOG can be re-expressed similarly.} For causal inference, an assumption like $\{X \indep U\}|W$ is generally already necessary for identification: one needs some kind of experiment or natural experiment providing exogenous variation in $X$. Eq. (\ref{eq:idx}) then simply requires this natural experiment to also render $X$ (conditionally) independent of $V$. Note that under EXOG, $U$ and $V$ may be arbitrarily correlated with one another (e.g. if happier individuals have more optimistic reporting functions).\footnote{This is a feature that distinguishes my approach from the treatment of measurement error by \citet{hausmanabreyava}, who assume (in my notation) that $R \indep X | H$, which amounts to $V \indep U | H$. They also restrict the model functionally, with a linear index structure for $h$ and scalar errors with monotonicity.} In Appendix \ref{sec:iv}, I relax EXOG to consider identification using instrumental variables. Appendix \ref{sec:regression} illustrates through an example how violations of EXOG can affect results.

The assumption that response behavior is independent of a treatment variable may be restrictive in many contexts, especially in the absence of a credible research design. Appendix \ref{sec:testing} describes one specific threat to Assumption EXOG, that reporting functions might themselves be affected by the treatment variables $X$. I show there that EXOG can be relaxed slightly, and in fact tested under additional structural assumptions. Whether reporting functions might themselves be affected by a given treatment variable must be considered on a case-by-case basis.\footnote{Other approaches to allowing for reporting-function heterogeneity that do not require EXOG rely on particular models of that heterogeneity (e.g. \citealt{CPBL}) or auxiliary data sources. Such sources include ``anchoring vignettes'' \citep{kingetal, kapetyn,molina,jeboreversing,stantchevaannrev}, memories of past life satisfaction \citep{kaiserjebo}, calibration questions \citep{adjustingforscaleuse} and survey response times \citep{happytimes}.}

\section{What is identified from continuous variation in $X$} \label{seccontiniousvariation}

Given the model outlined in the last section, let us consider what can be identified by looking at responses given variation in $X$. In this section, I suppose that at least one component of $X$ is continuously distributed.

Denote by $f_H(h|x,v,w)$ the density of $H_i$ at $h$, conditional on $X_i=x$, $V_i=v$ and $W_i=w$, and assume the following:
\begin{assumption*}[REG$_j$ (regularity conditions for $X_j$)]
	The following hold for given $j$: i) $X_{ji}$ is continuously distributed; ii) $f_{H|XVW}(h|x,v,w)$ exists; iii) $\partial_{x_j}  Q_{H|XVW}(\alpha|h,v,w) \le M < \infty$ for all $\alpha \in [0,1]$, $h \in \mathcal{H}$, where $Q_{H|XVW}$ is the conditional quantile function of $H$ given $X,V,W$; iv) for each $x,w$ and $h$, $f_{H,\partial_{x_j}h(x,U)|XVW}(h,h'|x,v,w)$ exists and is upper bounded by some $c(h')$ where $\int c(h')|h'|dh' < \infty$, for all $v \in \mathcal{V}$.
\end{assumption*}
\noindent Assumption REG reflect standard regularity conditions, as described in \citet{hoderleinmammen2007}. The only substantive modification above is that I take the conditions to hold conditional on each reporting function type $V_i=v$.

\subsection{Derivatives of the response distribution in terms of causal responses} \label{sec:singler}
Let $P(R_{i} \le r|x,w):=P(R_{i} \le r|X_i=x,W_i=w)$ denote the observed distribution of responses $R_i$ given values $x$ of treatments $X_i$ and $w$ of the control variables $W_i$. For brevity, I will often use this type of shorthand in long expressions.
\begin{theorem}[] \label{propflow}
	Assume MONO and EXOG hold REG$_j$ holds for a $j \in \{1,\dots, J\}$. Then:
	$$\partial_{x_j} P(R_i \le r|x,w) = -\mathbbm{E}\left\{\left.{f_H(\tau_{V_i}(r)|x,V_i,w)} \cdot \mathbbm{E}\left[\partial_{x_j} h(x,U_i)|H_i=\tau_{V_i}(r),x,V_i,w\right]\right|W_i=w\right\}$$
\end{theorem}
\noindent Theorem \ref{propflow} shows that the derivative of $P(R_i \le r|X_i=x,W_i=w)$ with respect to changes in $x_j$ provides a positively-weighted linear combination of the causal response in $H$ due to $X_j$: ``marginal'' causal effects $\partial_{x_j} h(x,U_i)$ due to a small change in $X_j$. The proof of Theorem \ref{propflow} relates the derivative of the conditional CDF of $R$ to a mixture of (infeasible) quantile regressions that condition on response type $V_i$ (Lemma \ref{lemma:hdist}), and then makes use of a connection between quantile regressions and local average structural derivatives \citep{hoderleinmammen2007,sasaki_2015}. As an intermediate step in establishing Theorem \ref{propflow}, Lemma \ref{lemma:hdist} in Appendix \ref{sec:proofs} shows establishes the connection between $\partial_{x_j} P(R_i \le r|x,w)$ and the conditional quantiles of $H$ given $X$.\\

\noindent \textit{Example: }Theorem \ref{propflow} generalizes the well-known formula for ``marginal effects'' in the probit model: $\partial_{x_j}P(R_i=1|X_i=x) = \sigma^{-1}\phi(x^T \beta/\sigma)\cdot \beta_j$, where $\phi$ is the standard normal probability density function. In the probit model, $v$ is degenerate and the single threshold $\tau_v(0)=0$, while $h(x,u) = x^T \beta + u$ and $H_i|X_i=x \sim \mathcal{N}(x^T\beta, \sigma^2)$. Thus, $f_H(\tau(0)|x)=f_H(0|x)=1/\sigma\cdot \phi(-x'\beta/ \sigma)=\phi(x'\beta)$.\\

\textbf{Normalization: }  It is well-known that $\beta$ in the probit model is only identified up to an overall scale normalization, often achieved by fixing the variance of the error distribution $\sigma^2=1$. Similarly, we lack from Theorem \ref{propflow} the ability to pin down the overall scale of derivatives of the structural function $\partial_{x_j}h(x,U_i)$. The inner expectation in Theorem \ref{propflow} (indicated by square brackets [ ]) is over heterogeneity in causal effects $U_i$, while the outer expectation (indicated by curly brackets \{ \}) is over heterogeneity $V_i$ in reporting functions. Expanding this second expectation out, we have
\begin{equation} \label{eq:propflow}
	\partial_{x_j} P(R_i \le r|x,w) =-\int dF_{V|W}(v|w) \cdot f_H(\tau_{v}(r)|x,v,w) \cdot \mathbbm{E}\left[\partial_{x_j} h(x,U_i)|H_i=\tau_{v}(r),x,v,w\right]
\end{equation} 
The weights $dF_{V|W}(v|w) \cdot f_H(\tau_{v}(r)|x,v,w)$ that multiply the conditional expectation do not necessarily integrate to one---indeed all that we can say about $\mathbbm{E}[f_H(\tau_{V_i}(r)|x,V_i,w)|W_i=w]=\int dF_{V|W}(v|w) \cdot f_H(\tau_{v}(r)|x,v,w)$ is that it is positive. However, considering the ratio of \textit{two} derivatives cancels out the dependence on this unknown scale:
\begin{align} \label{eqprobratio}
	\frac{\partial_{x_2} P(R_{i}\le r|x,w)}{\partial_{x_1} P(R_{i}\le r|x,w)} &= \frac{\mathbbm{E}\left\{\omega_r(x,V_i,w)\cdot \mathbbm{E}\left[\partial_{x_2} h(x,U_i)|H_i=\tau_v(r),x,v,w\right]|W_i=w\right\}}{\mathbbm{E}\left\{\omega_r(x,V_i,w)\cdot \mathbbm{E}\left[\partial_{x_1} h(x,U_i)|H_i=\tau_v(r),x,v,w\right]|W_i=w\right\}}
\end{align}
where $\omega_r(x,v,w):= f_H(\tau_v(r)|x,v,w)/\mathbbm{E}[f_H(\tau_{V_i}(r)|x,V_i,w)|W_i=w]$. The function $\omega_r$ yields weights that are positive and integrate to one, i.e. $\mathbbm{E}[\omega_r(x,V_i)|X_i=x,W_i=w]=1$. To contrast this with the positive but non-normalized integration measure that appears in (\ref{eq:propflow}), I refer to weights such as the $\omega_r$ appearing in (\ref{eqprobratio}) as ``convex''. Note that the convex weight applied to each group characterized by $H_i=\tau_v(r),X_i=x,V_i=v,W_i=w$ is exactly the same in both the numerator and denominator of (\ref{eqprobratio}). Eq. \eqref{eqprobratio} can be seen as a ratio $\tilde{\beta}_2/\tilde{\beta}_1$ of two parameters of the form $\tilde{\beta}_j=\mathbbm{E}[\rho_i \cdot \partial_{x_j} h(X_i,U_i)]$ described in Section \ref{sec:params}, where $\rho_i$ picks out individuals with $H_i$ close to $\tau_{V_i}(r)$ (and $X_i=x,W_i=w$).\footnote{\label{fn:borel}In particular, let $\rho_i:=\rho(U_i,V_i,X_i,V_i)=\frac{\mathbbm{1}\left\{h(x,U_i) \in B_\epsilon^1(\tau_{V_i}(r)), X_i \in B_\epsilon^J(x), W_i \in B_\epsilon^{d_W}(w)\right\}}{\mathbbm{E}\left[\mathbbm{1}\left\{h(x,U_i) \in B_\epsilon^1(\tau_{V_i}(r)), X_i \in B_\epsilon^J(x), W_i \in B_\epsilon^{d_W}(w)\right\}\right]}$, where $B_\epsilon^d(p)$ denotes an open ball of radius $\epsilon$ centered around $p \in \mathbbm{R}^d$, e.g. $B_\epsilon^1(p) = (p-\epsilon,p+\epsilon)$. Then consider the limit of $\tilde{\beta}_j$ as $\epsilon \rightarrow 0$. See the end of the proof of Theorem \ref{propflow} for more details about this limit.}

Note that the practice sometimes seen in applied work of reporting standardized or ``beta'' coefficients (in which each regressor $X_j$ is normalized against it's standard deviation) would break this important property of (\ref{eqprobratio}). In that case, the ratio of the total weights appearing in the numerator and denominator would become $sd(X_{2})/sd(X_{1})$ rather than unity. By contrast, rescaling regression coefficients only by the standard deviation of the \textit{outcome} $R_i$ leaves \eqref{eqprobratio} unchanged.\\

\textbf{Intuition for Theorem \ref{propflow}: } By Eq. (\ref{eq:propflow}), the ``weight'' in the observable $\partial_{x_j} P(R_i \le r|x,w)$ placed on an individual with happiness close to $\tau_v(r)$ is positive and proportional to $dF_{V|W}(v|w) \cdot f_H(\tau_{v}(r)|x,v,w)$. Figure \ref{figflow} provides intuition for this particular weighting.

\begin{figure}[H]
	\begin{center} \vspace{.2cm}
		\includegraphics[height=2in]{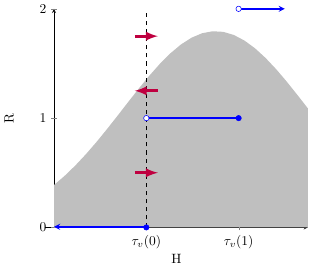}	\quad \quad \quad \quad \quad
		\includegraphics[height=2in]{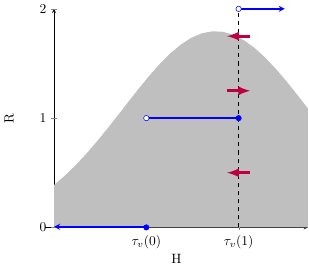}	
		\caption{Intuition for Theorem \ref{propflow}: the derivative of $P(R_i\le r|X_i=x)$ with respect to $x_j$ captures the ``flow'' of individuals over threshold $\tau_v(r)$ due to a small change in $x_j$. Left: $\partial_{x_j} P(R_i\le 1|X_i=x)$ captures flows over $\tau_v(0)$. Right: $\partial_{x_j} P(R_i\le 1|X_i=x)$ captures flows over $\tau_v(1)$. The gray shaded curve in the background depicts the density of $H_i$.} \label{figflow}
	\end{center}
\end{figure}
Suppose for simplicity there are no controls $w$. By the law of iterated expectations, we can write $\partial_{x_j} P(R_i \le r|X_i=x)$ as a weighted average of $\partial_{x_j} P(R_i \le r|X_i=x,V_i=v)$ across the various reporting functions $v$ in the population. For a given $v$, $\partial_{x_j} P(R_i \le r|X_i=x,V_i=v)$ captures the ``flow'' of individuals over the threshold $\tau_{v}(r)$ due to a small change in $x_j$, in one direction or the other. Some of these individuals can have negative effects: $\partial_{x_j} h(x,U_i)<0$, denoted by arrows to the left in Figure \ref{figflow}. Others can have positive effects $\partial_{x_j} h(x,U_i)>0$, indicated by rightward arrows in Figure \ref{figflow}. The net effect captured by $\partial_{x_j} P(R_i \le r|X_i=x,V_i=v)$ depends on the average derivative $\mathbbm{E}\left[\partial_{x_j} h(x,U_i)|H_i=\tau_{v}(r),x,v\right]$ local to the threshold. Since the derivative $\partial_{x_j}$ considers an infinitesimal change in $X$, any such ``flow'' over the threshold requires a positive density there: $f_H(\tau_{v}(r)|x,v)>0$.\footnote{The quantity $f_H(h|x,v) \cdot \mathbbm{E}\left[\partial_{x_j} h(x,U_i)|H_i=h,x,v\right]$ at a given $h$ is sometimes referred to as a ``flow density'', and appears in \cite{kasy2022}, \citet{goff2022} and in the physics of fluids, where it arises from the conservation of mass.}

While Theorem \ref{propflow} is specific to a fixed value of $X_i=x$ (and $W_i=w$), Corollary \ref{corr:avgderivative} to come shows that averaging back over the distribution of $X_i,Z_i$ yields a simpler formula for the average derivative: $\mathbbm{E}[\partial_{x_j} P(R_i \le r|X_i,W_i)] = -f_{H-\tau_{V}(r)}(0) \cdot \mathbbm{E}\left[\partial_{x_j} h(X_i,U_i)|H_i=\tau_{V_i}(r)\right]$. This again captures a an average causal response among respondents who are located at their individual-specific threshold $\tau_{V_i}(r)$, up to a non-identified but positive scale factor. The estimand of Theorem \ref{propflow} is a more disaggregated parameter, representing a more fundamental identification result.

I refer to individuals with $H_i=\tau_{V_i}(r)$ for some $r$ as ``marginal'', or ``indifferent'' between response categories. Theorem \ref{propflow} shows that local derivatives of the distribution of $R_i$ conditional on $X_i$ and $R_i$ only average causal effects among these marginal respondents. These marginal respondents averaged over in the RHS of Theorem \ref{propflow} cannot be individually identified, since neither $H_i$ nor $\tau_{V_i}(r)$ are observed for a given $i$. However, I show in Appendix \ref{sec:characterizing} that if the sign of causal effects is assumed to be common across individuals, average characteristics of the marginal respondents can be identified (Section \ref{sec:empirical} provides an implementation). Appendix \ref{app:helphurt} shows that reporting function heterogeneity can have a counter-intuitive benefit: if the heterogeneous thresholds $f_{V_i}(r)$ are so spread out that they are approximately uniform across the support of $H_i$, then $\partial_{x_j}P(R_i \le r|x,w)$ is proportional to $\mathbbm{E}\left[\partial_{x_j} h(x,U_i)|X_i=x,W_i=w\right]$, which averages over both the marginal and infra-marginal respondents having $X_i=x$ and $W_i=w$.

\subsection{Implications of Theorem \ref{propflow} for mean regression at a point} \label{sec:ratios}
Beyond the case of binary survey questions, researchers do not typically estimate regressions of response the CDF evaluated at a fixed category $r$, as contemplated by Theorem \ref{propflow}. However, the result allows us to study the more common practice of modeling the conditional mean of $R_i$ given $X_i$. To see this, suppose that $\mathcal{R}$ consists of integers $\{0,1, \dots, \bar{R}\}$ for some $\bar{R}$. Note that the following identity holds for all $i$:
\begin{equation} \label{eqidentity}
	R_i = \sum_{r=1}^{\bar{R}} \mathbbm{1}(r \le R_i)=\sum_{r=0}^{\bar{R}-1} \mathbbm{1}(r < R_i)
\end{equation}
From this it then follows that for any $x$: $\mathbbm{E}[R_i|X_i=x] = \sum_{r =0}^{\bar{R}-1} P(r<R_i|X_i=x)=\bar{R}-\sum_{r =0}^{\bar{R}-1} P(R_i \le r|X_i=x)$. Then, applying Theorem \ref{propflow}:
\begin{align} 
	&\partial_{x_j} \mathbbm{E}[R_i|X_i=x,W_i=w] \nonumber \\
	& \hspace{1cm} = \int dF_{V|W}(v|w) \cdot \sum_{r =0}^{\bar{R}-1} f_H(\tau_{v}(r)|x,v,w) \cdot \mathbbm{E}\left[\partial_{x_j} h(x,U_i)|H_i=\tau_{v}(r),x,v,w\right] \label{eq:expflow}
\end{align}
For brevity, I use the shorthand $\sum_r$ for the definite sum $\sum_{r =0}^{\bar{R}-1}$.  Collecting (\ref{eq:expflow}) across all continuous regressors, we can summarize as:
\begin{corollary} \label{propnabla}
	Under the assumptions of Theorem \ref{propflow}, if $\mathcal{R}=\{0,1, \dots, \bar{R}\}$ then for each $j$ that satisfies $REG_j$:
	\begin{equation*}
		\partial_{x_j} \mathbbm{E}[R_i|x,w] = \mathbbm{E}\left\{\left.\sum_{r} f_H(\tau_{V_i}(r)|x,V_i,w) \cdot \mathbbm{E}\left[\partial_{x_j}  h(x,U_i)|H_i=\tau_{V_i}(r),x,V_i,w\right]\right|W_i=w\right\}
	\end{equation*}
\end{corollary}

\noindent \textit{Remark:} if instead of the integers, the researcher associates alternative numerical values $r_j$ with the ordered responses $\mathcal{R}$, where $r_0 < r_1 < \dots < r_R$, then instead of (\ref{eqidentity}) we have $R_i = r_0+\sum_{j=0}^{R-1} (r_{j+1}-r_j)\cdot \mathbbm{1}(r_j < R_i)$. The above results thus generalize with $f_H(\tau_{v}(r_j)|x,v,w)$ upweighted by the positive factor $(r_{j+1}-r_{j})$. This implies that different labeling schemes could be used in estimation to achieve different weightings over local causal effects, though the most information one could learn is by simply repeating Theorem \ref{propflow}, one $r$ at a time. When considering mean regression, using integer category labels is natural in that it weighs each threshold in proportion to its occupancy, as demonstrated in Eq. \eqref{eqexpratio} below. Corollary \ref{propnabla} also holds unchanged so long as $\mathcal{R}$ reflects any set of consecutive integers, with $\Sigma_r$ denoting a sum over all but the highest integer in $\mathcal{R}$.\\

\begin{figure}[H]
	\begin{center} \vspace{.2cm}
		\includegraphics[height=2in]{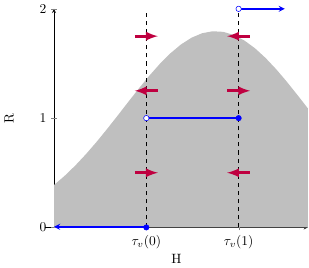}
		\caption{With $\mathcal{R}$ integers, the derivative of $\mathbbm{E}[R_i\le r|X_i=x]$ with respect to $x_j$ captures the ``flow'' of individuals over either threshold $\tau_v(r)$ due to a small change in $x_j$. Compare to left and right panels of Fig \ref{figflow}.} \label{figflow2}
	\end{center}
\end{figure}

Another way to express Corollary \ref{propnabla} is to let $\tau_v  := \{\tau_v(r)\}_{r \in \mathcal{R}}$ denote the set of all thresholds for individuals with reporting function $v$. Then for each $j$ that satisfies $REG_j$:
\begin{equation} \label{eqtauV}
	\partial_{x_j} \mathbbm{E}[R_i|x,w]  = \mathbbm{E}\left[\rho(x,V_i,w) \cdot  \partial_{x_j} h(x,U_i)|H_i \in \tau_{V_i}, X_i=x,W_i=w\right]
\end{equation}
where $\rho(x,v,w):=\sum_r f_H(\tau_v(r)|x,v,w)$ we assume that $\lim_{h\rightarrow \infty}f_H(h|x,v,w)=0$ and that for each $v \in \mathcal{V}$, the $\tau_v(r)$ are all distinct.\footnote{Let $A$ and $B$ be random variables, where $B$ is absolutely continuous and let $\mathcal{B}$ be a finite set of distinct values. Assume that $\mathbbm{E}[A|B=b]$ is continuous in $b$, so we can then define $\mathbbm{E}[A|B \in \mathcal{B}]$ simply as $\lim_{\epsilon \downarrow 0} \mathbbm{E}[A|\min_{b \in \mathcal{B}} |B-b| < \epsilon]$ which works out to $\sum_{b \in \mathcal{B}} \frac{f_{B}(b)}{\sum_{b' \in \mathcal{B}}f_{B}(b')} \cdot \mathbbm{E}[A|B=b]$.} This expression shows that $ \partial_{x_j} \mathbbm{E}[R_i|x,w]$ averages over all units having $X_i=x$ (and $W_i=w$), located at \textit{any} of their individual-specific happiness thresholds, with (positive but not convex) weights $\rho(X_i,V_i,X_i)$. 

As in (\ref{eqprobratio}), if we consider the ratio of such regression derivatives for two continuous treatment variables $X_1$ and $X_2$, the ``total'' weight cancels out:
\begin{align} \label{eqexpratio}
	\frac{\partial_{x_2}\mathbbm{E}[R_{i}|X_i=x,W_i=w]}{\partial_{x_1}\mathbbm{E}[R_{i}|X_i=x,W_i=w]} &= \frac{\tilde{\beta}_2(x,w)}{\tilde{\beta}_1(x,w)}
\end{align}
where $\tilde{\beta}_j(x,w):=\mathbbm{E}\left[\omega(x,V_i,w)\cdot \partial_{x_j} h(x,U_i)|H_i \in \tau_{V_i},X_i=x,W_i=w\right]$ and $ \omega(x,v,w):=\rho(x,v,w)/\mathbbm{E}\left[\rho(x,V_i,w)|H_i \in \tau_{V_i},X_i=x,w\right]$. The quantity $\tilde{\beta}_j(x,w)$ is thus a convex combination of causal effects with respect to $X_j$ across individuals in the population, in the sense described in Sections \ref{sec:params} and \ref{sec:singler}. Note that the weights $\omega$ appearing in the numerator and denominator are the same for any $(x,w)$.

\subsection{Averaging local regression derivatives back over $X$ and $W$} \label{app:avgderivatives}

Theorem \ref{propflow} shows how observable derivatives $\partial_{x_j}P(R_i \le r|x,w)$ and $\partial_{x_j}\mathbbm{E}[R_i|x,w]$ can be interpreted in terms of average causal effects among individuals $i$ who are marginal between response categories and for whom $X_i=x$, $W_i=w$. These local derivatives at a given $x,w$ are identified by a non-parametric regression of $\mathbbm{1}(R_i\le r)$ or $R_i$ on $X_i$ and $W_i$, respectively.

Corollary \ref{corr:avgderivative} shows furthermore that if one averages these local regression derivatives across the observable distribution of $X_i,W_i$, one obtains an average causal effect that remains ``local'' to individuals who are on the margin between response categories, but is no longer specific to individuals having a particular value of $X_i$ and $W_i$:
\begin{corollary} \label{corr:avgderivative}
	Under the assumptions of Theorem \ref{propflow}: $$\mathbbm{E}[\partial_{x_j} P(R_i \le r|X_i,W_i)] = -f_{H-\tau_{V}(r)}(0) \cdot \mathbbm{E}\left[\partial_{x_j} h(X_i,U_i)|H_i=\tau_{V_i}(r)\right]$$
\end{corollary}
\noindent The density $f_{H-\tau_{V}(r)}(0)$ is not identified by the data, but it does not depend on $j$. Thus we have as in (\ref{eqprobratio}) that this unidentified density cancels out in ratios, i.e. $\frac{\mathbbm{E}[\partial_{x_2} P(R_i \le r|X_i,W_i)]}{\mathbbm{E}[\partial_{x_1} P(R_i \le r|X_i,W_i)]} = \frac{\mathbbm{E}\left[\partial_{x_2} h(X_i,U_i)|H_i=\tau_{V_i}(r)\right]}{\mathbbm{E}\left[\partial_{x_1} h(X_i,U_i)|H_i=\tau_{V_i}(r)\right]}$, and if $\mathcal{R} = \{0,1, \dots \bar{R}\}$ we have similarly for the mean that:
\begin{align} \label{eq:avgexpratio}
	\frac{\mathbbm{E}[\partial_{x_2}\mathbbm{E}[R_{i}|X_i,W_i]]}{\mathbbm{E}[\partial_{x_1}\mathbbm{E}[R_{i}|X_i,W_i]]} 
	&= \frac{\tilde{\beta}_2}{\tilde{\beta}_1} = \frac{\mathbbm{E}[\partial_{x_2} h(X_i,U_i)|H_i \in \tau_{V_i}]}{\mathbbm{E}[\partial_{x_1} h(X_i,U_i)|H_i \in \tau_{V_i}]}
\end{align}
where $\tilde{\beta}_j:=\sum_{r =0}^{\bar{R}-1} \omega_r\cdot \mathbbm{E}[\partial_{x_j} h(X_i,U_i)|H_i = \tau_{V_i}(r)]$ and $ \omega_r:=\frac{f_{H-\tau_{V}(r)}(0)}{\sum_{r' =0}^{\bar{R}-1}f_{H-\tau_{V}(r')}(0)}$ are positive weights that sum to one. Response thresholds $r$ that are more ``populated'' in the sense that $f_{H-\tau_{V}(r)}(0)$ is larger, receive higher weight, in such a way that $\tilde{\beta}_j = \mathbbm{E}[\partial_{x_j} h(X_i,U_i)|H_i \in \tau_{V_i}]$.\footnote{i.e. $\sum_{r' =0}^{\bar{R}-1} f_{H-\tau_{V}(r)}(0) \cdot \mathbbm{E}\left[\partial_{x_j} h(X_i,U_i)|H_i=\tau_{V_i}(r)\right] = \rho \cdot \mathbbm{E}[\partial_{x_j} h(X_i,U_i)|H_i \in \tau_{V_i}]$, with $\rho = \sum_{r' =0}^{\bar{R}-1}f_{H-\tau_{V}(r')}(0)$.} In the case with no control variables $W_i$, we then obtain Eq. \eqref{eq:compare} stated in the introduction.

If the conditional mean function $\mathbbm{E}[R_i|X_i=x,W_i=w]$ happens to be linear in $x$ and $w$, then the quantity $\mathbbm{E}[\partial_{x_j}\mathbbm{E}[R_{i}|X_i,W_i]]$ on the LHS of Eq. \eqref{eq:avgexpratio} and \eqref{eq:mrsavg} is simply the coefficient $\gamma_j$ from the OLS regression
\begin{equation} \label{eq:ols}
	R_i = \gamma_1 X_{1i}+\gamma_2 X_{2i} + \dots + \gamma_J X_{Ji} + \lambda^TW_i+\epsilon_i
\end{equation}
where the vector of control variables $W$ includes a constant. While specification \ref{eq:ols} is the standard in empirical practice, Appendix \ref{sec:regression} discusses the implications of this practice when the functional form is misspecified, i.e. when $\mathbbm{E}[R_i|X_i=x,W_i=w]$ is not actually linear but the researcher proceeds in estimating \eqref{eq:ols} anyways. Such issues are generally a concern when selection on observables identification arguments are implemented via linear regression, and are not specific to the use of subjective ordinal outcome variables.

\subsection{Marginal rates of substitution} \label{sec:weaksepmain}
Equation \eqref{eqexpratio} shows that a ratio of regression derivatives at $X=x$ identifies the ratio of a conditional average causal effect of $X_2$ on $H$ to the same conditional average of the effect of $X_1$ on $H$, among individuals for whom $X=x$. Similarly, \eqref{eq:avgexpratio} shows that a ratio of average regression derivatives (or simply OLS regression coefficients in the case of a linear conditional mean) has a similar interpretation, but averaging over $x$. \citet{luttmer2005} and \citet{ditellaetal} represent two prominent empirical studies in which the relative magnitude of regression coefficients (with subjective well-being as the dependent variable) is interpreted as yielding the implicit trade-off between two goods.

In general, a ratio of averages is not the same as an average of ratios, and thus neither \eqref{eqexpratio} nor  \eqref{eq:avgexpratio} immediately yields an average marginal rate of substitution parameter $\widetilde{MRS}$ of the form introduced in Section \ref{sec:params}. For example, Equation \eqref{eq:avgexpratio} does not immediately yield an average of $MRS_i(X_i)$. A sufficient condition however is that $Cov\left(\left.MRS_i(X_i), \partial_{x_1} h(X_i,U_i)\right|H_i \in \tau_{V_i}\right) = 0$. In this case
\begin{equation} \label{eq:mrsavg}
	\frac{\mathbbm{E}[\partial_{x_2}\mathbbm{E}[R_i|X_i,W_i]]}{\mathbbm{E}[\partial_{x_1}\mathbbm{E}[R_i|X_i,W_i]]}= \mathbbm{E}\left[\left.MRS_i(X_i)\right|H_i \in \tau_{V_i}\right]
\end{equation}
capturing the average marginal rate of substitution between $X_1$ and $X_2$, among respondents who are marginal at any threshold, i.e. $H_i = \tau_{V_i}(r)$ for some $r$. This covariance condition says that heterogeneity in $MRS_i(X_i)$ across individuals is uncorrelated with heterogeneity in the magnitude of the marginal effect of $X_1$ alone. 

Proposition \ref{propmrs} of Appendix \ref{sec:mrs} also shows how a similar result to Eq. \eqref{eq:mrsavg} holds using the estimand $\frac{\partial_{x_2}\mathbbm{E}[R_{i}|X_i=x,W_i=w]}{\partial_{x_1}\mathbbm{E}[R_{i}|X_i=x,W_i=w]}$ of \eqref{eqexpratio}, which fixes a value of $x$. In this case note that variation in $MRS_i(x)$ conditional on $H_i$ and $X_i$ comes from $U_i$ alone. Thus if $U_i$ is degenerate conditional on $X_i$ and the value of $H_i$ (e.g. if $h$ is invertible in a scalar $u$), then the needed covariance condition holds automatically. Proposition \ref{propmrs} generalizes this with a covariance restriction similar to the above, which conditions on $X_i$ and $W_i$. Proposition \ref{propmrs} also shows how one can obtain a one-sided bound on the RHS of \eqref{eq:mrsavg} by relaxing this to assume a known \textit{sign} of the correlation between $MRS_i(X_i)$ and $\partial_{x_1} h(X_i,U_i)$. 

Below I consider two particular cases in which assuming some natural structure for the function $h(x,u)$ is sufficient to interpret $\frac{\partial_{x_2}\mathbbm{E}[R_{i}|X_i=x,W_i=w]}{\partial_{x_1}\mathbbm{E}[R_{i}|X_i=x,W_i=w]}$ as a marginal rate of substitution, and the ratio of averages of such derivatives as an average of such marginal rates of substitution across individuals.

\subsubsection{The weakly separable special case} \label{sec:weaksepmain}
We say that the potential outcomes function $h(x,u)$ is weakly separable between $x$ and $u$ when
\begin{equation} \label{eqweaksep}
	h(x,u) = \texttt{h}(g(x),u),
\end{equation}
i.e. some function $g: \mathcal{X}\rightarrow \mathbbm{R}$ aggregates over the treatments $X$ into a scalar $g(x)$, which is then combined through $\texttt{h}$ with heterogeneity $u$ in a way that may or may not be additively separable. For example, a linear model $h(x,u) = x^T\beta + u$ sets $g(x) = X^T\beta$ and $\texttt{h}(g,u) = g+u$, combining a linear causal response with an additive scalar error term. For any $g(x)$, the additively separable form $\texttt{h}(g,u) = g+u$ is equivalent to imposing that the causal effect of changing between any treatment values $x$ and $x'$ is the same for all individuals. This assumption is implicit in much empirical work employing regressions with subjective outcome data.

When (\ref{eqweaksep}) holds, Eq. \eqref{eqtauV} yields
\begin{align} \label{eqexpratio2}
	\frac{\partial_{x_2}\mathbbm{E}[R_{i}|x,w]}{\partial_{x_1}\mathbbm{E}[R_{i}|x,w]} &= \frac{\textcolor{purple}{ \int dF_{V|W}(v|w) \cdot \rho(x,v,w)} \cdot \partial_{x_2} g(x) \cdot \textcolor{purple}{\mathbbm{E}\left[ \partial_{g} \texttt{h}(g(x),U_i)|H_i\in \tau_v,x,v,w\right]}}{\textcolor{purple}{\int dF_{V|W}(v|w) \cdot \rho(x,v,w)} \cdot \partial_{x_1} g(x) \cdot \textcolor{purple}{\mathbbm{E}\left[ \partial_{g} \texttt{h}(g(x),U_i)|H_i\in \tau_{v},x,v,w\right]}} \nonumber \\ & =\frac{\partial_{x_2} g(x)}{\partial_{x_1} g(x)}
\end{align}
where the highlighted factors cancel out in the numerator and denominator, since the derivatives of $g(x)$ do not depend on $v$. In the weakly separable model the marginal rate of substitution between $X_1$ and $X_2$ when $X=x$ is the same for all individuals and equal to $\frac{\partial_{x_2} g(x)}{\partial_{x_1} g(x)}$. Thus the ratio of local regression derivatives at a point $X=x$ identifies \textit{the} MRS at that point $x$.\footnote{Weakly separable models for ordered response in which $u$ is a scalar have been studied by \citet{matzkin1994}. Appendix \ref{sec:weaksep} discusses how Eq. \eqref{eqexpratio2}, which does not require $u$ to be a scalar, relates to that body of work.} A testable implication of the weakly separable model is therefore that the LHS of Eq. \eqref{eqexpratio2} does not depend on the value of the controls $w$.\footnote{Another testable implication is that $\partial_{x_2}P(R_{i} \le r|x,w)/\partial_{x_1}P(R_{i} \le r|x,w)$ does not depend on $r$. Appendix \ref{sec:testing} applies this insight to test the assumption that $X$ does not directly affect reporting functions. \citet{DHAULTFOEUILLE2024105075} consider testable restrictions of a similar weakly-separable structure in certain IV models, while relaxing exclusion.}

\subsubsection{The quasilinear special case} \label{sec:quasilinear}

Suppose that $h$ represents preferences and for each individual, these preferences are quasi-linear in $X_1$ such that $h(x,u) = x_1 + h(x_2, \dots x_J,u)$.\footnote{Any preference relation that is quasi-linear in $X_1$, continuous, and strictly ``increasing'' in $X_1$ admits of a representation $h(x,u) = x_1 + h(x_2, \dots x_J,u)$ \citet{rubinstein}. In the other direction, we can see that if $h(x,u)$ is a representation of quasi-linear preferences that is strictly increasing in $x_1$ and differentiable in $x_2 \dots x_J$, then it must be the case that $h(x,u) = \phi(x_1 + h(x_2, \dots x_J,u))$ where $\phi_u$ is a strictly increasing function.}  Quasi-linear utility is widely used in economics to simplify welfare analysis (see e.g. \citealt{FENG2025105927}).\footnote{Although quasi-linearity is a property of \textit{preferences}, we can think of $h(x,u) = x_1 + h(x_2, \dots x_J,u)$ as a cardinalization of these ordinal preferences (which will generally differ by individual $i$) in which a unit increase in $x_1$ has equal weight for any individual in population expectations involving $h$. Under this normalization, $\mathbbm{E}[h(x,U_i)]$ for example  represents a utilitarian social welfare function whose value is unaffected by transfers of $x_1$ between individuals.} When the elements of $X$ are priced, quasilinearity in $X_1$ can also deliver demand functions for the remaining goods that do not depend on income (see e.g. \citealt{qljet}).

In this case the condition $Cov\left(\left.MRS_i(X_i), \partial_{x_1} h(X_i,U_i)\right|H_i \in \tau_{V_i}\right) = 0$ is satisfied trivially, because $\partial_1 h(x,U_i) = 1$ with probability one. Thus we have that $\frac{\mathbbm{E}[\partial_{x_2}\mathbbm{E}[R_i|X_i,W_i]]}{\mathbbm{E}[\partial_{x_1}\mathbbm{E}[R_i|X_i,W_i]]}= \mathbbm{E}\left[\left.MRS_i(X_i)\right|H_i \in \tau_{V_i}\right]$. Further $\frac{\partial_{x_2}\mathbbm{E}[R_i|X_i=x,W_i=w]}{\partial_{x_1}\mathbbm{E}[R_i|X_i=x,W_i=w]} = \mathbbm{E}\left[\left.MRS_i(x)\right|H_i \in \tau_{V_i}, X_i=x,W_i=w\right]$, which can be derived as a special case of Proposition \ref{propmrs} given in Appendix \ref{sec:mrs}.


\section{Empirical illustration} \label{sec:empirical}
In a prominent paper, \citet{luttmer2005} studies the effects of absolute and relative income on life satisfaction, investigating whether individuals draw on social comparisons in assessing their personal well-being. To do so, \citet{luttmer2005} merges data from the 1987 and 1992 waves of the U.S. National Survey of Families and Households (NSFH)---which contains a question on self-reported satisfaction with life along with self-reported socioeconomic data---to information on the local average earnings for a given household constructed from the Current Population Study and the 1990 Census.

In the notation of the present paper, let $i$ denote the primary respondent of an individual household in the NSFH. We consider two treatment variables $X_i=(X_{1i},X_{2i})$, where $X_{1i}$ denotes the log of household income for $i$'s household (self-reported in the NSFH) and $X_{2i}$ denotes average predicted log earnings in the Public Use Microdata Area (PUMA) in which $i$ lives. The construction of this variable is described in detail in \citet{luttmer2005}. $R_i$ denotes $i$'s response to the question ``taking things all together, how would you say things are these days?'', reported on a one to seven Likert-type scale in which a response of one indicates ``very unhappy'' and seven ``very happy''.\footnote{The intermediate values 2-6 do not have associated descriptions in the survey (e.g. ``somewhat happy''), and are labeled by integers only.} Finally, $W_i$ represents a vector of control variables that includes home size/type/value, employment, education, gender, marriage, race religion, state fixed effects and PUMA characteristics.

I follow \citet{luttmer2005} and focus on households in which the main respondent was married in both waves of the NSFH. Details on the sample construction are provided in Appendix \ref{sec:empiricalmore}. While I let $i$ denote the main respondent for a household, the primary specification of \citet{luttmer2005} averages values of $R_i$ and $W_i$ between the main respondent and their spouse, finding very similar results. I focus on the individual-level specification for two reasons: i) it affords a more straightforward interpretation through the lens of the results of this paper, given that the main respondent and their spouse may have different reporting functions; and ii) I explore departures from linear models, where averaging across observations within a household does not affect the functional form of the regression. Nevertheless, results with this averaging are provided in Appendix \ref{sec:spouse}.

\subsection{Basic result interpreted through the lens of Theorem \ref{propflow}}

The main results of \cite{luttmer2005} exploit a selection-on-observables strategy, estimating an OLS regression of $R_i$ on $X_i$ and $W_i$, i.e. Eq \eqref{eq:ols}:
\begin{equation} \label{eq:ols2var}
	R_i = \gamma_1 X_{1i}+\gamma_2 X_{2i} + \lambda^TW_i+\epsilon_i
\end{equation}
and ascribing a causal interpretation to the coefficients $\gamma_1$ and $\gamma_2$. Luttmer uses fixed effects regressions as well as data on movers between PUMAs to argue that selection due to neighborhood choice is not a major concern in this context. Luttmer further argues that individuals' definitions of ``very happy'' or ``very unhappy'' are not affected by $X$, by replicating the qualitative results with other outcome variables that are expected to be less prone to this threat. I refer the reader to sections IV.B and IV.C of \citet{luttmer2005} for details. These arguments motivate making Assumption EXOG in this context.

Luttmer finds that an increase in household earnings increases subjective well-being $\gamma_1 > 0$, while an increase in the earnings of one's neighbors decreases subjective well-being $\gamma_2 < 0$. This provides evidence that well-being is influenced not only by one's absolute income, but also one's relative income compared with the reference group of one's neighbors.\footnote{This finding has since been replicated using experimental variation in beliefs about relative income \citep{jansens}.} While \citet{luttmer2005} also reports estimates that instrument for own-income to overcome potential measurement error, I focus on magnitudes from the benchmark OLS regression \eqref{eq:ols2var}. In this specification, the positive coefficient on own income has about half the magnitude as the negative coefficient on PUMA (neighbors') income. That is, if one's PUMA were to go up by 1\%, one's own income would need to go up by about 2\% to leave the respondents' well-being unaffected.

\begin{table}[h!]
	\centering
	{
\def\sym#1{\ifmmode^{#1}\else\(^{#1}\)\fi}
\begin{tabular}{l*{5}{c}}
\hline\hline
          &\multicolumn{1}{c}{(1)}&\multicolumn{1}{c}{(2)}&\multicolumn{1}{c}{(3)}&\multicolumn{1}{c}{(4)}&\multicolumn{1}{c}{(5)}\\
          &\multicolumn{1}{c}{OLS}&\multicolumn{1}{c}{Luttmer Table 1}&\multicolumn{1}{c}{Semiparametric}&\multicolumn{1}{c}{OLS}&\multicolumn{1}{c}{Kernel}\\
\hline
Own ln income&   0.0877\sym{***}&    0.111\sym{***}&    0.122\sym{***}&   0.0446\sym{***}&    0.125\sym{***}\\
          &   (3.78)         &   (4.62)         &  (10.82)         &   (3.94)         &  (11.05)         \\
PUMA ln income&   -0.229\sym{**} &   -0.248\sym{**} &   -0.202\sym{**} &   -0.169\sym{**} &   -0.246\sym{***}\\
          &  (-2.73)         &  (-2.99)         &  (-3.28)         &  (-2.75)         &  (-4.01)         \\
\hline
Ratio PUMA/own&   -2.614         &   -2.234         &   -1.581         &   -3.792         &   -1.626         \\
se(ratio) &    1.160         &        .         &        .         &    1.534         &        .         \\
\hline    &                  &                  &                  &                  &                  \\
Controls  &        X         &        X         &        X         &                  &                  \\
Clustered se&        X         &        X         &                  &        X         &                  \\
\hline    &                  &                  &                  &                  &                  \\
Sample size&     7939         &     8023         &     7939         &     7939         &     7939         \\
\hline\hline
\multicolumn{6}{l}{\footnotesize \textit{t} statistics in parentheses}\\
\multicolumn{6}{l}{\footnotesize \sym{*} \(p<0.05\), \sym{**} \(p<0.01\), \sym{***} \(p<0.001\)}\\
\end{tabular}
}
\\ \vspace{.5cm}
	\caption{Replication of \citet{luttmer2005}'s results for the main respondent, and alternative non-linear estimators. Standard errors are clustered at the PUMA level unless otherwise noted. For semiparametric and non-parametric columns, the first two rows report average local derivatives, and ``Ratio'' measures the average ratio of local derivatives, cf. Eq. \eqref{eqexpratioavg2}---see Footnote \ref{fn:ses} for further details.} \label{table:mainself}
\end{table}

I confirm this finding qualitatively in Column (1) of Table \ref{table:mainself}. Column (2) reports the numerical results from Table 1 of \citet{luttmer2005} (main respondent column), in which $\hat{\gamma}_2/\hat{\gamma}_1 = -2.23$. In column (1) I implement regression \eqref{eq:ols2var} the publicly available NSFH data merged with the PUMA income variable constructed by Luttmer (the replication data construction is described in Appendix \ref{sec:empiricalmore}). I obtain similar results in both sign and magnitude, with $\hat{\gamma}_1 = 0.0877$ and $\hat{\gamma}_2 = -0.229$ for a ratio of $\hat{\gamma}_2/\hat{\gamma}_1 = -2.614$.\footnote{That I am not able to match the numerical results exactly is likely explained by the many choices involved in how exactly to define some of the control variables, or possible updates to the underlying NSFH data over the last two decades.} 

If regression \eqref{eq:ols2var} is correctly specified---that is $\mathbbm{E}[R_i|X_i,W_i]$ is indeed a linear function of $X_i$ and $W_i$---then Theorem \ref{propflow} implies via \eqref{eqexpratio} that the quantity $\tilde{\beta}_2(x,w)/\tilde{\beta}_1(x,w)$ is approximately constant over values $x$ of the treatment variables and $w$ of the control variables (and equal to $-2.614$), where recall that $\tilde{\beta}_j(x,w)$ is a weighted average of $\partial_{x_j} h(x,U_i)$ over individuals with $U_i$ such that their happiness is exactly at the threshold between two response categories when $X_i=x,W_i=w$. This is consistent for example with a structural function that takes the linear form $h(x,u) = x^T \beta+u$, in which case $\tilde{\beta}_2(x,w)/\tilde{\beta}_1(x,w) = \beta_2/\beta_1 = -2.614$. However, the magnitudes of $\beta_1$ and $\beta_2$ would not be identified separately, even with this strong functional form assumption about $h$. Appendix \ref{sec:empiricalmore} reports results in which $R_i$ and $W_i$ are constructed by averaging responses of the main respondent and those of their spouse, which are similar.

\subsection{Semi-parametric estimates and marginal rates of substitution}
If $\mathbbm{E}[R_i|X_i,W_i]$ is not linear in fact in $X_i$ and $W_i$, then OLS estimates of Eq. \eqref{eq:ols2var} are not guaranteed to be interpretable in terms of causal effects, even if the assumptions of Theorem \ref{propflow} do hold. Appendix \ref{sec:regression} discusses this issue generally, and in this section I discuss the robustness of the ratio $\gamma_2/\gamma_1 = -2.614$ to relaxing this functional form assumption.

Column (3) of Table \ref{table:mainself} employs a semi-parametric estimator following \citet{robinsonestimator} that assumes the partially linear form $\mathbbm{E}[R_i|X_i,W_i] = f(X_{1i},X_{2i})+\lambda^T W_i$, in which $\lambda$ is estimated by residualizing $R_i$ and each component of $W_i$ with respect to $X_i$, before performing a bivariate kernel regression of $R_i - \hat{\lambda}^T W_i$ on $X_i$ to estimate $f(\cdot, \cdot)$. In this specification the coefficients $\gamma_j$ from Eq. \eqref{eq:ols2var} are replaced with
$$\gamma_j(x) := \partial_{x_j}\mathbbm{E}[R_{i}|X_i=x,W_i=w]$$
which does not depend on $w$ owing to the additively separable structure of $\mathbbm{E}[R_i|X_i,W_i]$. This implies that the estimand $\gamma_j(x)$ can be interpreted as proportional to an average causal effect $\tilde{\beta}_j(x)$ that also does not depend on $w$.

The first two rows of Column (3) report $\gamma_j(X_i)$ averaged across the empirical distribution of $X_i$. These estimates of $\mathbbm{E}[\gamma_j(X_i)]$ are numerically fairly similar to the $\hat{\gamma}_j$ reported by \citet{luttmer2005} from the OLS specification \eqref{eq:ols2var}. These average derivatives appear to mask only minor non-linearity in $\mathbbm{E}[R_i|X_i,W_i]$ with respect to $X_i$. Dividing the first two rows yields an estimate of $-0.202/0.122 \approx -1.66$ for $\frac{\mathbbm{E}[\gamma_2(X_i)]}{\mathbbm{E}[\gamma_1(X_i)]} = \frac{\mathbbm{E}[\tilde{\beta}_2(X_i)]}{\mathbbm{E}[\tilde{\beta}_1(X_i)]}$, but when computing the ratio of regression derivatives evaluated at $X_i$ for each observation, and then averaging across the sample, one instead obtains a value of $\hat{\mathbbm{E}}[\hat{\gamma}_2(X_i)/\hat{\gamma}_1(X_i)] = -1.581$. This average ratio is reported in the row labeled ``Ratio'' of Column (3) and estimates 
\begin{align} \label{eqexpratioavg2}
	\mathbbm{E}\left[\frac{\gamma_2(X_i)}{\gamma_1(X_i)}\right]=\int dF_{X}(x) \cdot \frac{ \partial_{x_2}\mathbbm{E}[R_{i}|X_i=x,W_i]}{ \partial_{x_1}\mathbbm{E}[R_{i}|X_i=x,W_i]}&= \mathbbm{E}\left[\frac{\tilde{\beta}_2(X_i)}{\tilde{\beta}_1(X_i)}\right]
\end{align}
If we assume a weakly separable causal model $h(x,u) = \mathtt{h}(g(x),u)$, then this value of $-1.581$ in turn represents an estimate of
$\mathbbm{E}\left[\frac{\partial_{x_1}g(X_i)}{\partial_{x_2}g(X_i)}\right]$, the overall population mean of the marginal rate of substitution between own income and neighbors' income, which is in this model common among all individuals sharing a value of $X_i$.\footnote{If we instead make the assumptions of Appendix Proposition \ref{propmrs}, then $\mathbbm{E}\left[\frac{\gamma_2(X_i)}{\gamma_1(X_i)}\right]$ is equal to
	$\int dF_{X}(x) \cdot  \mathbbm{E}\left[\left.\frac{\partial_{x_1}h(x,U_i)}{\partial_{x_2}h(x,U_i)}\right|H_i \in \tau_{V_i}, X_i=x\right]$, using as well that $\mathbbm{E}\left[\left.\frac{\partial_{x_1}h(x,U_i)}{\partial_{x_2}h(x,U_i)}\right|H_i \in \tau_{V_i}, X_i=x,W_i\right]$ does not depend on $W_i$ if $\mathbbm{E}[R_{i}|X_i=x,W_i=w]$ is separable between $x$ and $w$.} This value is substantially smaller than the $-2.614$ reported in column (1), which assumes linear conditional means.

Column (5) of Table \ref{table:mainself} shows the PUMA/own ratio to be similar when the controls $W_i$ are omitted, and fully non-parametric regression of $R_i$ on $X_{1i}$ and $X_{2i}$ becomes feasible. For comparison, column (4) implements OLS with no controls. Taking column (5) as our estimate of $\mathbbm{E}[\tilde{\beta}_2(X_i)/\tilde{\beta}_1(X_i)]$ would avoid the functional form restriction that $\mathbbm{E}[R_i|X_i, W_i]$ be linear in $W_i$, but at the expense of requiring Assumption EXOG to hold without the control variables $W_i$.\footnote{\label{fn:ses}I report only the point estimate for the sample mean of $\hat{\gamma}_2(X_i)/\hat{\gamma}_1(X_i)$ in columns (3) and (5) of Table \ref{table:mainself}, as a bootstrap computation of standard errors would be computationally intensive in the case of (3) given the number of control variables $W_i$. Standard errors are computed for the sample means of $\hat{\gamma}_2(X_i)$ and $\hat{\gamma}_1(X_i)$ separately in (3) and (5), but neither are clustered at the PUMA level as the \texttt{npregress} command in Stata does not accommodate cluster robust inference.} The gap in the PUMA/own ratio between linear and nonlinear models is much greater without controls, suggesting that the control variables eliminate much of the non-linearity with respect to $X_i$ in the conditional mean of $R_i$.

\subsection{Decomposing mean effects by response category}

While we know by Corollary \ref{propnabla} that the regression derivatives reported in Table \ref{table:mainself} average over respondents who are on the margin between two adjacent response categories, we also know from Theorem \ref{propflow} that we can isolate causal effects for respondents that are on a single such margin $r$ and $r+1$, for some $r \in \{1, 2, \dots 6\}$.

Table \ref{table:eachthresholdself} reports coefficients from a linear probability model that takes, for a given $r$, the conditional expectation function $\mathbbm{E}[\mathbbm{1}(R_i \le r)|X_i=x,W_i=w] = P(R_i \le r|X_i=x,W_i=w)$ to be linear in $X_i$ and $W_i$, with coefficients $(\gamma_{1r},\gamma_{2r}, \lambda_r)$ specific to that response category $r$, i.e. $\mathbbm{1}(R_i \le r) = \gamma_{1r} X_{1i}+\gamma_{2r} X_{2i} + \lambda_r^TW_i+\epsilon_{ri}$ with $\mathbbm{E}[\epsilon_{ri}|X_i,W_i]=0$.

\begin{table}[H]
	\centering
	{
\def\sym#1{\ifmmode^{#1}\else\(^{#1}\)\fi}
\begin{tabular}{l*{6}{c}}
\hline\hline
          &\multicolumn{1}{c}{(1)}&\multicolumn{1}{c}{(2)}&\multicolumn{1}{c}{(3)}&\multicolumn{1}{c}{(4)}&\multicolumn{1}{c}{(5)}&\multicolumn{1}{c}{(6)}\\
          &\multicolumn{1}{c}{R$\le$1}&\multicolumn{1}{c}{R$\le$2}&\multicolumn{1}{c}{R$\le$3}&\multicolumn{1}{c}{R$\le$4}&\multicolumn{1}{c}{R$\le$5}&\multicolumn{1}{c}{R$\le$6}\\
\hline
Own ln income&-0.000281         &  0.00205         &  0.00829         &   0.0211\sym{**} &   0.0405\sym{***}&   0.0159\sym{*}  \\
          &  (-0.14)         &   (0.73)         &   (1.96)         &   (2.95)         &   (4.67)         &   (2.05)         \\
PUMA ln income& -0.00798         &  -0.0151         &  -0.0242         &  -0.0448         &  -0.0891\sym{**} &  -0.0480         \\
          &  (-1.16)         &  (-1.62)         &  (-1.69)         &  (-1.93)         &  (-2.75)         &  (-1.65)         \\
\hline
Ratio PUMA/own&    28.44         &   -7.360         &   -2.920         &   -2.121         &   -2.198         &   -3.009         \\
se(ratio) &    201.4         &    10.82         &    2.211         &    1.330         &    0.896         &    2.288         \\
Sample size&     7939         &     7939         &     7939         &     7939         &     7939         &     7939         \\
\hline\hline
\multicolumn{7}{l}{\footnotesize \textit{t} statistics in parentheses}\\
\multicolumn{7}{l}{\footnotesize \sym{*} \(p<0.05\), \sym{**} \(p<0.01\), \sym{***} \(p<0.001\)}\\
\end{tabular}
}
\\ \vspace{.5cm}
	\caption{Coefficients from a linear probability model for each response category. All regressions include the controls from Table \ref{table:mainself} and standard errors clustered by PUMA.} \label{table:eachthresholdself}
\end{table}

Table \ref{table:eachthresholdself} reveals that the sign of $\gamma_{1r}$ is positive for all $r$ when it is statistically significant, the sign of $\gamma_{2r}$ is consistently negative when it is statistically significant, and the ratio $\gamma_{2r}/\gamma_{1r}$ is never differs from the ``aggregate'' value of -2.614 recovered by mean regression in a statistically significant way.

The information in Table \ref{table:eachthresholdself} is further visualized in Figure \ref{fig:acrossr_self}. The top-left panel depicts the coefficients $\gamma_{1r}$ versus $r$. By Theorem \ref{propflow}, we know that if the linear model correctly captures the conditional mean function of $R_i$ given $X_i$ and $W_i$, then each $\gamma_{1r}$ captures a positively weighted aggregation of the marginal effect of own income $\partial_{x_1} h(x,U_i)$ across individuals whose $U_i$ put them on their individual-specific threshold between response categories $r$ and $r+1$.

\begin{figure}[H]
	\begin{center} \vspace{.2cm}
		\includegraphics[height=1.75in]{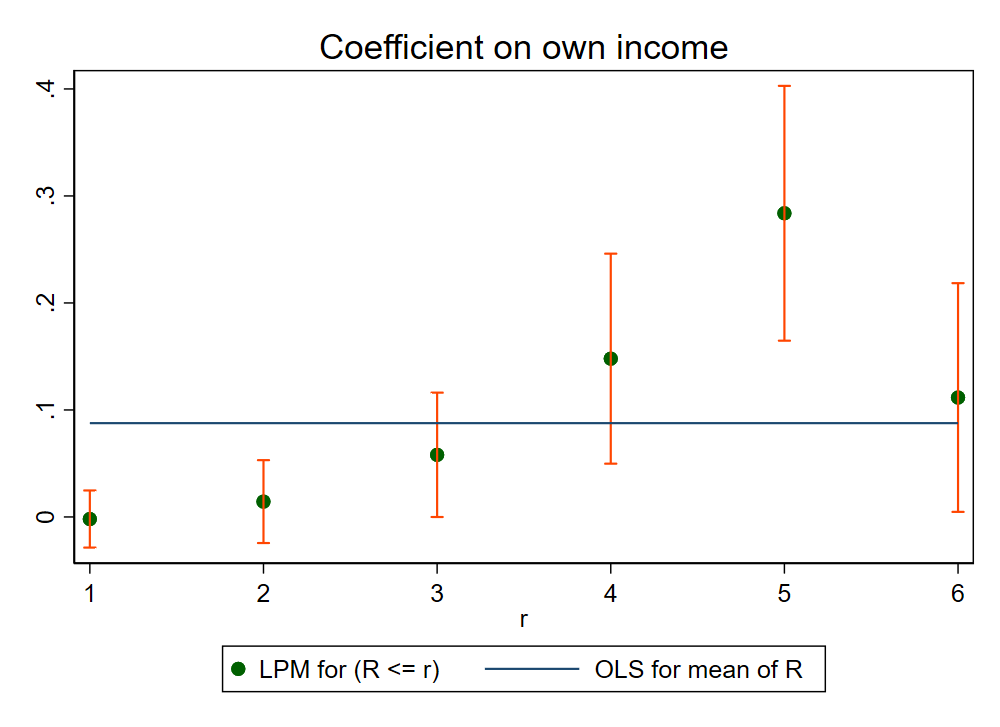}	\quad \quad \includegraphics[height=1.75in]{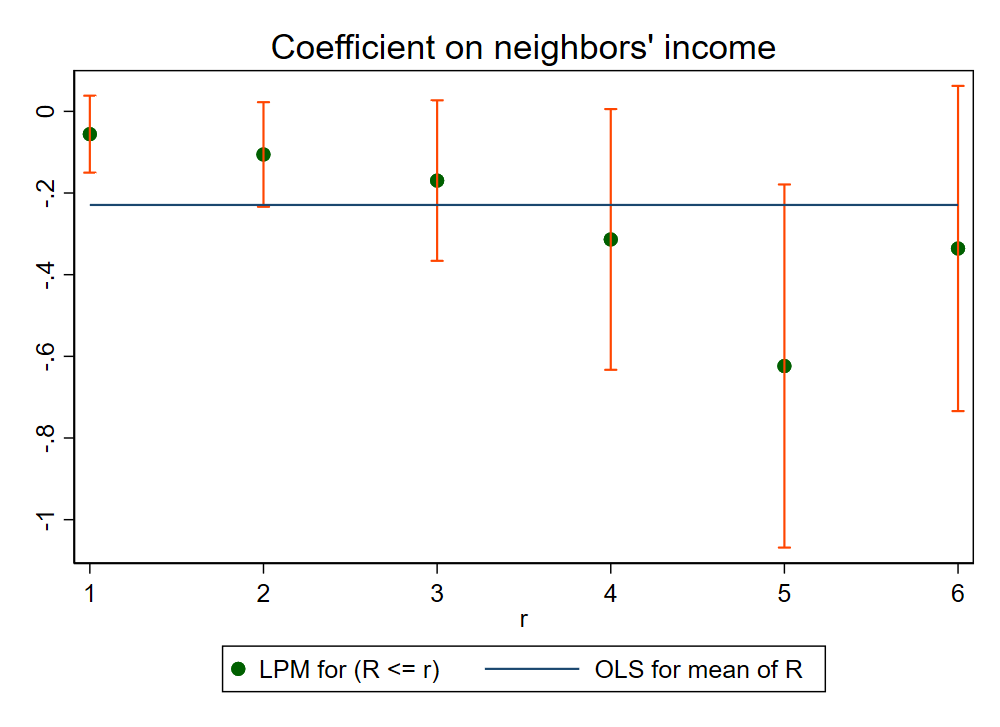} \\ \includegraphics[height=1.75in]{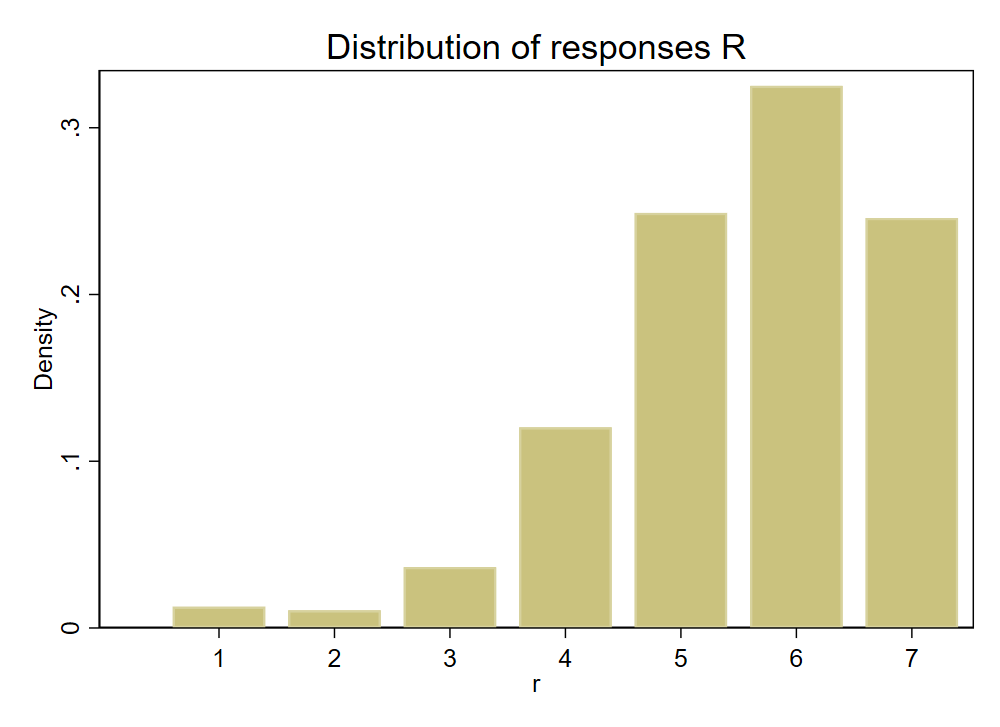}
		\quad \quad \includegraphics[height=1.75in]{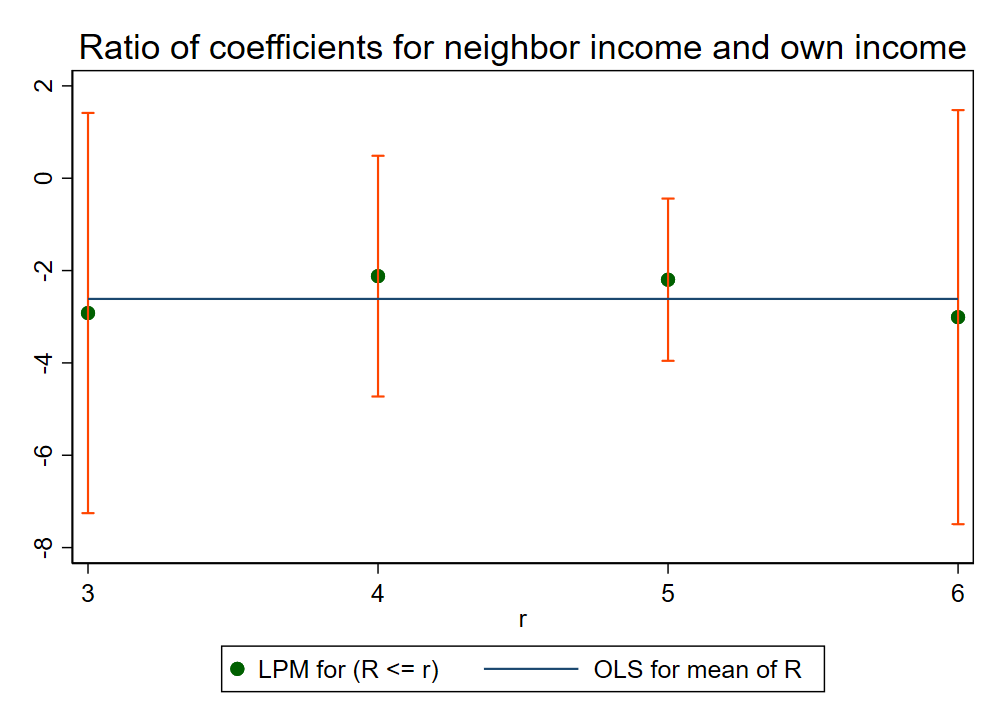}
		\caption{Visualization of the OLS estimates of $\gamma_{1r}$ (top-left), $\gamma_{2r}$ (top-right), the ratio $\gamma_{2r}/\gamma_{1r}$ (bottom-right) from the regression $\mathbbm{1}(R_i \le r) = \gamma_{1r} X_{1i}+\gamma_{2r} X_{2i} + \lambda_r^TW_i+\epsilon_{ri}$ for $r \in \{1, 2, \dots 6\}$, along with a histogram of the response categories $r\in \{1, \dots 7\}$. The horizontal line in the upper panels and bottom right panel depicts the corresponding value from mean regression (see Table \ref{table:mainself}).} \label{fig:acrossr_self}
	\end{center}
\end{figure}

The clear hump-shaped pattern across values of $r$ could be explained by heterogeneity in the mean causal effect among the individuals at each of the thresholds, or by differences in the density of individuals at that threshold. Suppose that this effect were a constant $\partial_{x_1} h(x,U_i) = \beta_1$ for all $i$. Then, Theorem \ref{propflow} shows that $\gamma_{1r}$ would be equal to $\beta_1 \cdot \mathbbm{E}[f_H(\tau_{V_i}(r)|X_i,V_i,W_i)]$ for each $r$.\footnote{This uses that since $\gamma_{1r}$ does not depend on $x$ or $w$, we must have that $\mathbbm{E}[f_H(\tau_{V_i}(r)|x,V_i,w)|W_i=w]=\mathbbm{E}[f_H(\tau_{V_i}(r)|x_i,V_i,W_i)]$ for all $x$ and $w$.} The quantity $\mathbbm{E}[f_H(\tau_{V_i}(r)|X_i,V_i,W_i)]$ is unobservable, but note that for any $r \in \{1, 2, \dots 6\}$ the observable probability $P(R_i=r) = P(R_i \le r)-P(R_i \le r-1)$ identifies the quantity
\begin{align*}
	&\mathbbm{E}[F_{H}(\tau_{V_i}(r)|X_i,V_i,W_i)-F_{H}(\tau_{V_i}(r-1)|X_i,V_i,W_i)]\\
	& \hspace{2.2in} \approx \mathbbm{E}\left[\{\tau_{V_i}(r)-\tau_{V_i}(r-1)\} \cdot f_H(\tau_{V_i}(r)|X_i,V_i,W_i) \right]
\end{align*}
where the approximation takes the density $f_H(\tau_{V_i}(r)|X_i,V_i,W_i)$ to be roughly constant on the interval $[\tau_{V_i}(r),\tau_{V_i}(r-1)]$. This will be a good approximation if that interval is small with high probability (i.e. in the limit of many categories), in which case $P(R_i=r)$ is roughly proportional to $\mathbbm{E}[f_H(\tau_{V_i}(r)|X_i,V_i,W_i)]$, if $\mathbbm{E}[\tau_{V_i}(r)-\tau_{V_i}(r-1)]$ does not vary much with $r$. The bottom-left panel of Figure \ref{fig:acrossr_self} depicts $P(R_i=r)$ and reveals that it does indeed capture the same basic pattern as $\gamma_{1r}$. 

Similarly, the negative values $\gamma_{2r}$ depicted in the top-right panel of Figure \ref{fig:acrossr_self} capture a positive aggregation of the marginal effect of PUMA income on $H$ with the weights $\mathbbm{E}[f_H(\tau_{V_i}(r)|X_i,V_i,W_i)]$. Again, the pattern of $\gamma_{2r}$ mirrors that of $P(R_i=r)$, which is consistent with a model in which this effect $\partial_{x_2} h(x,U_i)$ is captured by a single number $\beta_2$ for all $i$. Finally, we see in the bottom-right panel of Figure \ref{fig:acrossr_self} the observation made after Table \ref{table:eachthresholdself}, that the pattern cancels out and $\gamma_{2r}/\gamma_{1r}$ is roughly constant across $r$ (categories 1 and 2 are omitted due to being very imprecisely estimated). An F-test of equality of $\gamma_{2r}/\gamma_{1r}$ across all $r$ fails to reject (p-value: $0.98$).

Overall, the strong similarity in the shapes of the first three panels of Figure \ref{fig:acrossr_self} are suggestive that the differences in $\gamma_{1r}$ and $\gamma_{2r}$ are driven by the underlying latent density of happiness, than by heterogeneity in causal effects across the happiness distribution. This is consistent with a simple constant effects model in which $\gamma_{2r}/\gamma_{1r} = \beta_2/\beta_1$, or more generally by a weakly-separable model of the form $h(x,u) = \mathtt{h}(g(x),u)$.

\subsection{Who are the marginal respondents?}
Figure \ref{fig:attributes_self} compares the gender balance and education of respondents that on the margin between categories $r$ and $r+1$, for each $r$, with that of the population as a whole. These comparisons are based on Proposition \ref{prop:chars} in Appendix \ref{sec:characterizing}, which leverages additional assumptions to identify averages of an attribute $A_i$ among marginal respondents.

The upper panels of Figure \ref{fig:attributes_self} report estimates of $\mathbbm{E}[A_i|H_i = \tau_{V_i}(r)]$, under an assumption that $\{X_{i} \indep (A_i,U_i,V_i)\}|W_i$ and imposing the additional restriction that the sign of the effect of household income on happiness is the same for all units (not that this assumption is \textit{not} imposed for the main results). The implementation further takes the conditional expectation of $A_i \cdot \mathbbm{1}(R_i\le r)$ to be linear in $x$ and $w$ and assumes a linear probability model for $\mathbbm{1}(R_i\le r)$ (see Appendix \ref{sec:characterizing} for details).

In particular, the top left panel displays 95\% confidence intervals for $\mathbbm{E}[A_i|H_i = \tau_{V_i}(r)]$ versus $r \in \{2,3 \dots 6\}$ when $A_i$ is taken to be an indicator for the main respondent $i$ attending college.\footnote{Confidence intervals for $r=1$ are dropped in all panels of Figure \ref{fig:attributes_self} for visibility, as the standard error is much larger than for other $r$.} The horizontal line (orange) depicts the overall sample mean (an estimate of $\mathbbm{E}[A_i]$). For none of the margins $r$ can we reject the null hypothesis that the average rate of college among marginal respondents for that category is the same as the overall population mean. A similar result appears in the top-right panel, in which this calculation is repeated with $A_i$ equal to $i$'s years of education. There is some weak evidence that individuals on the margin on categories four and five out of seven have fewer years of education than the average. This is consistent with the finding of \citet{CPBL} that lower-education individuals are more likely to ``bunch'' at focal points in the response space $\mathcal{R}$, for example the midpoint (which is indeed 4 on the 1 to 7 scale).

\begin{figure}[H]
	\begin{center} \vspace{.2cm}	
		\includegraphics[height=1.75in]{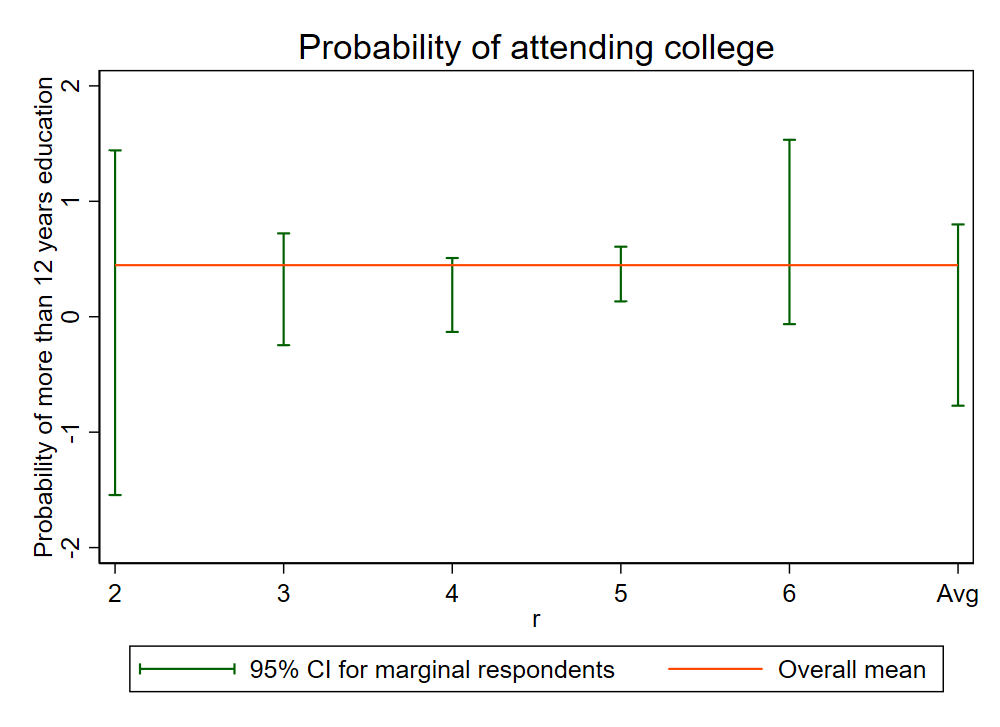} \quad \quad	
		\includegraphics[height=1.75in]{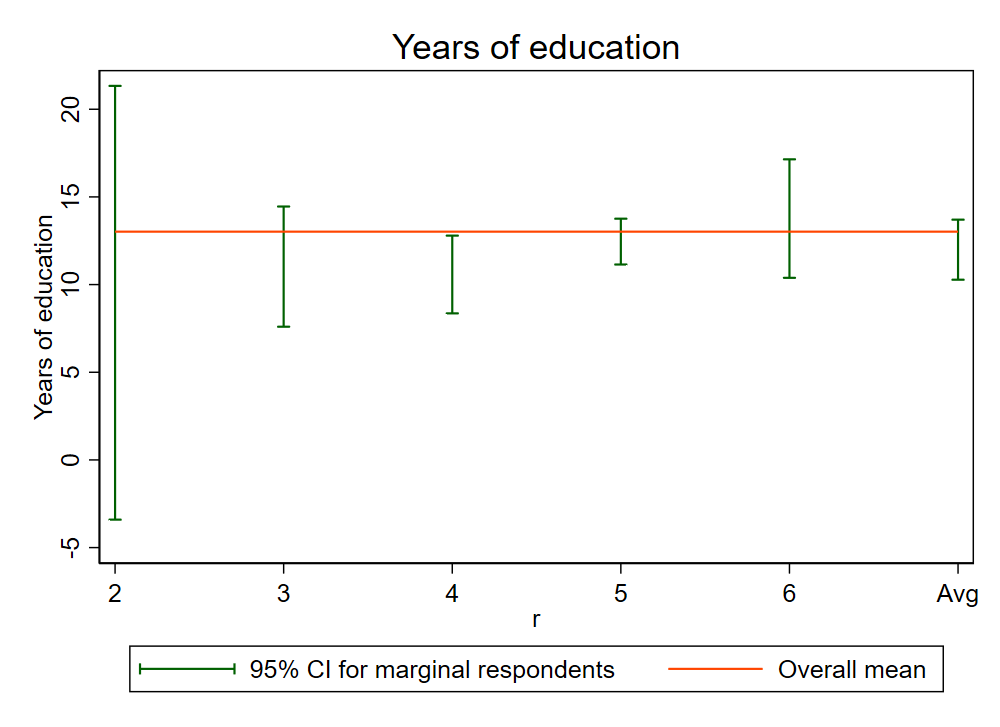}	
		\includegraphics[height=1.75in]{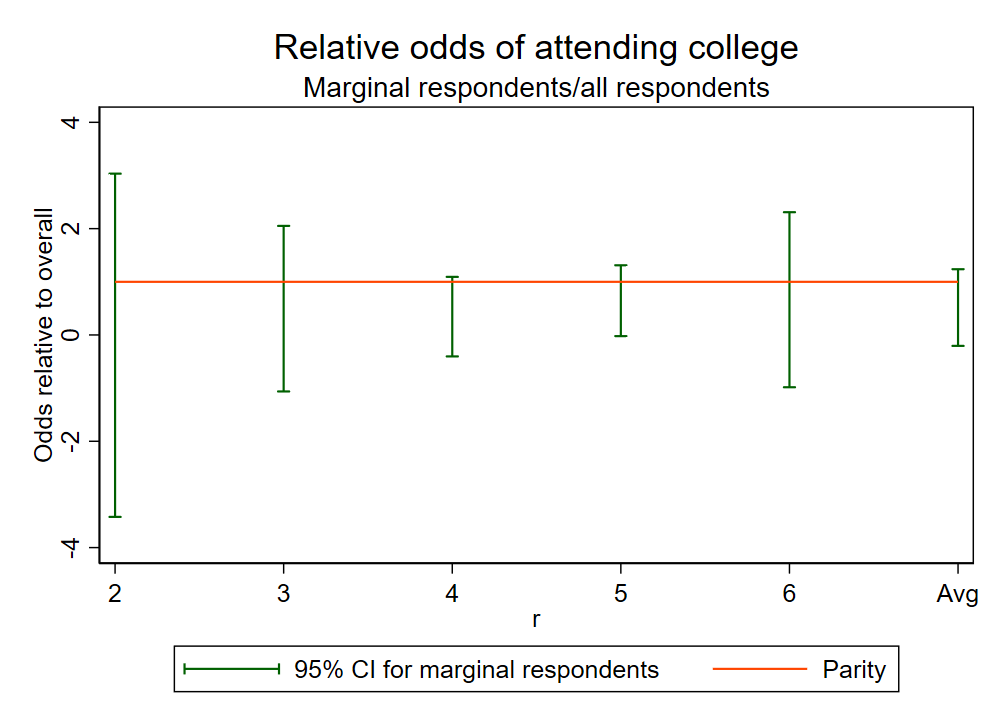} \quad \quad	\includegraphics[height=1.75in]{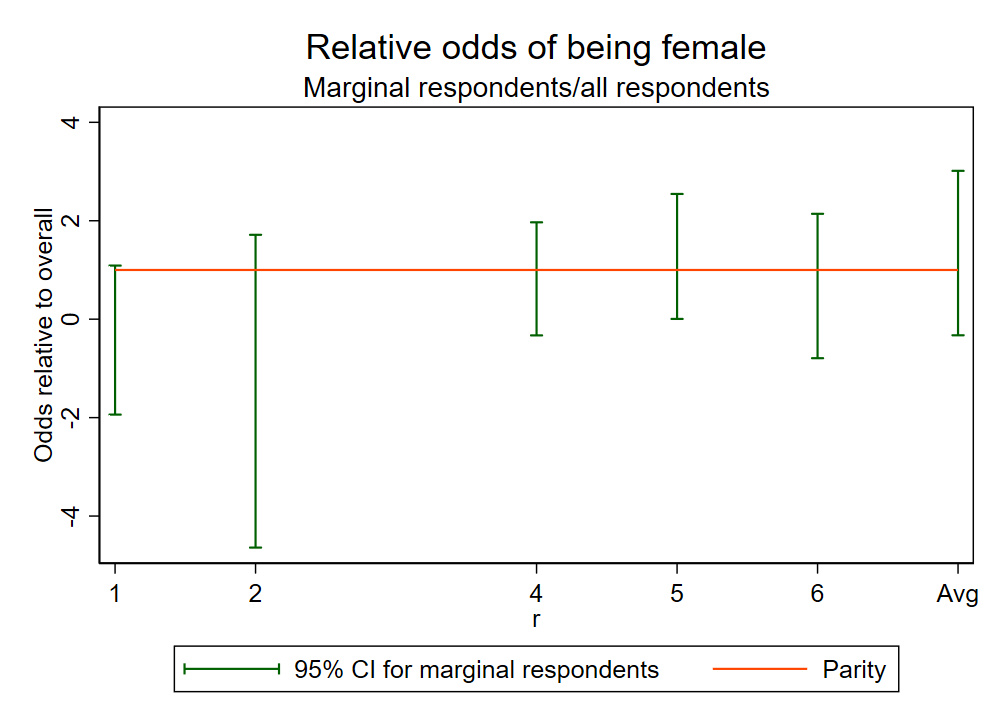}\\
		\caption{Attributes of marginal respondents, main respondent data. See text for details.} \label{fig:attributes_self}
	\end{center}
\end{figure}

The bottom panels of Figure \ref{fig:attributes_self} exploit the identification of the relative odds for a binary $A_i$, comparing marginal respondents to the population as a whole:
$$\frac{P(A_i=1|H_i = \tau_{V_i}(r),x,w)/P(A_i=0|H_i = \tau_{V_i}(r),x,w)}{P(A_i=1|x,w)/P(A_i=0|x,w)},$$
See Eq. \eqref{eq:relativeodds} in Appendix \ref{sec:characterizing}. This result makes use of the weaker assumption in Proposition \ref{prop:chars} that only assumes that a binary $A_i$ would represent a valid control variable to add to $W_i$. In this case all that is required is to implement regressions of $\mathbbm{1}(R_i \le r)$ on $x$ and $w$ separately by subsample defined by $A_i$.

The bottom panels report 95\% confidence intervals for this ratio of odds, with the horizontal line (orange) depicting unity (equal odds in both populations). The bottom-right panel sets $A_i$ to be an indicator for the main respondent being female, and compares the relative odds of being female among marginal respondents to the population overall. None are statistically different from unity. For clarity, the confidence interval for $r=3$, which is very large, is not shown. 

Overall, the results of this section indicate there is some weak evidence that marginal respondents have somewhat less education than the overall population, for the central category in the response space. No differences are detected across gender. The rightmost confidence interval in each panel of Figure \ref{fig:attributes_self}, labeled ``Avg'', replaces indicators for $R_i \le r$ with $R_i$ to approximate an ``average'' comparison considering all of the response categories at once. In all cases, we do not find any evidence that the marginal respondents overall differ from the infra-marginal respondents in education or gender.

\subsection{Summary of empirical results}
The results of the proceeding sections suggest that the results of the OLS regression implemented by \citet{luttmer2005} can indeed be interpreted as being informative about the causal effects of own and PUMA income on subjective well-being. If one is willing to assume that these two treatments are as-good-as-randomly assigned the sense that $\{(X_{1j},X_{2j}) \indep (U_i,V_i)\}|W_i$ (implying EXOG), then the coefficients $\gamma_1$ and $\gamma_2$ from Eq. \eqref{eq:ols2var} have causal interpretations under weak and fully non-parametric assumptions about the latent heterogeneity underlying causal effects and response functions. Using OLS does require that $\mathbbm{E}[R_i|X_i,W_i]$ is indeed linear in $X_i$ and $W_i$, but this caveat is \textit{not} a product of the ultimate outcome of interest $H_i$ being unobserved: rather, a correct specification of conditional mean functions is important for any selection-on-observables research design with control variables or setting in which multiple treatment variables are considered (see Appendix \ref{sec:regression} for details). Nonetheless, the quantitative estimates are similar if assuming a semi-parametric regression function that is partially linear in the controls, or a fully non-parametric regression that drops the controls altogether.

Although the magnitudes of $\gamma_1$ and $\gamma_2$ are not directly interpretable in terms of causal effects, the coefficient of $\gamma_2/\gamma_1$ is. This supports the main conclusion of \citet{luttmer2005} that relative-earnings considerations are indeed important to subjective well-being. A ratio of roughly -2 suggests that a 1\% income increase to one's neighbors' average income would require a roughly 2\% increase to one's own household income, in order to leave individuals equally happy overall. This interpretation in terms of a marginal rate of substitution is justified if one assumes a potential outcomes model that is weakly-separable in the treatments. The bottom-right panel of Figure \ref{fig:acrossr_self} finds that the ratio $\gamma_{2r}/\gamma_{1r}$ of category-specific regression coefficients is relatively constant over $r$, a key implication of a weakly-separable model. This supports the interpretation of $\gamma_2/\gamma_1$ as a marginal rate of substitution; another sufficient condition for this interpretation would be to assume utility to be quasi-linear in the log of own-income, as we saw in Section \ref{sec:quasilinear}.

Finally, I find that although the local regression derivatives of $R_i$ at a specific value of $X_i=x$ only capture causal effects among individuals that are indifferent between two response categories when $X_i=x$, these marginal respondents do not appear to be substantially different than infra-marginal respondents in terms of education or gender. 

\section{What is identified from discrete variation in $X$} \label{secdiscrete}
The analysis thus far has considered what is identified by examining how the conditional distribution of $R$ changes over infinitesimal differences in $X$. This section now considers taking discrete differences in treatment values (nesting the results thus far in the limit of small changes). I find that differences in the distribution of $R$ over discrete changes in $X$ can again be interpreted causally, and identify the sign of causal effects if those effects have the same sign across units. However, unlike the case with continuous treatments, magnitudes cannot be quantitatively compared between regressors absent further assumptions. Discrete treatment variables are prevalent in practice, so this highlights a limitation of, e.g. experiments with two treatment arms and subjective outcomes. 

Consider any two fixed values $x$ and $x'$, and define $\Delta_i := h(x',U_i)-h(x,U_i)$ to be the ``treatment effect'' of moving from $X_i=x$ to $X_i=x'$ for unit $i$. Further, let $f_H(y|\Delta,x,v,w)$ denote the density of $H_i$ conditional on $\Delta_i=\Delta$, $X_i=x$,$V_i=v$ and $W_i=w$. As before, let $P(R_{i} \le r|x,w)$ denote a shorthand for $P(R_{i} \le r|X_i=x,W_i=w)$. The following expression shows what is identified from the conditional distribution of $R_i$ across this discrete change between values $x$ and $x'$:
\begin{theorem} \label{propdiscrete}
	Under MONO and EXOG:
	$$P(R_{i} \le r|x',w)-P(R_{i} \le r|x,w) = -\mathbbm{E}[\bar{f}_H(\tau_{V_i}(r)|\Delta_i,x,V_i,w)\cdot \Delta_i|W_i=w]$$
	where $\bar{f}_H(y|\Delta,x,v,w):= \frac{1}{\Delta} \int_{y-\Delta}^{y} f_H(h|\Delta,x,v,w) \cdot dh$ is the average density between $y-\Delta$ and $y$, among units with reporting function $v$, treatment effect $\Delta$, and $(X_i,W_i)=(x,w)$.\footnote{By ``between $y-\Delta$ and $y$'' I mean in the interval $[\min\{y-\Delta,y\}, \max\{y-\Delta,y\}]$, regardless of the sign of $\Delta$. Note that $\bar{f}_H(y|\Delta,x,v)$ is positive even if $\Delta < 0$, in which case it is equal to the average density between $y$ and $y+|\Delta_i|$.}
\end{theorem}
\noindent Similar to Theorem \ref{propflow}, Theorem \ref{propdiscrete} shows that the change in $P(R_i \le r|W_i=w,X_i=x)$ over discrete changes in $x$ can be written as a positive linear combination of the causal effect of that variation in $X$ on $H$---a quantity proportional to a parameter of the form $\tilde{\Delta}$ introduced in Section \ref{sec:params}. The quantity $\bar{f}_H(\tau_{V_i}(r)|\Delta_i,x,V_i,w)$ is positive for each $i$ but unknown to the researcher, determined in part by individuals' reporting functions and the underlying distribution of $H_i$. Intuitively, respondents with treatment effect value $\Delta$ are ``counted'' in the above average if there exists a positive mass of such individuals with $(X_i,W_i)=(x,w)$ and happiness $H_i$ in the range $\tau_{V_i}(r)-\Delta$ to $\tau_{V_i}(r)$. Note that Theorem \ref{propdiscrete} exhausts all implications of the observable data $(R_i,X_i)$ regarding variation in the potential outcome functions $h(x,u)$ with respect to $x$ (for a fixed value of the controls $W_i$).\footnote{Given any such fixed $w$, once $P(R_{i} \le r|X_i=x,W_i=w)$ is known for all $r$ for some fixed reference value $x$ of the explanatory variables, along with the distribution of $X_i|W_i=w$, the only remaining information available from the data takes the form of differences $P(R_{i} \le r|x',w)-P(R_{i} \le r|x,w)$ for various values of $x'$ and $r$.}\\

\noindent \textit{Theorem \ref{propflow} as a limiting case of Theorem \ref{propdiscrete}:} A similar expression to that of Theorem \ref{propdiscrete} shows up in the ``bunching design", which leverages bunching at kinks in decision-makers' choice sets for identification of behavioral elasticities. Since the kink compares just two distinct slopes, an identification problem emerges for elasticity parameters \citep{blomquist_bunching_2019}. An assumption sometimes used sidestep this issue is that the kink is ``small'' (e.g. \citealt{saez2011,kleven2016}, see \citealt{goff2022} for a discussion). An analogous assumption in the context of Theorem \ref{propdiscrete} would be that $\Delta_i$ is small with probably one so that for each $\Delta \in supp\{\Delta_i\}$, the density $f_H(h|\Delta,x,v,w)$ is approximately constant for all $h$ between $\tau_v(r)-\Delta$ and $\tau_v(r)$. 
Under this assumption, Theorem \ref{propdiscrete} would simplify to:
\begin{align}
	P&(R_{i}  \le r|x',w)-P(R_{i} \le r|x,w) \nonumber =-\int dF_{V|W}(v|w) \cdot f_H(\tau_{V_i}(r)|\Delta_i,x,V_i,w)\cdot \Delta_i \\
	&=-\int dF_{V|W}(v|w) \cdot f_H(\tau_v(r)|x,v,w)\cdot  \mathbbm{E}[\Delta_i|H_i=\tau_v(r), X_i=x, V_i=v,W_i=w]
	\label{eq:discreteconst}
\end{align}
Eq. (\ref{eq:discreteconst}) exactly recovers the weighting over individuals achieved by Theorem \ref{propflow} using continuous variation in $x$. In particular, the quantity $\mathbbm{E}[\Delta_i|\tau_v(r), x, v,w]$ appears above with the same weight $-dF_{V|W}(v|w)\cdot f_H(\tau_v(r)|x,v,w)$ as $\mathbbm{E}[\partial_{x_j} h(x,U_i)|\tau_v(r),x,v,w]$ does in Eq. (\ref{eq:propflow}). Unfortunately, the constant density assumption used to obtain \eqref{eq:discreteconst} is quite hard to justify \textit{except} in the limit that $\Delta_i$ is very small with probability one.\footnote{\label{fn:lemmasmall}If we consider the limit $x' \rightarrow x$ with the two differing only in component $j$, this approximation becomes exact and Eq. (\ref{eq:discreteconst}) applied to $(P(R_{i}  \le r|x',w)-P(R_{i} \le r|x,w))/(x_j'-x_j)$ reduces to Theorem \ref{propflow}. See Lemma SMALL in \citet{goff2022}.} Section \ref{sec:comparing} thus explores this issue further when $\Delta_i$ is not small, in the context of mean regression.\\

\noindent \textit{Intuition for Theorem \ref{propdiscrete}}: We can obtain some intuition for Theorem \ref{propdiscrete} as depicted in Figure \ref{figdiscrete}. Suppose there are two response categories $\mathcal{R} = \{0,1\}$ with a common reporting function $r(h) = \mathbbm{1}(h \ge \tau)$. By iterating expectations over $\Delta_i$, we can consider a single value $\Delta$ of $\Delta_i$ at a time. Thus we aim to show that $\mathbbm{E}[R_{i}|x',\Delta]-\mathbbm{E}[R_{i}|x,\Delta] = \bar{f}_H(\tau|x,\Delta)\cdot \Delta$, using that $P(R_{i} \le 0|x',\Delta)=1-\mathbbm{E}[R_{i}|x',\Delta]$. In Figure \ref{figdiscrete}, I make the conditioning on $\Delta_i=\Delta$ implicit to simplify notation, taking an example in which $X_i$ is an indicator for marriage with $x'=1$, $x=0$.\\

\begin{figure}[h!]
	\begin{center} \vspace{.2cm}
		\includegraphics[height=2in]{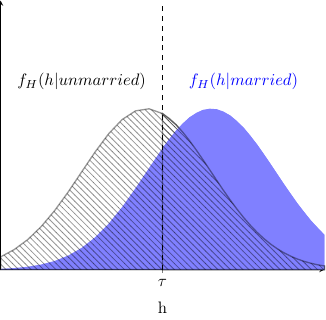} \quad	\quad \quad	\includegraphics[height=2in]{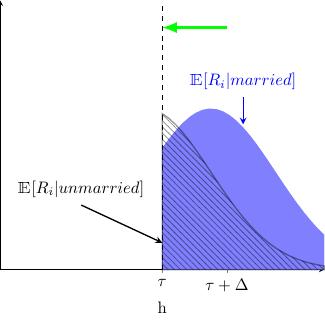} \quad 
		\includegraphics[height=2in]{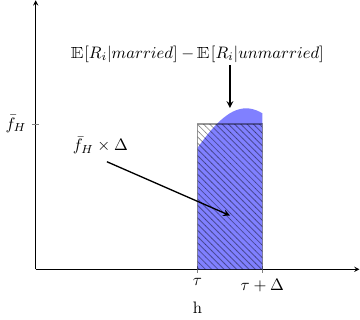}
		\caption{Visualization of Theorem \ref{propdiscrete}. Conditional on $\Delta_i:=h(married,U_i)-h(unmarried,U_i) = \Delta$, $f_H(\cdot|married) = f_{h(married,U)}(\cdot)$ is a rightward shift of $f_H(\cdot|unmarried) = f_{h(unmarried,U)}(\cdot)$, by $\Delta$. Thus $\mathbbm{E}[R_{i}|married,\Delta]-\mathbbm{E}[R_{i}|unmarried,\Delta]$ is the area under $f_H(\cdot|married)$ between $\tau$ and $\tau+\Delta$, which is in turn equal to a rectangle of width $\Delta$ and height $\bar{f}_H$, where $\bar{f}_H$ is the average of $f_H(\cdot|married)$ across this interval. \label{figdiscrete}} 
	\end{center}
\end{figure}

\noindent \textit{Mean regression: } As our main focus is regressions capturing the conditional mean of $R_i$ with $\mathcal{R}$ an integer response scale, let us as in Eq. (\ref{eq:expflow}) aggregate Theorem \ref{propdiscrete} across the response categories $r$ to obtain:\footnote{To obtain the notation of Eq. (\ref{eq:discreteintro}) in the introduction from \eqref{eqdiscreteintermediate}, define $\bar{f}_H(\Delta,x,v,w) := \sum_{r=0}^{\bar{R}-1} \bar{f}_H(\tau_v(r)|\Delta,x,v,w)$.}
\begin{align}
	&\mathbbm{E}[R_i|x',w]-\mathbbm{E}[R_i|x,w]=\mathbbm{E}\left[\left. \sum_{r=0}^{\bar{R}-1} \bar{f}_H(\tau_v(r)|\Delta,x,v,w) \cdot \Delta_i\right|W_i=w\right] \label{eqdiscreteintermediate}
\end{align}
Recall from Theorem \ref{propflow} that derivatives of the conditional distribution of $R$ yield causal effects $\nabla_x h(x,U_i)$ with weights proportional to $\sum_r f_H(\tau_v(r)|x,v)$. By contrast, \eqref{eqdiscreteintermediate} shows that discrete differences in $X$ recover treatment effects $\Delta_i = h(x', U_i)-h(x,U_i)$ with ``weights'' that themselves depend upon $\Delta_i$ through $\sum_r \bar{f}_H(\Delta_i,x,v,w)$. Since this quantity depends not only on the density of $H$ at response thresholds $\tau_v(r)$ but also the density at points within $\Delta$ of such thresholds through $\bar{f}$, the two weighting schemes do not lead to estimands that can obviously be directly compared.\\

\noindent \textit{Note:} whether or not $\mathbbm{E}[R_i|x',w]-\mathbbm{E}[R_i|x,w]$ is positive or negative does \textit{not} reflect the sign of the average treatment effect: $\mathbbm{E}[\Delta_i]$. Rather, it depends on how positive and negative treatment effects are aggregated over by the weights $\sum_r \bar{f}_H(\tau_v(r)|\Delta,x,v,w)$. If the CDF functions (or equivalently, quantile functions) of $h(x,U_i)$ and $h(x',U_i)$ cross, then there must be some individuals with $\Delta_i<0$ while others with $\Delta_i>0$.\footnote{Specifically, then $P(\Delta_i < 0) \ge \sup_t \left\{F_{h(x',U_i)}(t)-F_{h(x,U_i)}(t)\right\}$ and $P(\Delta_i > 0) \ge \sup_t \left\{F_{h(x,U_i)}(t)-F_{h(x',U_i)}(t)\right\}$; see e.g. \citet{fanpark}}. This connects Theorem \ref{propdiscrete} to the result of \citet{bondandlang}, discussed further in Appendix \ref{sec:bl}. 

\subsection{Comparing discrete and continuous regressors} \label{sec:comparing}

Given the foregoing analysis, Theorems \ref{propflow} and \ref{propdiscrete} together imply that regression coefficients between discrete and continuous treatment variables can be meaningfully compared quantitatively in terms of causal effects in the limit that effects $\Delta_i$ for the discrete treatment are very small, if the conditional mean function is indeed linear.

More generally, a researcher who is interested in comparing a local regression derivative to the mean difference across two discrete groups can construct ratios like:
\begin{equation} \label{eq:genratio}
	\frac{\mathbbm{E}[R_i|X_i=x',W_i=w]-\mathbbm{E}[R_i|X_i=x,W_i=w]}{\partial_{x_1}\mathbbm{E}[R_i|X_i=x'',W_i=w]}
\end{equation}
for some $x$,$x'$, and $x''$. For example, if $X = (income, marriage)$ with $x'=(y,married)$ and $x=(y,unmarried)$ for any income $y$ and $x''=(y,m)$ for $m \in \{married, unmarried\}$, then Eq. \eqref{eq:genratio} would yield a comparison of regression contrasts involving income to those involving marriage. If $\mathbbm{E}[R_i|X_i,W_i]$ were fully linear, then the numerator of \eqref{eq:genratio} would not depend on $y$ or $w$ and the denominator would not depend on $y$, $m$ or $w$, yielding a ratio of two linear regression coefficients.

Out goal now is to examine the causal interpretation of Eq. \eqref{eq:genratio} outside of the limit that $\Delta_i$ is very small. Combining Eq. (\ref{eqdiscreteintermediate}) with Corollary \ref{propnabla}, we know that the ratio in Eq. \eqref{eq:genratio} is equal to
\begin{equation} \label{eq:ratiogeneral}
	\frac{\mathbbm{E}\left[\color{purple} \sum_r\bar{f}_H(\tau_{V_i}(r)|\Delta_i,x,V_i,w) \color{black} \cdot \Delta_i|W_i=w\right]}{\mathbbm{E}\left\{\left.\color{purple}\sum_{r} f_H(\tau_{V_i}(r)|x'',V_i,w) \color{black} \cdot \mathbbm{E}\left[\partial_{x_1} h(x'',U_i)|H_i=\tau_{V_i}(r),x'',V_i,w\right]\right|W_i=w\right\}}
\end{equation}
To interpret this as informative about the relative magnitudes of $\Delta_i$ and $\partial_{x_1}h(x'',U_i)$, the relevant question is how similar the sum $\sum_{r} f_H(\tau_{V_i}(r)|x'',V_i,w)$ over densities at the thresholds is to the corresponding sum over mean densities: $\sum_r \bar{f}_H(\tau_{V_i}(r)|\Delta_i,x,V_i,w)$, at least on average. If these quantities tend to be close to one another in magnitude, then Eq. (\ref{eq:genratio}) uncovers something close to the ratio of two convex averages of causal effects. If they differ by an unknown amount, then interpreting (\ref{eq:genratio}) in terms of the relative magnitudes of causal effects is not possible.

Reasoning about the magnitudes involved in \eqref{eq:ratiogeneral} is challenging in full generality, but it is possible to derive analytical results to guide our intuition by assuming that there are ``many'' response categories in $\mathcal{R}$. Given the definition of $\bar{f}$, notice that $\sum_{r} f_H(\tau_v(r)|x'',v,w)$ and $\sum_r \bar{f}_H(\tau_v(r)|\Delta,x,v,w)$ are similar for a given $(\Delta,v,w)$ if
\begin{equation} \label{eq:desiredapprox}
	\sum_{r} \frac{1}{\Delta}\int_{\tau_v(r)-\Delta}^{\tau_v(r)} f_H(y|\Delta,x,v,w)dy \approx \sum_r f_H(\tau_v(r)|x'',v,w)
\end{equation}
Observe that the two sides of (\ref{eq:desiredapprox}) can \textit{only} differ because the summation occurs over $H_i$ evaluated at the discrete thresholds $\tau_v(r)$. If instead the sums over $r$ were replaced by integrals over all possible values of $H_i$, we would have $\int \left\{\frac{1}{\Delta} \int_{h-\Delta}^{h} f_H(y|\Delta,x,v,w)dy\right\}dh = \int f_H(h|x'',v,w)\cdot dh$, which holds trivially because both sides evaluate to unity for any $\Delta,v,x,w$ and $x''$.\footnote{This is immediate for the RHS, which integrates a density. To see it for the LHS, reverse the order of integrals to obtain $\int dy \cdot f_H(y|\Delta, x,v,w) \left\{\frac{1}{\Delta}\int_{y}^{y+\Delta} dh\right\} = 1$.} Thus it would seem that we have a second ``limit'' in which discrete and continuous regression differences can be compared: when there are many response categories. However, I show in Appendix \ref{seccontinuousreporting} that discrete sums over the thresholds do not exactly correspond to equal-weighted integrals over $h$ in the limit of a continuum of response categories. Rather, in this limit the integrals also involve the quantity $r'(h,v)$, which measures how responsive response function $v$ is at $h$. Nevertheless, the intuition provided by the above logic suggests that looking at the limit of many categories may provide a tractable means of evaluating the quality of Eq. (\ref{eq:desiredapprox}) as an approximation.

\subsection{A tractable approximation in the limit of many response categories} \label{sec:tractable}
In Appendix \ref{seccontinuousreporting} I define a formal notion of the response categories being ``dense'' in the space of latent $H_i$, for each reporting function type $v$. This \textit{dense response limit} allows us to conceptualize there as being an infinite number of response categories, while remaining contained between $0$ and a fixed $\bar{R}$. The dense response limit delivers a tractable approximation which may be reasonable to apply in instances in which the survey question offers many response categories between a lower and upper limit (e.g. integers from 0 to 100).

Proposition \ref{propheterolinear} in Appendix \ref{seccontinuousreporting} shows how bounds on the ratio of total weights in Equation \eqref{eq:ratiogeneral} can be obtained in the dense response limit when each individual spaces out the thresholds $\tau_v(r)$ at roughly equal intervals---yielding reporting functions that are individually piecewise-linear. Intuitively, the assumption of linear reporting eliminates the effect of $r'(h,v)$, but only within the range of $h$ upon which each individuals' reporting function is increasing. Proposition \ref{propheterolinear} gives two sets of bounds. First, a more general bound suggests that discrete contrasts will tend to \textit{overstate} causal effects relative to regression derivatives, by a factor that is upper bounded by two. A second bound further assumes that the ``sensitivity'' of individual reporting functions is not too heterogeneous, and suggests that the inflation factor can also be bounded by the reciprocal of the fraction of the population that do not bunch at the endpoints of the response scale. This bound is close to unity when there are few such bunchers, which can be verified empirically.

To assess the performance of the theoretical bounds described above, Appendix \ref{sec:sim} simulates several data-generating-processes (DGPs) for $H_i$ and for the response functions $r(\cdot, V_i)$. The simulations generally provide an optimistic picture that the weights have similar overall magnitude in the numerator and denominator of Equation \eqref{eq:ratiogeneral}, across a wide variety of DGPs. Thus $\left\{\mathbbm{E}[R_i|x',w]-\mathbbm{E}[R_i|x,w]\right\}/\partial_{x_j}\mathbbm{E}[R_i|x,w]$ can be interpreted as close to a ratio of weighted averages of causal effects in those DGPs considered. In general, results do not seem to differ substantially whether the number of response categories is small, or whether there are few or many different reporting functions present in the population. When treatment effects become very \textit{large} relative to the dispersion of happiness in the population, non-linearity in the density of the conditional distribution of happiness becomes important and the sense in which comparisons of magnitude can become misleading is apparent in the simulations.

Appendix \ref{sec:regression} investigates the implications of Proposition \ref{propheterolinear} for practical regression analysis with mixed discrete and continuous regressors, focusing both on the common approach of linear regression and non-parametric alternatives. 

\section{Conclusion} \label{sec:conc}
This paper has investigated the identification of causal effects when using subjective responses as an outcome variable. Such reports typically ask individuals to choose a response from an ordered set of categories, and how individuals use those categories can be expected to differ by individual $i$. Nevertheless, researchers may be willing to suppose that individual responses reflect the value of a well-defined latent variable $H_i$.

Without observing $H_i$ and without assuming it is possible to rank individuals by $H_i$ on the basis of their responses $R_i$, we have seen that the conditional distribution of $R_i$ given exogenous covariates $X_i$ can still be informative about the effects of $X$ on $H$. While this allows one to observe the sign of causal effects under the assumption that this sign is common across individuals, we've seen that different discrete conditional mean comparisons can impose different total weightings over the causal effects of individuals in the population. Simulation evidence as well as theoretical results suggest the impact of this problem for comparisons of magnitude is somewhat limited in practice, and the problem goes away entirely in the limit of continuous treatment variables. Nevertheless, the results suggest that care is warranted in comparing the magnitude of regression coefficients across explanatory variables, even when they are as good as randomly assigned.

The results of this paper suggest three practical implications for using regression analysis for causal inference with subjective ordinal outcomes. First, the critique of \citet{bondandlang} that such responses are only ordinarily meaningful is most acute when comparing large, heterogeneous populations that differ among many dimensions. Isolating causal effects using exogenous variation in individual treatment variables is not subject to this critique in the sense that regression derivatives identify the sign of a convex average of causal effects, even though reporting functions are unknown to the researcher. Second, to make such regression derivatives quantitatively meaningful, researchers should focus on \textit{comparing} across treatment variables when more than one is available. While non-parametric regression methods are preferred from the standpoint of identification, this is not specific to the analysis of subjective outcome variables. Finally, notwithstanding the above, researchers should exercise some caution when comparing the magnitudes of two discrete treatment effects or between a discrete treatment effect and the slope for a continuous treatment. The relative magnitudes of convex averages of causal effects can still be partially identified in such settings with further assumptions, though weakening these assumptions represents a possible avenue for future research.

\normalsize
\printbibliography
\large
\begin{appendices}

\section{Reconciling my results with \citet{bondandlang}} \label{sec:bl}

For fixed treatment values $x$ and $x'$, define the parameter
$$\theta:=\mathbbm{E}[H_i|X_i=x']-\mathbbm{E}[H_i|X_i=x]$$
Drawing on results from \citet{manskitamer}, \citet{bondandlang} show that even if reporting functions are homogeneous, the sign of $\theta$ is not identified from the distribution of $(R_i,X_i)$ absent strong assumptions. Influentially, they argue that regressions of $R_i$ on $X_i$ are therefore generally uninformative about how the mean of $H_i$ varies across subgroups of the population.

\subsection{The identification problem for unweighted mean comparisons}
One way to see the problem highlighted by \citet{bondandlang} quite clearly is to rewrite $\theta$ using the identity $\mathbbm{E}[A] = \int_0^1 Q_A(u) \cdot du$ for any random variable $A$:
\begin{equation} \label{eq:QTEidentity}
\theta = \int_0^1 \left\{Q_{H|X=x'}(u)-Q_{H|X=x}(u)\right\}\cdot du
\end{equation}
where $Q_{H|X}$ is the conditional quantile function of $H_i$ given $X_i$. Meanwhile, the difference in the mean of $R_i$ between $x'$ and $x$, with a common reporting function, instead identify 
\begin{equation} \label{QTEexpansion}
\mathbbm{E}[R|X=x']-\mathbbm{E}[R_i|X_i=x] = \int_{0}^1 \bar{r}'_{x',x}(u) \cdot \left\{Q_{H|X=x'}(u)-Q_{H|X=x}(u)\right\} \cdot du,
\end{equation}
where $\bar{r}'_{x',x}(u):=\frac{r(Q_{H|X=x'}(u))-r(Q_{H|X=x}(u))}{Q_{H|X=x'}(u)-Q_{H|X=x}(u)}$ is the ``average rate of change'' in the common reporting function $r(\cdot)$ between $Q_{H|X=x}(u)$ and $Q_{H|X=x'}(u)$.\footnote{Eq. \eqref{QTEexpansion} is a special case of Proposition \ref{propidr2} from Appendix \ref{sec:idio}) in which there is a single reporting function, and no controls (note that in the case of homogeneous reporting functions $r(c\dot)$, Assumption IDR used in Proposition \ref{propidr2} holds trivially). Eq. \eqref{QTEexpansion} only assumes MONO for the common reporting function.}

Eq. \eqref{QTEexpansion} thus represents a re-weighting of the quantile differences $Q_{H|X=x'}(u)-Q_{H|X=x}(u)$ that appear in \eqref{QTEexpansion} with uniform weight under the integral over all $u \in [0,1]$ in Eq. \eqref{eq:QTEidentity}. The quantity $\bar{r}'_{x',x}(u)$ is weakly positive for any $(x,x',u)$ (since $r$ is weakly increasing). However these weights will not uniform because $r(\cdot)$ cannot be a linear function except in the limit of a continuum of response categories.  Instead, it will exhibit discrete jumps or falls at the $u$ for which $Q_{H|X=x'}(u)$ and $Q_{H|X=x}(u)$ lie on opposite sides of a response threshold $\tau(r)$. Where exactly the weight $\bar{r}'_{x',x}(u)$ is smaller or larger depends on the distribution of latent happiness and the spacing of the response thresholds $\tau(r)$, which are both unknown.

One special case in which the sign of $\mathbbm{E}[R|X=x']-\mathbbm{E}[R_i|X_i=x]$ does identify the sign of $\theta$ is when the conditional distribution $H|X=x'$ stochastically dominates the conditional distribution $H|X=x$ (or vice versa).\footnote{By saying that $A|X=x'$ (first order) stochastically dominates $A|X=x$, I mean that $P(A_i \le a|X_i=x') \le P(A_i \le a|X_i=x')$ for all values $a$.} In this case the quantile functions never cross and the sign of $Q_{H|X=x'}(u) - Q_{H|X=x}(u)$ is positive for all $u \in [0,1]$, implying that $\theta$ and $\mathbbm{E}[R_i|X_i=x']-\mathbbm{E}[R_i|X_i=x]$ will both be positive. If instead $Q_{H|X=x'}(u) < Q_{H|X=x}(u)$ for some $u$, while $Q_{H|X=x'}(u) > Q_{H|X=x}(u)$ for other $u$, then it will generally be possible to reverse the ordering of $\mathbbm{E}[R_i|X_i=x']$ and $\mathbbm{E}[R_i|X_i=x]$for a given $\theta$ depending on where  the unknown function $r(\cdot)$ increases the fastest (see \citet{ShroderYitzhaki2017} for a version of this argument). As \citet{bondandlang} note, even if the observable distribution $R_i|X_i=x'$ stochastically dominates $R_i|X_i=x$, this is not sufficient to conclude that $H|X=x'$ stochastically dominates $H|X=x$. 

Since the observable data is not dispositive on its own, one can of course proceed by making assumptions to identify the sign of $\theta$. Suppose that $X_i=x'$ indicates the $i$ is a resident of the United States and $X_i=x$ that $i$ is a resident of Japan. If one is willing to assume that the higher-mean country has a higher happiness at every quantile level $u \in [0,1]$---whichever country that is---then the sign of $\theta$ is identified. But since life differs in many ways between the US and Japan which may matter in different ways for different individuals, it is hard to make this argument compellingly. Indeed, the underidentification problem for the sign of $\theta$ is most acute when comparing means of $R_i$ between two distinct populations that differ from one another along multiple dimensions, and each of which is quite heterogeneous on its own.

\subsection{Convex averages of causal effects are still identified}

The above problem appears in a much less pronounced way when $X_i$ represents a vector of randomized treatments, and $x$ and $x'$ differ by just one component, as in Theorems \ref{propflow} and \ref{propdiscrete} of this paper. In particular, if the treatment effect $\Delta_i = h_i(x')-h_i(x)$ has the same sign for all units $i$, then $H_i|X_i=x'$ necessarily stochastically dominates $H_i|X_i=x'$. As an example, consider a linear potential outcomes model in which $h_i(x) = h_i(x,U_i) = x^T\beta + U_i$. The treatment effect $\Delta_i$ is then $\Delta := (x'-x)^T\beta$, the same for all $i$. Given randomization $U_i \indep X_i$, $Q_{H|X=x}(u) = Q_{h(X)}(u)=x^T\beta+Q_U(u)$ and the quantile difference $Q_{H|X=x'}(u)-Q_{H|X=x}(u)=\Delta$, numerically the same for all $u \in [0,1]$.\footnote{In fact without loss of generality we can normalize $U_i$ to be uniform on $[0,1]$, and $h(x,u) = Q_{H|X=x}(u)$. To see this, suppose instead that $U_i$ has CDF $F_U$, but given randomization we have that $U_i \indep X_i$. Note that with probability one, $h_i(x) = Q_{h(x)|X_i}(T_i)$ where $T_i := F_{h(x)|X}(h_i(x)|X_i)$. This is a general property of conditional distributions, see e.g. Lemma 3 of \citet{goff2024testingidentifyingassumptionsparametric} for a proof. Observe that since $h_i(x)=x^T\beta+U_i$ with $U_i \indep X_i$, $T_i=F_{U}(U_i)$. Define $\tilde{h}_i(x):=Q_{H(x)|X_i}(T_i)$. We can similarly work out $\tilde{h}_i(x)$ to be $\tilde{h}_i(x)=Q_{x^T \beta +U_i|X}(T_i|X_i) = x'\beta + Q_{U}(T_i)$ using $U_i \indep X_i$. Putting this all together, we have that with probability one $h_i(x)=\tilde{h}_i(x) = x'\beta + \tilde{U}_i$, where we define $\tilde{U}_i:=Q_{U}(F_U(U_i))$. Note that $Q_{U}(F_U(U_i)) \sim Unif[0,1]$ and is independent of $X_i$ (if $U_i$ is not continuously distributed, $\tilde{U}_i$ can be suitably redefined so that it remains uniform, see Lemma 4 of \citet{goff2024testingidentifyingassumptionsparametric}).}  In the context of causal inference about $H_i$, sufficient assumptions regarding treatment effect homogeneity are thus sufficient to afford inference on the sign of $\theta$, which in this context (randomization of $X_i$) corresponds to the overall average treatment effect (ATE) $$\mathbbm{E}[\Delta_i] = \mathbbm{E}[h(x',U_i)]-\mathbbm{E}[h(x,U_i)] = \mathbbm{E}[H_i|X_i=x']-\mathbbm{E}[H_i|X_i=x]$$
between treatment values $x'$ and $x$.

Assuming a linear causal model with homogeneous treatment effects would be very restrictive, but the above example illustrates a broader point. As Theorem \ref{propdiscrete} shows, differences in the distribution of $R_i$ between two distinct points $X_i=x'$ and $X_i=x$ reveal under random assignment positive aggregations of treatment effects $\Delta_i$, among units whose response value would change given a counterfactual shift from $x$ to $x'$. In the limit that $x' \rightarrow x$, the local derivative of $P(R_i \le r|x)$ at $x$ yields the sign of the average marginal causal effect of changing $x$ among individuals at the threshold between response categories $r$ and $r+1$, as shown in Theorem \ref{propflow}. Whether this local average effect among marginal respondents is informative about the overall effect of changing $x$ to $x'$ depends on how heterogeneous casual effects are in the population.

This logic is familiar from the analysis of instrumental variables with heterogeneous treatment effects. In the LATE model of \citet{Imbens2018}, a binary instrument reveals the average effect of a binary treatment among compliers. Whether this local average is informative about the overall ATE depends on how different treatment effects are between the compliers and other groups in the population. Unlike in the LATE context, the ``marginal respondents'' in our setting that are averaged over in the causal effects revealed by the data constitute a measure-zero subset of the population given that for each reporting function $v$ they represent a single value of the continuous variable $H$ (this is true even after averaging back over $X_i$, cf. Corollary \ref{corr:avgderivative}). Furthermore, the magnitudes of regression derivatives or differences reflect not only magnitudes of causal effects, but the density of happiness values near the thresholds between response categories. This underscores the value of \textit{comparing} the magnitudes of regression coefficients across treatment variables, rather than interpreting the magnitudes individually.

\subsection{Targeting ratios of effects rather than the effect of one treatment}
Indeed, recall from Eq. \eqref{eq:compare} that the ratio of local regression derivatives identifies the ratio of convex combinations of the causal effects of the two continuous treatment variables. In a model $h(x,u) = \mathtt{h}(g(x),u)$ that is weakly separable between $x$ and $u$, Section \ref{sec:weaksep} showed that this ratio in turn identifies both the sign and magnitude of marginal rates of substitution between the treatments. For example, if $g(x) = x^T\beta$, we identify $\beta_2/\beta_1$.

The weakly separable class of functions is quite broad, and includes cases in which we may not even be able to identify the sign of $\beta_1$ or $\beta_2$ individually do to the problem highlighted by \citet{bondandlang}. Yet, we can identify both the sign and the exact magnitude of their \textit{ratio}. This is a counter-intuitive result, so I illustrate it below with a simple example.

Consider the model $h(x,u)=(\beta_1 x_1 +\beta_2 x_2)\cdot \left(u-1/3\right)+1/3$ with $\beta_1, \beta_2 > 0$. This is a weakly separable model with $g(x) = x^T \beta$ and $\mathtt{h}(g,u) = g\cdot (u-1/3)+1/3$. Note that in this model the sign of the effect of a small increase in $x$ depends on $U_i$, if $U_i > 1/3$ then $\partial_{x_1}(x,U_i) = \beta_1 \cdot (U_i-1/3)$ is positive. If $U_i < 1/3$, then $\partial_{x_1}(x,U_i)$ is negative. The same considerations apply to $X_2$. If for example $U_i \sim Unif[0,1]$, then the average effect of a small increase in either treatment ends up being positive, since then $\mathbbm{E}[U_i - 1/3] = 1/6$. However if instead $U_i \sim Unif[0,1/2]$, then the average marginal effect is negative. The distribution of $U_i$ is not known by the researcher, and the sign of $\mathbbm{E}[\partial_{x_j} h(x,U_i)]$ is not identified from the data for either $j \in \{1,2\}$ and any $x$.

However, the sign and the magnitude of $\mathbbm{E}[\partial_{x_2} h(x,U_i)]/\mathbbm{E}[\partial_{x_1} h(x,U_i)]$ \textit{is} identified. The reason is that the unknown sign of each variable's individual effect cancels out in the ratio. If we let $\mathtt{h}'$ denote the partial derivative of $\mathtt{h}$ with respect to its first argument, then $\mathtt{h}'(g,u)=u-1/3$ and:
$$ \frac{\partial_2 \mathbbm{E}[R_i|X_i=x]}{\partial_1 \mathbbm{E}[R_i|X_i=x]} = \frac{\cancel{\mathbbm{E}[\mathtt{h}'(x^T\beta,U_i)]}\cdot \beta_2}{\cancel{\mathbbm{E}[\mathtt{h}'(x^T\beta,U_i)]}\cdot \beta_1} = \frac{\cancel{\mathbbm{E}[U_i-1/3]}\cdot \beta_2}{\cancel{\mathbbm{E}[U_i-1/3]}]\cdot \beta_1} = \frac{\beta_2}{\beta_1} $$
by Equation \eqref{eqexpratio2}. Although the sign of $\mathtt{h}'(g,u)$ varies with $u$ and $\mathbbm{E}[\mathtt{h}'(x^T\beta,U_i)]$ is not identified by the data, it appears in both the numerator and the denominator and does not inhibit knowledge of the ratio $\beta_2/\beta_1$.

We can see this phenomenon manifest with discrete differences in $X$ as well. Suppose that $\beta_2 = 10$ and $\beta_1 = 1$, and $U_i \sim Unif[0,1]$. Consider $x=(1,0)'$ and $x'=(1+\epsilon,0)'$ for $\epsilon > 0$, so that $x^T \beta = 1$ and $(x')^T \beta = 1+\epsilon$. We then have $Q_{H|X=x'}(u)-Q_{H|X=x}(u) =\epsilon\cdot (u - 1/3)$, so the conditional quantile functions always cross at $u=1/3$, as depicted in Figure \ref{fig:quantilesexample} for the case of $\epsilon = 1$. Accordingly, the sign of $\theta = \mathbbm{E}[H_i|X_i=x']-\mathbbm{E}[H_i|X_i=x]$ is not identified, as argued by \citet{bondandlang}. This holds for any $\epsilon$, even as it becomes very close to zero. Accordingly, the sign of the overall average marginal effect $\mathbbm{E}[\partial_{x_1}h(x,U_i)]$ with $x=(1,0)'$ remains unidentified as we take $\epsilon \downarrow 0$. However, as we saw above, the ratio $\mathbbm{E}[\partial_{x_2}h(x,U_i)]/\mathbbm{E}[\partial_{x_1}h(x,U_i)]$ \textit{is} identified, in this weakly separable model of potential outcomes.

\begin{figure}[H]
\small
\begin{center}
\begin{tabular}[t]{m{5cm} l}
	\hspace{-.75cm}
	\includegraphics[height=2.5in]{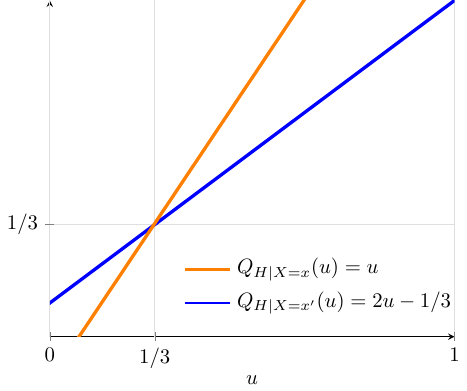}
\end{tabular}
\caption{\label{fig:quantilesexample} Conditional quantile functions from the example described in the text.}
\end{center}
\end{figure}

\section{Relationship to existing econometric models} \label{sec:orderedresponse}
The framework of this paper, outlined in Section \ref{secmodel} is primarily related to two strands of econometric literature: i) models of ordered response; and ii) nonseparable outcome models with possible endogeneity and instrumental variables. This section describes the relationship to both of these literatures.

\subsection{Ordered response models}
The model outlined in Section \ref{secmodel} nests familiar econometric models of ordered response, that typically make parametric assumptions about the functions $h$, $r$ and the distribution of unobservables, while entirely eliminating heterogeneity in $v$.

For example, the probit model treats the case in which $\mathcal{R} = \{0,1\}$, and lets $R_i = \mathbbm{1}(X_i'\beta + U_i \ge 0)$ where $U_i|X_i \sim N(0,\sigma^2)$ where often $\sigma$ is normalized to 1. This fits into the general model above with $V_i$ taken to be degenerate (all units share a value $v$), $\tau(0)=0$, $U_i$ a scalar and $h(x,u)=x^T\beta+u$ for some $\beta \in \mathbbm{R}^{d_x}$. The assumption that $U$ is independent of $X_i$ then implies EXOG. In the probit model, the effect on $H$ of a switch from $X_i=x$ to $X_i=x'$ is common across units, given by $(x'-x)^T\beta$. The ordered probit model maintains this same structure but with a larger set of categories $\mathcal{R}=\{0,1,\dots \bar{R}\}$, with corresponding thresholds $\tau(0), \tau(1), \dots, \tau(\bar{R}-1)$ common across individuals.

Despite the popularity of (ordered) probit and logit models, it is not necessary to impose a parametric structure on $h(x,u)$ or the distribution of $U$ to obtain identification in binary and ordered choice settings. \citet{matzkin1992} shows that $h$ can be identified up to scale under fairly general conditions if $u$ is a scalar and $h$ admits a separable structure: $h(x,u) = g(x)+u$ for some function $g$. This model allows for individual-specific reporting functions in a trivial sense, since owing to the additive separability the distinction between thresholds $\tau_v(r)$ and the error $u$ is simply a matter of definition.\footnote{Indeed, fixing any $r$ and defining $Y^r_i = \mathbbm{1}(R_i \le r)$ we may write $Y^r_i = \mathbbm{1}(g(X_i) +\eta^r_i \le 0)$ where $\eta^r_i = U_i - \tau_{V_i}(r)$. Under conditions given by \citet{matzkin1992}, the function $g$ and the distribution of $\eta^r$ can be identified (up to a scale normalization). See also \citet{cunhaetal}. Since this can be done for each value $r$, the function $g$ is in fact overidentified with more than two categories (see Appendix \ref{sec:testing} for a generalization). \citet{matzkin1994} establishes conditions for identification of $g$ in a weakly separable model $Y_i = r(\texttt{h}(g(X_i),\eta_i))$, but requires $\eta_i$ to be scalar.} However, a separable model like $h(x,u) = g(x)+u$ for potential outcomes, like the probit model, imposes the strong restriction that treatment effects are the same for all individuals. My results allow for treatment effect heterogeneity, and nests a leading case of \citet{matzkin1992} when the treatment variables are all continuous (see Appendix \ref{sec:weaksep}).

\subsection{Nonseparable outcome models with or without endogeneity}

Suppose for the moment that $H$ were observed. Then Equation \eqref{modelh} along with Assumption EXOG would yield a nonseparable model for the outcome $H$ with a set of exogenous regressors $X$, with no restrictions on the dimension of heterogeneity $U$ or functional restrictions like monotonicity in $X$ or $U$. In this general setting, \citet{hoderleinmammen2007,sasaki_2015} show that with continuous $X$ quantile regressions reveals outcome-conditioned average treatment effect parameters (this terminology is due to \citealt{outcome-conditioned}). \citet{kasy2022} provides similar results for a multi-dimensional set of outcome variables, and \citet{chernozukovetal2015} extend to panel data settings. \citet{BKM17} use invertibility assumptions to afford identification of an entire structural function with multi-dimensional outcomes.

However in my setting only $R$ is observed, and not $H$. This leads to the model of Section \ref{secmodel} in which $R,H$ and $X$ are related through Equations \eqref{modelr} and \eqref{modelh}. This structure resembles triangular instrumental variables (IV) models, where my $X$ plays the role of the instrument(s) and Eq. \eqref{modelh} represents the ``first stage'' relationship between the instrument(s) and endogeneous regressor. Reporting functions play the role of the outcome equation in an IV setup, and ``endogeneity'' arises if $U_i \cancel{\indep} V_i$, explicitly allowed in my model. However unlike IV settings, one cannot observe the ``endogenous variable'' $H_i$, which renders the analysis of identification very different in my setting.\footnote{Indeed, the IV analogy yields some intuition for my results: although variation in $X$ induces exogeneous variation in $H$ and in $R$ through $H$, we cannot re-scale the ``reduced form'' relationship between $X$ and $R$ by the ``first stage'' relationship between $H$ and $X$, since $H$ is unobserved.} In the literature thus far that has assumed $H$ is observed, it has been found that monotonicity assumptions can be helpful in securing identification when structural functions are taken to be nonseparable as they are in my model \citep{ImbensNewey2009,triangularDF,triangularT,triangularHHKM}.

The result of \citet{hoderleinmammen2007} for nonseparable models with exogeneity has previously been used to study identification from discrete choice probabilities in \citet{chernozukovetal2019}. \citet{MATZKIN201983} also analyzes some nonseparable models of discrete choice. To my knowledge the present paper is the first to leverage results on the link between quantile regressions and conditional average causal effects to address the concerns highlighted by \citet{bondandlang} regarding the use of ordinal scales.

Finally, I note that this paper is related to the literature on measurement error and misclassification, in that one might view $R$ as a imperfect measure of $H$ contaminated by the reporting function. However, I let latent happiness $H$ and responses $R$ exist on entirely different scales (e.g. $H$ in $\mathbbm{R}$ and $R$ in a set of integers), in the tradition of ordered response models and in common with \citet{bondandlang}. This feature also distinguishes the approach of this paper from models of rounding \citep{hoderleinrounding}, measurement error \citep{huschennach2}, and discrete misclassification \citep{hu2008,misclassification}.

\section{Extensions of the basic model} \label{sec:extensions}

\subsection{Using instrumental variables for identification} \label{sec:iv}
Suppose for that rather than making Assumption EXOG, we instead have a set of observed variables $Z_i$ to use as instruments for $X_i$. We assume each $X_j$ for $j=1 \dots J$ is continuously distributed, and $Z_i$ contains a continuously distributed instrument corresponding to each $X_j$, i.e.	
$$X_{1i} = x_1(Z_i,W_i, \eta_{1i}), \quad X_{2i} = x_2(Z_i,W_i, \eta_{2i}) \quad  \dots \quad X_{Ji}=x_J(Z_i,W_i, \eta_{Ji})$$
Finally, for each $j = 1 \dots J$, suppose that	$x_j(z,w,\eta_j)$ is strictly increasing in $\eta_j$. Let $\eta_i = (\eta_{1i}, \eta_{2i}, \dots \eta_{Ji})^T$. We now assume that $Z_i$, rather than $X_i$, is (conditionally) independent of all other heterogeneity across individuals $i$:
\begin{assumption*}[INSTRUMENT (conditional independence of instruments)]
$$ \left\{Z_i \indep (\eta_i, U_i, V_i)\right\} | W_i$$
\end{assumption*}
\noindent The following result adapted from \citet{ImbensNewey2009} implies that under INSTRUMENT we can use $\eta_i$ as a control variable in $W_i$, in the sense that
\begin{lemma*}
Under INSTRUMENT and the IV model above: $\{X_i \indep (U_i, V_i)\}| (\eta_i, W_i)$
\end{lemma*}
\begin{proof}
Note that INSTRUMENT implies that $\{Z_i \indep (U_i,V_i)\}|W_i, \eta_i$. Furthermore, conditional on $\eta_{ji}$ and $W_i$, the only remaining variation in $X_{ji}$ comes from $Z_{i}$. This is true for each $j$, so conditional on $\eta_i$ and $W_i$, the only variation in $X_i$ comes from variation in $Z_i$, i.e. $X_i$ is simply a function of $Z_i$. The result then follows.
\end{proof}
\noindent Thus, if $\eta_i$ is simply included in the vector $W_i$ to begin with, Theorem \ref{propflow} holds under the weaker assumption of INSTRUMENT, since INSTRUMENT then implies EXOG. ``Controlling'' for $\eta_{i}$ is feasible, because given that each $x_j(z,w,\eta_j)$ is strictly increasing in $\eta_j$, we can without loss redefine $\eta_{ji} = F_{X_j|Z,W}(X_{ji}|Z_i,W_i)$ which can be estimated from the data for each $j$ and individual $i$.\footnote{Note that since $x_j(z,w,\eta_j)$ is strictly increasing in $\eta_j$, 
$F_{X_j|Z,W}(x_j|Z_i=z,W_i=w) = P(\eta_{ji} \le x_j^{-1}(z,w,x_j)|W_i=w)$ where we have also used INSTRUMENT. Define $\tilde{\eta}_{ji}:=F_{X_j|Z,W}(X_{ji}|Z_i,W_i)=P(\eta_{ji} \le x_j^{-1}(Z_i,W_i,X_{ji})|W_i) = F_{\eta_j|W}(\eta_{ji}|W_i)$. Observe from this that we can write $\tilde{\eta{ji}}$ as a function of $\eta_{ji}$, conditional on $W_i$.  Define $\tilde{\eta}_i = (\tilde{\eta}_{1i}, \tilde{\eta}_{2i}, \dots \tilde{\eta}_{Ji})^T$ which is similarly a deterministic function of $\eta_{ji}$ conditional on $W_i$. Since conditioning on $\tilde{\eta}_i$ and $W_i$ is the same as conditioning on $\eta_i$ and $W_i$, the random vector $\tilde{\eta}_{i}$ satisfies $\{Z_i \indep (U_i,V_i)\}|W_i, \tilde{\eta}_i$. Note finally that $\tilde{\eta}_{ji} \sim Unif[0,1]$ and with probability one $X_{ji} = \tilde{x}_j(Z_i,W_i,\tilde{\eta}_{ji})$ where $\tilde{x}_j(z,w,u):=Q_{X_j|W=w,Z=z}(u))$ for each $u \in [0,1]$. Thus the Lemma holds after redefinition of $\eta_i$ to be $\tilde{\eta}_i$ and each function $x_j$ to be $\tilde{x}_j$.} If no controls are needed for INSTRUMENT, then simply let $W_i=F_{X_j|Z}(X_{ji}|Z_i)$ and EXOG now holds.

\subsection{Subjectively-defined latent variables} \label{sec:extended}
In the main body of the paper, I assume that individuals use a reporting function $r_i(h)$ that is an increasing function of the variable $h$ that the researcher is interested in. Given this, the model can accommodate arbitrary heterogeneity in $r_i(\cdot)$ (or equivalently: the locations of the thresholds that $i$ uses), so long as this variation is independent of explanatory variables.

However in many applications, one might worry that not only are the definitions of the categories $\mathcal{R}$ subjective, but so is the definition of the quantity that individuals are asked to use in answering the survey question. For example, when answering a life-satisfaction question some individuals might think about their recent life experiences, while others may think about their whole life in aggregate. Some might spend a lot of time thinking about the question, while others might answer quickly and intuitively. Accordingly, let individual $i$ use variable $\tilde{H}^i$ when they answer the survey question, where $\tilde{H}_i:=\tilde{H}^i_i$ is $i$'s value of this quantity that they define for themself. The key assumption that will allow us to extend the model to account for this kind of heterogeneity is that $\tilde{H}$ is a weakly increasing function of $H$, where $H$ is an objectively-defined variable of ultimate interest to the researcher.

I extend the model as follows: observables $(R_i,X_i)$ are now related by
\begin{align}
R_i &= \tilde{r}_i(\tilde{H}_i)=\tilde{r}(\tilde{H}_i, S_i) \label{modelr2}\\
\tilde{H}_i &= \tilde{h}_i(H_i)=\tilde{h}(H_i, T_i) \label{modelhtilde} \\
H_i &= h_i(X_i)=h(X_i, U_i) \label{modelh2}
\end{align}
where \textit{both} $\tilde{r}(\cdot,s)$ and $\tilde{h}(\cdot,t)$ are assumed to be weakly increasing and left-continuous. The new function, $\tilde{h}_i(h)$, can be defined in terms of counterfactuals: what would $i$'s value of their subjectively-defined latent variable $\tilde{H}^i$ be if their objectively-defined happiness $H_i$ were $h$? $T_i$ can be of arbitrary dimension, allowing individual-specific mappings between $H$ and $\tilde{H}$.

Now suppose that $\{X_{ji} \indep (T_i,U_i,V_i)\} | \ W_i.$ If we define $V_i= (S_i,T_i)$, then EXOG holds, and defining $r(\cdot,v) = \tilde{r}(\tilde{h}(\cdot,t),s)$ MONO now holds as well, allowing us to apply the main results of the paper. Note that EXOG is now stronger than it was in the baseline model: if we want to accommodate heterogeneity in what latent variable $\tilde{H}$ individuals use to answer the question, we must assume that heterogeneity to also be conditionally independent of $X_j$. In addition to the existing exclusion restriction that variation in $X_j$ does not alter reporting functions $\tilde{r}_i$, we now have an additional implicit exclusion restriction that variation in $X_j$ does not affect the subjective definitions $T_i$ that individuals apply to generate $\tilde{H}_i$ in terms of $H_i$.

One nice feature of this extended version of the model is that the researcher may be more willing to make structural assumptions about the function $h(x,u)$ now that it is made explicit that $H$ may differ from what individual's actually have in their mind when they answer the question. For example, if causal effects on some notion of objective life satisfaction $H$ are assumed to be homogeneous (so that $h(x,u)=g(x)+u$), then marginal rates of substitution can be identified through Eq. (\ref{eqexpratio2}), despite individuals using $\tilde{H}$ rather than $H$ to answer the survey question.

\subsection{Multivariate latent variables} \label{sec:multivariate}
In some settings, it may be appealing to assume that subjective responses are driven by a vector of latent variables rather than a single one. For simplicity, I in this section assume no control variables $W_i$ are needed for EXOG.

For example, \citet{gradmentalhealth} studies the mental health of economics graduate students in U.S. PhD programs, and include a question in which respondents are asked to agree or disagree with the statement ``I have very good friends at my Economics Department''. In such a case, respondents might consider both the quantity and quality of friendships in their definition of ``having good friends''. The emphasis that respondents place on each may also vary by individual.

To model this case, we might replace Eq. (\ref{modelr}) with
$$R_i = r(H_{1i},H_{2i},V_i)$$
where $r$ is weakly increasing in both $H_1$ (number of friends) and $H_2$ (``average'' quality of friendships). We further assume two separate structural functions $h_1(X,U)$ and $h_2(X,U)$ describing the effects of the $X$ on quantity and quality of friendships, respectively. 

For simplicity, let us first consider a case with a single reporting function $r(H_1,H_2)$, and a scalar $x$. It will be useful to write
\begin{equation} \label{eq:Tintegral}
\frac{d}{dx _j} P(R_i \le r|X_i=x) = \int \int_{T(r)} \frac{d}{dx} f_H(h_1,h_2|x) \cdot dh_1 dh_2
\end{equation}
where $T(r)$ is the set of $(h_1, h_2)$ such that $r(h_1, h_2) \le r$. In the above I have assumed dominated convergence so that one can interchange the integrals and derivative.

In the two-dimensional case, Eq. 4.1 of \citet{hoderleinmammen_EJ} show that a quantity like $\frac{d}{dx} f_H(h_1,h_2|x)$ can be rewritten as:
$$\frac{d}{dx} f_H(h_1,h_2|x) = -\nabla \circ \begin{pmatrix} f_H(h_1,h_2|x) \cdot \mathbbm{E}[\partial_x h_1(x,U)|H_{1i}=h_1,H_{2i}=h_2,X_i=x]\\
f_H(h_1,h_2|x) \cdot \mathbbm{E}[\partial_x h_2(x,U)|H_{1i}=h_1,H_{2i}=h_2,X_i=x]
\end{pmatrix}$$
where for a vector-valued function $\mathbbm{h}(x)$, we let $\nabla \circ \mathbbm{h}$ denote the divergence of $\mathbbm{h}$. More generally, \citet{kasy2022} shows that for a vector $\mathbbm{h} = (h_1,h_2,\dots h_K)'$ of any finite dimension $K$:
$$\frac{d}{dx} f_H(\mathbf{h}|x) = -\nabla \circ \left\{f_H(\mathbf{h}|x) \cdot \mathbbm{E}[\partial_x \mathbf{h}(x,U)|\mathbf{h},x] \right\}$$
where we let $\mathbf{h}(x,U)$ be a vector of $(\mathbf{h}_1(x,U),\mathbf{h}_2(x,U) \dots \mathbf{h}_K(x,U))'$.

In the general case with any $K \ge 1$ and again allowing reporting-function heterogeneity (satisfying EXOG), and multiple treatment variables, Eq. (\ref{eq:Tintegral}) becomes
\begin{equation} \label{eq:Tintegral2}
\frac{d}{dx _j} P(R_i \le r|X_i=x) = \int dF_{V|W}(v|w) \int_{T_v(r)} \frac{d}{dx_j} f_H(\mathbf{h}|x) \cdot d\mathbf{h}
\end{equation}
where $T_v(r) := \{\mathbf{h}: r(\mathbf{h},v) \le r\}$.

An application of the divergence theorem allows us to rewrite Eq. (\ref{eq:Tintegral}) as an integral over \textit{the boundary} $\partial T_v(r)$ of the set $T_v(r)$:
\begin{align*}
\frac{d}{dx _j} P(R_i \le r|X_i=x) &= \int dF_{V|W}(v|w)\int_{\partial T_v(r)} f_H(\mathbf{h}|x,v) \cdot \mathbbm{E}[\partial_{x_j} \mathbf{h}(x,U)|\mathbf{h},x,v] \circ \mathbf{n}_{v}(\ell) \cdot d\ell
\end{align*}
where $\mathbf{n}_{x,v}(\ell)$ represents a normal vector perpendicular to $\partial_{T_v(r)}$ at a point indexed by $\ell$. Figure \ref{fig:2d} depicts this in the two-dimensional example. In that case, $\ell$ is a scalar index that parameterizes the path along the one-dimensional boundary of $T_v(r)$.
\begin{figure}[H]
\small
\begin{center}
\begin{tabular}[t]{m{5cm} l}
	\hspace{-.75cm}
	\includegraphics[height=2.75in]{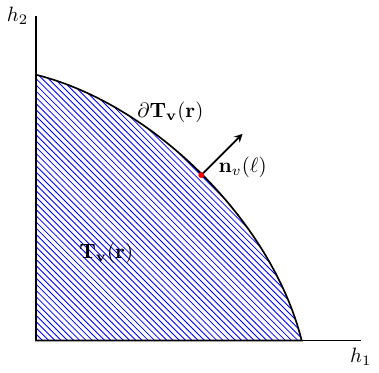}
\end{tabular}
\caption{\label{fig:2d}Components of $\hat{n}(\ell)$ are positive, by monotonicity of $h(h_1,h_2,v)$ w.r.t $h_1$ and $h_2$.}
\end{center}
\end{figure}
Provided that $r(\mathbf{h},v)$ is weakly increasing in each component of $\mathbf{h}$ (for all reporting functions $v$), the components $n_{v,j}(\ell)$ of $\mathbf{n}_v(\ell)$ will be positive, as illustrated in Figure \ref{fig:2d}.\\

\noindent In the two-dimensional case for example, we have:
\begin{align*}
-\frac{d}{dx_j}P(R_i \le r|X_i=x)   
&=\int dF_{V|W}(v|w) \int_{\partial T_v(r)} f_H(\mathbf{h}|x,v) \cdot \left\{ \hat{n}_{v,1}(\ell) \cdot \mathbbm{E}[\partial_{x_j} h_1(x,U)|\mathbf{h},x,v]\right.\\
&\left. \hspace{2in} + \hat{n}_{v,2}(\ell) \cdot \mathbbm{E}[\partial_{x_j} h_2(x,U)|\mathbf{h},x,v] \right\} \cdot d\ell
\end{align*}
Suppose for the moment that $h_j(x,u) = x'\beta_k+u$ where $\beta_{jk}$ represents the effect of treatment variable $X_j$ on $H_k$. Then this becomes
$$\frac{d}{d x_j}P(R_i \le r|X_i=x) = -\mathbbm{E}\left[\left.\int _{\partial T_{V_i}(r)} \left\{\beta_{j1}\cdot \hat{n}_{v,1}(\ell) + \beta_{j2}\cdot \hat{n}_{v,2}(\ell)\right\} \cdot d \ell\right|X_i=x\right]$$
where the expectation is over response functions $V_i$. 

Unless the boundary $\partial T_v(r)$ is linear in $\mathbf{h}$, the positive weights $\hat{n}_{v,2}(\ell)$ will generally vary with $\ell$ across the inner integral. However, the effects of two treatment variables can still be meaningfully compared. For example, suppose we have two continuous treatment variables of interest: $X_1$ and $X_2$, and that for any latent variable $H_k$, the effect of $X_1$ on $H_k$ is $\gamma$ times as large as the effect of $X_2$ on $H_k$. Then:
\begin{align*}
\frac{\frac{d}{d x_1}P(R_i \le r|X_i=x)}{\frac{d}{d x_2}P(R_i \le r|X_i=x)} &= \frac{\mathbbm{E}\left[\left.\int _{\partial T_{V_i}(r)} \left\{\beta_{11}\cdot \hat{n}_{v,1}(\ell) + \beta_{12}\cdot \hat{n}_{v,2}(\ell)\right\} \cdot d \ell\right|X_i=x\right]}{\mathbbm{E}\left[\left.\int _{\partial T_{V_i}(r)} \left\{\beta_{21}\cdot \hat{n}_{v,1}(\ell) + \beta_{22}\cdot \hat{n}_{v,2}(\ell)\right\} \cdot d \ell\right|X_i=x\right]}\\
&= \frac{\mathbbm{E}\left[\left.\int _{\partial T_{V_i}(r)} \left\{\gamma \beta_{21}\cdot \hat{n}_{v,1}(\ell) + \gamma \beta_{22}\cdot \hat{n}_{v,2}(\ell)\right\} \cdot d \ell\right|X_i=x\right]}{\mathbbm{E}\left[\left.\int _{\partial T_{V_i}(r)} \left\{\beta_{21}\cdot \hat{n}_{v,1}(\ell) + \beta_{22}\cdot \hat{n}_{v,2}(\ell)\right\} \cdot d \ell\right|X_i=x\right]} = \gamma
\end{align*}

\section{Additional identification results for continuous treatments} \label{sec:additionalcont}

\subsection{Additional results in the weakly separable case} \label{sec:weaksep}

This appendix continues the analysis of a weakly separable structural function $h(x,u) = \mathtt{h}(g(x),u)$ from Section \ref{sec:weaksepmain} in the main text.

In the still simpler case of a partially linear $h$ function, (\ref{eqexpratio2}) leads to the following:
\begin{corollary} \label{proplinearratio}
Suppose MONO and EXOG and REG$_j$ for $j=\{1,2\}$ hold, and that $h(x,u)$ takes the form $h(x,u) = x_1 \beta_1+x_2\beta_2 + g(x_3, \dots x_J)+u$ (e.g. $h(x,u)=x^T\beta+u$) with $\beta_2 \ne 0$. Then, if EXOG holds with no control variables $\mathbbm{E}[R_i|x]$ is also weakly separable, i.e.
$\mathbbm{E}[R_i|x]=\phi\left(\gamma_1 x_1 + \gamma_2 x_2, x_3 \dots x_{J}\right)$ for some function $\phi$, and ${\gamma_2}/{\gamma_1}=\beta_2/\beta_1$. With controls, we instead have that $\mathbbm{E}[R_i|x,w]$ is weakly separable in $x$ for a fixed $w$, that is $\gamma_1, \gamma_2$ and function $\phi$ may all depend on $w$.
\end{corollary}
\begin{proof}
Fix a $w$, and let $m(x):=\mathbbm{E}[R_i|X_i=x,W_i=w]$. By (\ref{eqexpratio2}), we have that $\partial_{x_2} m(x)/\partial_{x_1} m(x) = \beta_2/\beta_1$, for all $x$. This implies that $m$ takes the form of $\phi$ above.
\end{proof}
\noindent As a final note, we can see how Theorem \ref{propflow} recovers an identification result of \citet{matzkin1994} in the case of no controls $W_i$ and $REG_j$ holding for all components $X_j$ of $x$. Note first that given the weakness of the assumptions made, we could only ever hope to identify $g(x)$ up to an increasing transformation. One functional restriction that removes this arbitrariness, considered by \citet{matzkin1994}, is to suppose $g(x)$ is homogeneous of degree one. \citet{matzkin1994} also imposes that $u$ be a scalar. In this case, Eq. \eqref{eqexpratio2} implies that $g$ is identified up to scale:
\begin{proposition} \label{prop:matzkin}
Suppose MONO and EXOG hold, there are no controls $W$, and each of the $X_1 \dots X_J$ are continuously distributed satisfying REG. Suppose further that $h(x,u)=\texttt{h}(g(x),u)$, where $g$ is homogeneous of degree one, continuously differentiable, and for some $k$: $\partial_{x_k} g(x) \ne 0$ for all $x \in \mathcal{X}$ with $\mathcal{X}$ a convex set in $\mathbbm{R}^J$. Then $g(x)$ is identified up to an overall scale.
\end{proposition}
\begin{proof}
See Appendix \ref{sec:proofs}, in which Eq. (\ref{eq:g}) gives an explicit expression for $g(x)$.
\end{proof}
\noindent Note that Proposition \ref{prop:matzkin} does not require $U_i$ to be a scalar, generalizing the result of \citet{matzkin1994} in this regard.

\subsection{Details: marginal rates of substitution} \label{sec:mrs}
A convenient feature of a weakly separable model like (\ref{eqweaksep}) is that since individual heterogeneity $U$ affects the $X$ variables after they are aggregated by $g$, ratios like $\partial_{x_1} g(x)/\partial_{x_2} g(x)$ captures the marginal rate of substitution between $x_1$ and $x_2$ for each unit. By contrast, (\ref{eqexpratio}) is not necessarily equal to a weighted average over marginal rates of substitution in the population, when they are heterogeneous between units. The following proposition gives a special case in which it does, without the strong condition of weak separability.
\begin{proposition} \label{propmrs}
If in addition to the assumptions of Theorem \ref{propflow} for $j=1,2$, we have
\begin{itemize}
\item $Cov\left(\left.\frac{\partial_{x_2}h(x,U_i)}{\partial_{x_1} h(x,U_i)}, \partial_{x_1} h(x,U_i)\right|H_i \in \tau_{V_i}, x,w\right) = 0$
\item $\{V_i \indep U_i\} \textrm{ }| \textrm{ } (H_i \in \tau_{V_i},X_i,W_i)$ 
\end{itemize} then	$$ \mathbbm{E}\left[\left.\frac{\partial_{x_2}h(x,U_i)}{\partial_{x_1}h(x,U_i)}\right|H_i \in \tau_{V_i}, X_i=x,W_i=w\right] = \frac{\partial_{x_2}\mathbbm{E}[R_i|X_i=x,W_i=w]}{\partial_{x_1}\mathbbm{E}[R_i|X_i=x,W_i=w]}$$
If $Cov\left(\left.\frac{\partial_{x_2} h(x,U_i)}{\partial_{x_1} h(x,U_i)}, \partial_{x_1} h(x,U_i)\right|H_i \in \tau_{V_i}, x,w\right) \le 0$, then $ \mathbbm{E}\left[\left.\frac{\partial_{x_2}h(x,U_i)}{\partial_{x_1}h(x,U_i)}\right|H_i \in \tau_{V_i}, x,w\right] \ge \frac{\partial_{x_2}\mathbbm{E}[R_i|X_i=x,W_i=w]}{\partial_{x_1}\mathbbm{E}[R_i|X_i=x,W_i=w]}$ and vice-versa if the inequality is reversed.
\end{proposition}
\begin{proof}
Given $\{V_i \indep U_i\} \textrm{ }| \textrm{ } (H_i \in \tau_{V_i},X_i,W_i)$ and (\ref{eqtauV}), we have for $j \in \{1,2\}$
\begin{align*}
\partial_{x_j}\mathbbm{E}[R_i|x,w]
&=\mathbbm{E}\left[\rho(x,V_i,w)|H_i \in \tau_{V_i}, X_i=x,W_i=w\right] \cdot \mathbbm{E}\left[\partial_{x_j} h(x,U_i)|H_i \in \tau_{V_i}, X_i=x,W_i=w\right]
\end{align*}
So the RHS of (\ref{eqexpratio}) becomes: $\mathbbm{E}\left[\left.\partial_{x_2} h(x,U_i)\right|H_i \in \tau_{V_i}, x,w\right]/\mathbbm{E}\left[\left.\partial_{x_1} h(x,U_i)\right|H_i \in \tau_{V_i}, x,w\right]$. Now, using $Cov\left(\left.\frac{\partial_{x_2} h(x,U_i)}{\partial_{x_1} h(x,U_i)}, \partial_{x_1} h(x,U_i)\right|H_i \in \tau_{V_i}, x,w\right) \le 0,$ $$\mathbbm{E}\left[\left.\partial_{x_1} h(x,U_i)\right|H_i \in \tau_{V_i}, x,w\right]\le\mathbbm{E}\left[\left.\frac{\partial_{x_2} h(x,U_i)}{\partial_{x_1} h(x,U_i)}\right|H_i \in \tau_{V_i}, x,w\right]\cdot \mathbbm{E}\left[\left.\partial_{x_1} h(x,U_i)\right|H_i \in \tau_{V_i}, x,w\right]$$
and analogously if $\le$ is replaced with $\ge$.
\end{proof}
\noindent Proposition \ref{propmrs} requires reporting heterogeneity $V_i$ to be conditionally orthogonal to structural function heterogeneity $U_i$. Further, one must be able to at least sign the correlation of marginal rates of substitution and heterogeneity in marginal effects with respect to $x_2$. This correlation might be negative, if for example, individuals with high returns to $x_2$ do not have returns to $x_1$ that are proportionally as high, on average.

\subsection{Characterizing the marginal respondents} \label{sec:characterizing}
The following result gives conditions under which average characteristics of respondents on the margin between response category $r$ and $r+1$, which drive the average causal effect identified by Theorem \ref{propflow}, can be identified from the data:
\begin{proposition} \label{prop:chars}
	Let $A_i$ be an individual characteristic such that EXOG holds conditionally on $A_i$, i.e. $\{X_{i} \indep U_i\}|(A_i,W_i,V_i)$ and $\{X_{ji} \indep V_i\}|(A_i,W_i)$. Suppose further that for treatment $j$ the sign of $\partial_{x_j}h(x,U_i)$ is the same for all individuals $i$. Then (under $REG_j$ and further regularity conditions described in the proof):
	\begin{equation} \label{eq:Amean1}
		\mathbbm{E}[A_i|H_i = \tau_{V_i}(r),X_i=x,W_i=w] = \frac{ \mathbbm{E}[A_i \cdot \partial_{x_{j}}P(R_i \le r|A_i,x,w)|x,w]}{\mathbbm{E}[\partial_{x_{j}}P(R_i \le r|A_i,x,w)|x,w]}
	\end{equation}
	Under the stronger independence condition that $\{X_{i} \indep (A_i,U_i,V_i)\}|W_i$, this becomes
	\begin{equation} \label{eq:Amean2}
		\mathbbm{E}[A_i|h(x,U_i) = \tau_{V_i}(r), X_i=x,W_i=w] = \frac{\partial_{x_j} \mathbbm{E}[A_i \cdot \mathbbm{1}(R_i \le r)|x,w]}{\partial_{x_j} P(R_i \le r|x,w)}
	\end{equation}
\end{proposition} 
\noindent 
The stronger assumption $\{X \indep (A,U,V)\}|W$ in Proposition \ref{prop:chars} leading to Eq. \eqref{eq:Amean2} is a natural one if the treatment(s) $X$ are as-good-as-randomly assigned (conditional on $W$), and $A$ represents a characteristic of individuals unaffected by the treatments $X$. In this case $A$ will be independent of the treatments in the same sense that $U$ and $V$ are. As an example, one could in a study in which gender is observed estimate the proportion of respondents at each response margin $r$ that are women. To do this, one only needs to supplement the regression contemplated by Theorem \ref{propflow} with another than multiplies $\mathbbm{1}(R_i \le r)$ by characteristic $A_i$, and compute the ratio of regression derivatives.

The weaker condition leading to Eq. \eqref{eq:Amean1} would hold if $A$ is a variable that \textit{could} be added as a valid control variable in $W$, but does not \textit{need to} be for EXOG to hold. This is perhaps harder to motivate, but it is certainly weaker than the above. \citet{abadie} similarly considers the identification of mean attributes of IV compliers, when those attributes represent valid control variables.\footnote{In the case of complier characteristics, the LATE monotonicity assumption plays a role analogous to the assumption that $\partial_{x_j}h(x,U_i)$ is common across individuals in Proposition \ref{prop:chars}. Analogously, the compliers are not individually identified.} The result of Proposition \ref{prop:chars} is also related to an intermediate result used in the proof of Theorem 1 in \citet{triangularHHKM}.

A particularly simple special case occurs when $A_i$ is binary. Then \eqref{eq:Amean1} yields $P(A_i=1|H_i = \tau_{V_i}(r),x,w)/P(A_i=1|x,w) = \partial_{x_{j}}P(R_i \le r|A_i=1,x,w)/\mathbbm{E}[\partial_{x_{j}}P(R_i \le r|X_i=x,A_i)|x,w]$. As a consequence, we then have that:
\begin{equation} \label{eq:relativeodds}
	\frac{P(A_i=1|H_i = \tau_{V_i}(r),x,w)/P(A_i=0|H_i = \tau_{V_i}(r),x,w)}{P(A_i=1|x,w)/P(A_i=0|x,w)} = \frac{\partial_{x_{j}}P(R_i \le r|A_i=1,x,w)}{\partial_{x_{j}}P(R_i \le r|A_i=0,x,w)}
\end{equation}
This says that, for example, the ratio of the local regression derivative for $R_i \le r$ between the male and female subsamples reveals the odds (conditional on $X_i=x$) of being a woman for the marginal respondents of response category $r$, as compared to the odds of being a woman for \textit{all} respondents (including the infra-marginal ones).

The simplest implementation of Eq. \eqref{eq:Amean2} would take the conditional expectation of $A_i \cdot \mathbbm{1}(R_i\le r)$ to be linear in $x$ and $w$, in addition to assuming a linear probability model for $\mathbbm{1}(R_i\le r)$. Given this restriction, the identified quantity
$$\mathbbm{E}[A_i|h(x,U_i) = \tau_{V_i}(r), X_i=x,W_i=w] = \frac{\partial_{x_j} \mathbbm{E}[A_i \cdot \mathbbm{1}(R_i \le r)|x,w]}{\partial_{x_j} P(R_i \le r|x,w)}$$
does not depend on $x$ or $w$, and thus the ratio of the coefficient on $X_{ji}$ in these two regressions identifies $\mathbbm{E}[A_i|H_i = \tau_{V_i}(r)]$. The results in Figure \ref{fig:attributes_self} choose $X_{ji}$ to be $i$'s household income, though Proposition \ref{prop:chars} could also be applied using the regression coefficients for PUMA income instead under the same assumptions.

The bottom panels of Figure \ref{fig:attributes_self} similarly approximate the relevant regressions with linear probability models, which in turn implies that the local relative odds
$$\frac{P(A_i=1|H_i = \tau_{V_i}(r),X_i=x,W_i=w)/P(A_i=0|H_i = \tau_{V_i}(r),X_i=x,W_i=w)}{P(A_i=1|X_i=x,W_i=w)/P(A_i=0|X_i=x,W_i=w)}$$
do not depend on $x$ or $w$ for a given $r$.

\subsection{How reporting heterogeneity can help instead of hurt causal inferences} \label{app:helphurt}
Ex-ante, it would seem that allowing for heterogeneity in reporting functions $r(\cdot,V_i)$ across individuals should make inferences about causal effects on $H_i$ only more difficult. After all, heterogeneity in $V_i$ precludes interpersonal comparisons of $H_i$ between two individuals on the basis of their observed $R_i$ (as illustrated in Figure \ref{fighonest}). 

It is perhaps counter-intuitive, then that reporting function heterogeneity can in fact be \textit{helpful} in drawing inferences about average causal effects on $H_i$ in the population overall. As Theorem \ref{propflow} demonstrates, assumption EXOG is sufficient to make differences in the distribution of $R_i$ with respect to $X_i$ at the population level causally interpretable, though the effects are local to the individuals that happen to be on the margin between response categories. 

In the extreme case, if $\tau_{V_i}(r)$ were degenerate at value $\tau(r)$, the observable derivative $\partial_{x_j} P(R_i \le r|X_i=x)$ identifies a very specific local average effect:
\begin{equation} \label{eq:degeneratev}
\partial_{x_j} P(R_i \le r|X_i=x) =-f_{h(x,U)}(\tau(r)) \cdot \mathbbm{E}\left[\partial_{x_j} h(x,U_i)|h(x,U_i)=\tau(r)\right]
\end{equation}
where we consider the case with no controls, for simplicity. The RHS of \eqref{eq:degeneratev} might be far from representative of the population mean of $\partial_{x_j} h(x,U_i)$, and may depend heavily on $x$ and $r$ if causal effects are quite heterogeneous.\footnote{For a concrete example, we must look beyond an additively separable model $h(x,U_i) = x^T\beta+U_i$ in which causal effects are homogeneous. With multiplicative heterogeneity in effects $h(x,U_i) = x'\beta \cdot U_i$, the above evaluates to $f_{h(x,U)}(\tau(r)) \cdot \beta_j \cdot \mathbbm{E}\left[U_i|U_i=\tau(r)/x^T\beta \right] = -f_{h(x,U)}(\tau(r)) \cdot \frac{\beta_j \cdot \tau(r)}{x^T\beta}=  f_{U}\left(\frac{\tau(r)}{x^T\beta}\right) \cdot \frac{\beta_j \cdot \tau(r)}{(x^T\beta)^2}$ which necessarily depends on both $x$ and $r$.} Even if we average the local regression derivative over $X_i$, we know by Corollary \ref{corr:avgderivative} that even given a linear model of $P(R_i \le r|X_i,W_i)$, the coefficient $\gamma_{rj}$ remains ``local'' to the marginal respondents, under reporting function homogeneity.

In the other extreme, we could consider a limit of ``maximum'' spread in reporting function heterogeneity, conceptualized as the response thresholds $\tau_{V_i}(r)$ being  uniformly distributed across the real line (or a convex subset of it that contains all happiness values in the population). Corollary \ref{corr:uniformv} shows that if this holds and reporting function heterogeneity is furthermore independent of potential outcomes, then $\partial_{x_j} P(R_i \le r|X_i=x)$ in fact identifies the overall \textit{unconditional} causal effect $\mathbbm{E}\left[\partial_{x_j} h(x,U_i)\right]$, rather than the conditional effect among individuals whose combination of reporting function and $U_i$ make them marginal between response categories $r$ and $r+1$. It then follows that the average derivative $\mathbbm{E}[\partial_{x_j} P(R_i \le r|X_i,W_i)]$ identifies the overall population mean $\mathbbm{E}\left[\partial_{x_j} h(X_i,U_i)\right]$.
\begin{corollary} \label{corr:uniformv}
Suppose that in addition to the assumptions of Theorem \ref{propflow}, i) $U_i \indep V_i|W_i,X_i$ and ii) $\tau_{V_i}(r)$ is uniformly distributed on $\textrm{supp}\{\tau_{V_i}(r)|W_i=w\}=[\mu_w,\ell_w] \subset \mathbbm{R}$ with $\textrm{supp}\{h(x,U_i)\} \subseteq [\mu_w,\ell_w]$. Then
$$\partial_{x_j} P(R_i \le r|x,w)=\frac{-1}{\mu_w-\ell_w}\cdot \mathbbm{E}\left[\partial_{x_j} h(x,U_i)|X_i=x,W_i=w\right]$$
\end{corollary}
\noindent The key feature of Corollary \ref{corr:uniformv} is that it establishes conditions under which $\partial_{x_j} P(R_i \le r|x,w)$ averages over all individuals with $X_i=x,W_i=w$, and not just those on the margin between two response categories. Note that the assumptions of Corollary \ref{corr:uniformv} imply that $\mathbbm{E}[\partial_{x_j} P(R_i \le r|X_i,W_i)] = \mathbbm{E}\left[\left.\frac{-1}{\mu_{W_i}-\ell_{W_i}}\cdot \mathbbm{E}\left[\partial_{x_j} h(X_i,U_i)|W_i=w\right]\right.\right]$, and carry the observable implications that $\partial_{x_j} P(R_i \le r|x,w)$ is the same across $r$ and if $\mathcal{R} = \{0,1, \dots \bar{R}\}$: $\partial_{x_j} \mathbbm{E}[R_i|x,w] = \bar{R}\cdot \partial_{x_j} P(R_i \le r|x,w)$ for each $r$. These implications do not appear satisfied in the empirical application, given the patterns of $\gamma_{jr}$ in Figure \ref{fig:acrossr_self}.

\subsection{Relaxing and testing reporting function invariance} \label{sec:testing}

This section relaxes the assumption that reporting behavior $V_i$ is fixed for each individual and therefore unaffected by variation in $X_i$. In particular, I show that Assumption EXOG is compatible with reporting functions depending directly on observables, in a limited way. I then discusses how even the weakest version of this assumption still leads to testable implications when homogeneity assumptions are placed on causal effects.

\subsubsection{Reporting-function invariance individually versus in distribution}
The assumption that variation in $X$ does not affect reporting functions may be strong. For example, \citet{CPBL} notes that the tendency to bunch at endpoints or the mid-point of scales for life-satisfaction questions is higher among individuals with less formal education, which suggests that a regression of life satisfaction on years of schooling might conflate reporting heterogeneity with variation in actual life satisfaction.\footnote{See also \citet{contirestat} and \citet{jeboreversing} for evidence of non-independence between $V$ and gender.} While a natural experiment could yield variation in schooling uncorrelated with this heterogeneity ex-ante (before schooling takes place), the assumption that education does not still directly change individuals' reporting functions (e.g. their definition of an ``eight'' out of ten in life satisfaction) may be hard to defend. 

To formalize the idea of reporting functions at the individual level being unchanged by $X$, introduce counterfactual notation $V_i^x$ to represent the reporting function that would occur for individual $i$ if $X_i=x$. In this notation, the \textit{actual} reporting function for this individual is $V^{X_i}_i$.\footnote{This counterfactual notation is equivalent to instead treating $V_i$ as fixed for an individual and letting $X_i$ enter directly into the reporting function: $R_i = r(H_i,X_i,V_i)$.}  The following assumption says that components $1 \dots J$ of $X$ are excludable from the reporting function, so that only $W_i$ can enter directly:
\begin{assumption*}[EXCLUSION (full reporting function invariance)] \label{ass:exclusion}
For all $i$, $V_i^{x}=V_i^{x'}$ for any $x$ and $x'$ that differ only in components $1 \dots J$.
\end{assumption*}
\noindent Given EXCLUSION, we may let $V_i= V_i^{W_i}$ and proceed with Assumption EXOG as stated above. However, EXCLUSION is stronger than necessary for my main results, and can be relaxed along similar lines to the ``rank similarity'' assumption of \citet{ch2005}:
\begin{assumption*}[INVARIANT (invariant reporting functions in distribution)] \label{ass:invariant}
Conditional on  $W_i=w$, $V^x \sim V^{x'}$ for any $x$ and $x'$ that differ only in components $1 \dots J$ and for which the remaining components equal $w$. Also, in addition to the second item of EXOG we have:
$\{X_{ji} \indep V^x_i \} \textrm{ }| \textrm{ } W_{i}=w$
for all $w$ and $x$ consistent with $w$.
\end{assumption*}
\noindent Given INVARIANT, we can proceed the definition $V_i = V_i^{X_i}$, and Assumption EXOG now follows. Assumption INVARIANT may also be strong in a given setting, but shows that EXOG does not require full reporting function invariance at the individual level.\\

\subsubsection{Testing reporting-function invariance in separable models} 
Given either EXCLUSION or INVARIANT, it is plausible to make Assumption EXOG under explicit randomization or selection-on-observables type variation in $X_i$. A violation of the ``exclusion restriction'' that $X_i$ does not enter into an individual's reporting function $r_i(\cdot)$ would threaten the first condition of Assumption EXOG that $\{V_i \indep X_{ji}\} | W_i$. This condition has testable implications, when additional structure is assumed on the causal response function $h(x,u)$.

In particular, consider the weak separability condition Eq. (\ref{eqweaksep}) considered in Section \ref{seccontiniousvariation}, that $h(x,u) = \texttt{h}(g(x),u)$ for some function $\texttt{h}$. Then:
\begin{align} \label{eqprobratio2}
\frac{\partial_{x_1}P(R_{i} \le r|x,w)}{\partial_{x_2}P(R_{i} \le r|x,w)} &= \frac{\int dF_{V|W}(v|w) \cdot f_H(\tau_v(r)|x,v,w) \cdot \mathbbm{E}\left[\partial_{x_1} h(x,U_i)|H_i=\tau_v(r),x,v,w\right]}{\int dF_{V|W}(v|w) \cdot f_H(\tau_v(r)|x,v,w) \cdot \mathbbm{E}\left[\partial_{x_2} h(x,U_i)|H_i=\tau_v(r),x,v,w\right]} \nonumber\\
&= \frac{\partial_{x_1}g(x)}{\partial_{x_2}g(x)}
\end{align}
which generalizes Eq. (\ref{eqexpratio2}) to hold for the CDF of responses at any $r$ rather than only for the mean. Importantly, the expression $\frac{\partial_{x_1}g(x)}{\partial_{x_2}g(x)}$ does not depend on $r$, leading to a set of overidentification restrictions when there are multiple thresholds (the number of response categories is 3 or greater). 

This restriction can be leveraged to construct a test for $\{V_i \indep X_{i}\} | W_i$, with $h(x,u) = \texttt{h}(g(x),u)$, MONO, and $\{X_{i} \indep U_i \} \textrm{ }| \textrm{ } (W_{i},V_i)$ (the second component of EXOG) as maintained assumptions. Some algebra shows, using Assumption MONO (see proof of Theorem \ref{propflow}), that:
\begin{align}
\partial_{x_j}P(R_{i} \le r|x,w) &= \partial_{x_j} \int P(H_i \le \tau_v(r)|X_i=x, V_i=v,W_i=w) \cdot  dF(v|x,w) \nonumber\\
&=\int \left\{ \partial_{x_j} P(H_i \le \tau_v(r)|X_i=x, V_i=v,W_i=w) \right\}\cdot dF(v|x,w) \nonumber \\
& \hspace{.5in} +\int P(H_i \le \tau_v(r)|X_i=x, V_i=v,W_i=w)\cdot  \frac{\partial}{\partial_{x_j}} \left\{dF(v|x,w) \right\} 	\label{eq:biasterm}
\end{align}
The first term above evaluates to the quantity in Theorem \ref{propflow} while the second term may be nonzero if $V_i$ is correlated with $X_{ji}$ conditional on $W_i$.\\

\noindent Instead of Eq. (\ref{eqprobratio2}) which assumed EXOG, we now have using (\ref{eq:biasterm})
\begin{align} \label{eqprobratio3}
\frac{\partial_{x_1}P(R_{i} \le r|x,w)}{\partial_{x_2}P(R_{i} \le r|x,w)} &= \frac{\partial_{x_1} g(x) + \frac{\int P(H_i \le \textcolor{purple}{\tau_v(r)}|x,v,w)\cdot \{\partial_{x_1} F_{V|XW}(v|x,w)\}}{\int f_H(\textcolor{purple}{\tau_v(r)}|x,v,w)\cdot dF_{V|XW}(v|x,w) }}{\partial_{x_2} g(x) + \frac{\int P(H_i \le \textcolor{purple}{\tau_v(r)}|x,v,w)\cdot \{\partial_{x_2} F_{V|XW}(v|x,w)\}}{\int f_H(\textcolor{purple}{\tau_v(r)}|x,v,w)\cdot dF_{V|XW}(v|x,w)}}
\end{align}
where the second term in both the numerator and the denominator depend on $r$ through the quantity $\tau_v(r)$, highlighted. Under the maintained assumptions, the only way that $\frac{\partial_{x_1}P(R_{i} \le r|x,w)}{\partial_{x_2}P(R_{i} \le r|x,w)}$ can vary by $r$ is through a failure of $\{V_i \indep X_{i}\} | W_i$. If we further assume linearity of the structural function $g(x)=x^T\beta$, then we obtain additional overidentification restrictions that we can use with (\ref{eqprobratio3}). In particular, $\partial_{x_1}P(R_{i} \le r|x,w)/\partial_{x_2}P(R_{i} \le r|x,w) = \beta_1/\beta_2$ should not depend on $x$, if $\{V_i \indep X_{i}\} | W_i$ holds.

Additional indirect tests for reporting function invariance can be found in the literature. For example, \citet{luttmer2005} compares life satisfaction to other outcome measures often associated with well-being, such as depression and open disagreements within the household. Seeing effects in the same direction, \citet{luttmer2005} concludes that the main results are not likely to driven by individuals changing their ``definition'' of happiness with $X_i$.

Eq. (\ref{eq:biasterm}) can also be used to study the nature of the bias that occurs when the implication $V \indep X |W$ of EXOG fails. Using integration by parts, the second term of (\ref{eq:biasterm}) can be rewritten as
\begin{align}
&\int P(H_i \le \tau_v(r)|X_i=x, V_i=v,W_i=w)\cdot  \frac{ \partial}{\partial_{x_j}} \left\{dF(v|x,w) \right\} \label{eq:nonexogbias}\\
&= (-1)^{d_V}\int \left\{ \frac{\partial}{\partial_{x_j}} F(v|x,w) \right\} \cdot \left\{\partial_{v_1,v_2, \dots v_{d_V}} P(H_i \le \tau_v(r)|X_i=x, V_i=v,W_i=w)\right\} \cdot dv_1 \dots dv_{d_V} \nonumber
\end{align}
provided that $\frac{ \partial}{\partial_{\tilde{v}_1}}\frac{ \partial}{\partial_{\tilde{v}_2}}\dots \frac{ \partial}{\partial_{\tilde{v}_M}}  \left\{\frac{ \partial}{\partial_{x_j}}F(v|x,w) \right\}$ vanishes on the boundary of $\mathcal{V}$, for any subset $\tilde{v}_1 \dots \tilde{v}_M$ of the components of $V$.

Expression (\ref{eq:nonexogbias}) will be positive if, for example, $V$ is a scalar independent of $U$ (conditional on $X$), higher values of $v$ are represent more ``optimistic'' reporting functions (that is, lower thresholds $\tau_v(r)$), and $X_j$ is positively correlated with $V$ (so that $F(v|x)$ decreases as $x_j$ is increased).\footnote{It is in principle possible for this bias term to be negative even if $X_j$ is associated with more optimistic reporting functions: if $U$ and $V$ are correlated in such a way that conditional on $X$ that those with more optimistic reporting functions tend to be less happy (this is difficult, but not impossible, to have happen while $\{X_{j} \indep V\}|W$).} As a simple example, suppose heterogeneity in reporting functions is scalar and takes the form as an additive shift in all thresholds between individuals: $\tau_v(r) = \tau(r)-v$. Individuals with high $V$ are more ``optimistic reporters'', since they require lower values of $H$ to report a given response $r$. If furthermore $U \indep V|X,W$, then (\ref{eq:nonexogbias}) reduces to:
\begin{align*}
\partial_{x_j}& P(R_i \le r|x,w) = \textrm{causal term}\\
& \hspace{1cm} - \int \left\{ \frac{\partial}{\partial_{x_j}} F(v|x,w) \right\} \cdot \left\{\partial_{v} P(H_i \le \tau(r)-v|x,w)\right\} \cdot dv\\
&=\textrm{causal term} - \int f_{H}(\tau(r)-v|x,w) \cdot \left\{ -\frac{\partial}{\partial_{x_j}} F(v|x,w) \right\} \cdot dv
\end{align*}
The second term reflects a positively-weighted integral over the term in brackets, which measures the correlation between $x_j$ and ``reporting optimism'' $v$. If $X_j$ and $V$ are positively correlated, then the second term above in $\partial_{x_j} P(R_i \le r|x,w)$ will be positive, meaning that the observable relationship between $X_j$ and $R$ will be biased upwards by a positive non-causal term. If $V_j$ and $X$ were instead negatively correlated in this example, the bias would be in the other direction.\footnote{Note that if $f_{H}(\tau(r)-v|x,w)$ and $\frac{\partial}{\partial_{x_j}} F(v|x,w)$ are ``uncorrelated'' over $v$ in the sense that $\int \left\{f_{H}(\tau(r)-v|x,w)-\int f_{H}(\tau(r)-v'|x,w)\cdot dv'\right\} \cdot \left\{\frac{\partial}{\partial_{x_j}} F(v|x,w) - \int \frac{\partial}{\partial_{x_j}} F(v'|x,w) \cdot dv' \right\} \cdot dv = 0$, then the density integrates to one and the non-causal term above becomes $\int f_{H}(\tau(r)-v|x,w) \cdot \left\{ -\frac{\partial}{\partial_{x_j}} F(v|x,w) \right\} \cdot dv = - \frac{\partial}{\partial_{x_j}}\mathbbm{E}[V_i|X_i=x,W_i=w]$, i.e. the bias from a failure of independence between $X_j$ and reporting optimism $V$ is simply the rate at which the mean of reporting optimism varies with $X_j$.}

\subsection{Results that assume reporting functions vary idiosyncratically} \label{sec:idio}
The main results in the paper assume both parts of Assumption EXOG: $\{X_{i} \indep V_i \} \textrm{ }| \textrm{ } W_{i}$ and $\{X_{i} \indep U_i \} \textrm{ }| \textrm{ } (W_{i},V_i)$. These are both natural when there is idiosyncratic variation in $X_i$ arising from an experiment or natural experiment, and reporting functions $V_i$ are unaffected by $X_i$. However, if causal inference is not the researcher's goal, and the researcher simply wishes to document features of the joint distribution of $H_i$ and $X_i$, we can let the function $h$ simply represent the conditional quantile function of $H$ as in Eq. Footnote \ref{eqcondquantile} (with the definitions $U_i = (\theta_i, V_i,W_i)^T$, $\theta_i:=F_{H|XVW}(H_i|X_i,V_i,W_i)$ and then $h(x,u) = Q_{H|XVW}(\theta|v,w)$). In this case, model Eq. \ref{modelh} and the latter condition of EXOG holds automatically, since $\theta_i|(X_i,V_i,W_i) \sim Unif[0,1]$, for all $(X_i,V_i,W_i)$ (see Lemmas 3 and 4 of \citet{goff2024testingidentifyingassumptionsparametric} for a proof).

Thus, to learn about the joint distribution of $H_i$ and $X_i$, we only need to assume the first part of EXOG: that $X$ is conditionally independent of reporting heterogeneity $V$ (and not that it is independent of $U$ and $V$ \textit{jointly}). In this case all results from the body of the paper still hold as stated without the first part of EXOG as an explicit assumption.

A stronger assumption that may be attractive in these contexts is that it is reporting heterogeneity $V_i$, rather than $X_j$, that varies ``idiosyncratically''. I.e., we might assume:
\begin{assumption*}[IDR (idiosyncratic reporting)]
$\{V_i \indep (U_i, X_{i}) \} \textrm{ }| \textrm{ } W_{i}$
\end{assumption*}
\noindent Assumption IDR may be an attractive alternative to Assumption EXOG introduced in Section \ref{sec:ci}, though neither assumption nests the other (IDR only implies the first part of EXOG). EXOG aligns more with cases in which there is ``selection-on-observables'': Eq. (\ref{eq:idx}) that $\{X \indep (U,V) \}| W$ may follow naturally in settings in which the researcher has already argued for $\{X \indep U\}|W$. Furthermore, EXOG allows $U$ and $V$ to be arbitrarily correlated, unlike IDR. 

IDR leads to some alternative identification results to the ones in the body of this paper, for establishing features of the joint distribution of $H$ and $X$. To this end, we need not make reference to any structural function $h(x,u)$ for happiness, and can take IDR as saying simply that $\{V_i \indep (H_i, X_{i}) \} \textrm{ }| \textrm{ } W_{i}$. Note that this implication and IDR as stated above are equivalent under the mapping in Footnote \ref{eqcondquantile} that defines $h(x,u)$ as a conditional quantile function, without any causal interpretation. 

For example, if EXOG in Lemma \ref{lemma:hdist} is replaced by Assumption IDR, we can simplify the expression for $\partial_{x_j}P(R_i\le r|x,w)$ to remove conditioning on $V$ in the conditional densities, CDFs, and quantile functions, so that
\begin{equation} \label{propquantileidr}
\nabla_x \mathbbm{E}[R_i|x,w] = \int dF_{V|W}(v|w) \cdot \sum_{r\in \mathcal{R}} f_H(\tau_v(r)|x,w)\cdot \nabla_x \left.Q_{H|XW}(\alpha|x)\right|_{\alpha = F_{H|XW}(\tau_v(r)|x,w)}
\end{equation}
Another result that makes the alternative Assumption IDR rather than EXOG, but allows for discrete variation in $X$:
\begin{proposition}\label{propidr}
Given MONO, IDR, and that $f_{\tau_V|W}(h|w)=\frac{d}{dh}P(h \le \tau_{V_i}(r)|W_i=w)$ exists:
\begin{eqnarray*}
P(R_i \le r|x',w)-P(R_i \le r|x,w) = \int_h \left\{F_{H|XW}(h|x',w)-F_{H|XW}(h|x,w)\right\} \cdot f_{\tau_V|W}(h|w)
\end{eqnarray*}	
\end{proposition}
\noindent Consider the case of no controls $W$ for simplicity. One consequence of Proposition \ref{propidr} is that if $F_{H|X=x'}$ first order stochastically dominates $F_{H|X=x}$, i.e. that $F_{H|X}(h|x')\le F_{H|X}(h|x)$ for all $h$, then under IDR this will be reflected in $F_{R|X=x'}$ first order stochastically dominating $F_{R|X=x}$. That is, the idiosyncratic reporting function transformations preserve this ranking of conditional distributions, in aggregate. This generalizes results found in \citet{ShroderYitzhaki2017}, \citet{bondandlang} and \citet{Kaiser_2019}, which assume a common reporting function across individuals. Note that the existence of $\frac{d}{dh}P(h \le \tau_{V_i}(r))$ requires that for any response $r$ and happiness level $h$, there individuals in the population with thresholds for $r$ very close to $h$.

An alternative to Proposition \ref{propidr} considers the conditional mean rather than the conditional CDF of $R_i$:
\begin{proposition}\label{propidr2}
Given MONO and IDR: 
\begin{eqnarray*}
\mathbbm{E}[R_i|x',w]-\mathbbm{E}[R_i|x,w] = \int_{0}^1 \bar{r}'_{x',x,w}(u) \cdot \left\{Q_{H|XW}(u|x',w)-Q_{H|XW}(u|x,w)\right\}du
\end{eqnarray*}	
where $\bar{r}'_{x',x,w}(u):=\int dF_{V|W}(v|w)\cdot  \frac{r(Q_{H|XW}(u|x',w),v)-r(Q_{H|XW}(u|x,w),v)
}{Q_{H|XW}(u|x',w)-Q_{H|XW}(u|x,w)}$.
\end{proposition}
\noindent Note that Proposition \ref{propidr2} provides a generalization of the expression
$$\mathbbm{E}[H_i|X_i=x']-\mathbbm{E}[H_i|X_i=x] = \int_{0}^1 \left\{Q_{H|X=x'}(u)-Q_{H|X=x}(u)\right\}du,$$
which reveals how an (infeasible) comparison of means of $H_i$ between $x$ and $x'$ aggregates over conditional quantile differences.

\section{Further details on the empirical application} \label{sec:empiricalmore}

\subsection{Sample construction}
I use three data sources in my replication and extension of \citet{luttmer2005}. First, I access the public microdata files for the 1987 and 1992 waves of the NLSF from ICPSR, which constitutes a nationally representative sample of individuals nineteen or older (and able to speak English or Spanish). This provides the variables $R_i$, $X_{1i}$, and $W_i$ in  e.g. Eq. \eqref{eq:ols2var}. I follow \citet{luttmer2005} in deflating monetary values using the consumer price index from the Bureau of Labor Statistics CPI-U series.

Although accessing the geo-coded data from the NLSF is not currently supported, I obtained the predicted PUMA-level log-earnings variable $X_{2i}$ and PUMA identifiers (for clustering standard errors) through correspondence with Erzo F.P. Luttmer and a data sharing agreement with the Social Sciences Research Services at the University of Wisconsin. 
I thank the author for providing this variable to me and the cooperation of the University of Wisconsin. By merging these data with the publicly available data (by NLSF caseid and wave) and keeping only observations that are matched, I automatically implement the sampling restriction of \citet{luttmer2005} to respondents who were married or cohabiting in both waves of the NLSF.

In the regressions reported, the sample used throughout is that of the OLS regressions with controls. Non-parametric and semi-parametric regressions are implemented with the Stata \texttt{npregress kernel} command. In computing average derivatives, this command drops observations for which the local kernel-weighted design matrix is close to singular, which results in a loss of some observations. Table \ref{table:mainself} reports the size of the full sample passed to \texttt{npregress kernel}.

\newpage
\subsection{Results for average of main respondent and spouse} \label{sec:spouse}

\begin{table}[H]
\centering
{
\def\sym#1{\ifmmode^{#1}\else\(^{#1}\)\fi}
\begin{tabular}{l*{5}{c}}
\hline\hline
          &\multicolumn{1}{c}{(1)}&\multicolumn{1}{c}{(2)}&\multicolumn{1}{c}{(3)}&\multicolumn{1}{c}{(4)}&\multicolumn{1}{c}{(5)}\\
          &\multicolumn{1}{c}{OLS}&\multicolumn{1}{c}{Luttmer Table 1}&\multicolumn{1}{c}{Semiparametric}&\multicolumn{1}{c}{OLS}&\multicolumn{1}{c}{Kernel}\\
\hline
Own ln income&   0.0879\sym{***}&    0.123\sym{***}&   0.0844\sym{***}&   0.0349\sym{***}&    0.126\sym{***}\\
          &   (4.59)         &   (6.15)         &   (8.72)         &   (3.61)         &  (13.07)         \\
PUMA ln income&   -0.225\sym{***}&   -0.239\sym{***}&   -0.179\sym{***}&   -0.151\sym{**} &   -0.247\sym{***}\\
          &  (-3.38)         &  (-3.62)         &  (-3.77)         &  (-3.17)         &  (-5.20)         \\
\hline
Ratio PUMA/own&   -2.558         &   -1.943         &   -2.210         &   -4.312         &   -1.936         \\
se(ratio) &    0.937         &        .         &        .         &    1.640         &        .         \\
\hline    &                  &                  &                  &                  &                  \\
Controls  &        X         &        X         &        X         &                  &                  \\
Clustered se&        X         &        X         &                  &        X         &                  \\
\hline    &                  &                  &                  &                  &                  \\
Sample size&     8855         &     8944         &     7822         &     8856         &     7882         \\
\hline\hline
\multicolumn{6}{l}{\footnotesize \textit{t} statistics in parentheses}\\
\multicolumn{6}{l}{\footnotesize \sym{*} \(p<0.05\), \sym{**} \(p<0.01\), \sym{***} \(p<0.001\)}\\
\end{tabular}
}
\\ \vspace{.5cm}
\caption{Replication of \citet{luttmer2005}'s results for the average of main respondent and spouse, and alternative non-linear estimators. Standard errors are clustered at the PUMA level unless otherwise noted. For semiparametric and non-parametric columns, the first two rows report average local derivatives, and ``Ratio'' measures the average ratio of local derivatives, cf. Eq. \eqref{eqexpratioavg2}.} \label{table:main}
\end{table}
\begin{table}[H]
\centering
{
\def\sym#1{\ifmmode^{#1}\else\(^{#1}\)\fi}
\begin{tabular}{l*{6}{c}}
\hline\hline
          &\multicolumn{1}{c}{(1)}&\multicolumn{1}{c}{(2)}&\multicolumn{1}{c}{(3)}&\multicolumn{1}{c}{(4)}&\multicolumn{1}{c}{(5)}&\multicolumn{1}{c}{(6)}\\
          &\multicolumn{1}{c}{R$\le$1}&\multicolumn{1}{c}{R$\le$2}&\multicolumn{1}{c}{R$\le$3}&\multicolumn{1}{c}{R$\le$4}&\multicolumn{1}{c}{R$\le$5}&\multicolumn{1}{c}{R$\le$6}\\
\hline
Own ln income&  0.00134         &  0.00323         &  0.00849\sym{**} &   0.0203\sym{***}&   0.0327\sym{***}&   0.0141         \\
          &   (1.09)         &   (1.85)         &   (2.63)         &   (3.70)         &   (4.22)         &   (1.78)         \\
PUMA ln income&-0.000673         & -0.00475         &  -0.0275\sym{*}  &  -0.0476\sym{*}  &  -0.0731\sym{*}  &  -0.0764\sym{**} \\
          &  (-0.30)         &  (-1.10)         &  (-2.57)         &  (-2.24)         &  (-2.58)         &  (-2.79)         \\
\hline
Ratio PUMA/own&   -0.501         &   -1.473         &   -3.244         &   -2.346         &   -2.233         &   -5.413         \\
se(ratio) &    1.593         &    1.510         &    1.600         &    1.237         &    1.000         &    3.505         \\
Sample size&     8855         &     8855         &     8855         &     8855         &     8855         &     8855         \\
\hline\hline
\multicolumn{7}{l}{\footnotesize \textit{t} statistics in parentheses}\\
\multicolumn{7}{l}{\footnotesize \sym{*} \(p<0.05\), \sym{**} \(p<0.01\), \sym{***} \(p<0.001\)}\\
\end{tabular}
}
\\ \vspace{.5cm}
\caption{Average of main respondent and spouse.  Linear probability model. All regressions include the controls from Table \ref{table:mainself} and standard errors clustered by PUMA.} \label{table:eachthreshold}
\end{table}

\begin{figure}[H]
\begin{center} \vspace{.2cm}
\includegraphics[height=1.75in]{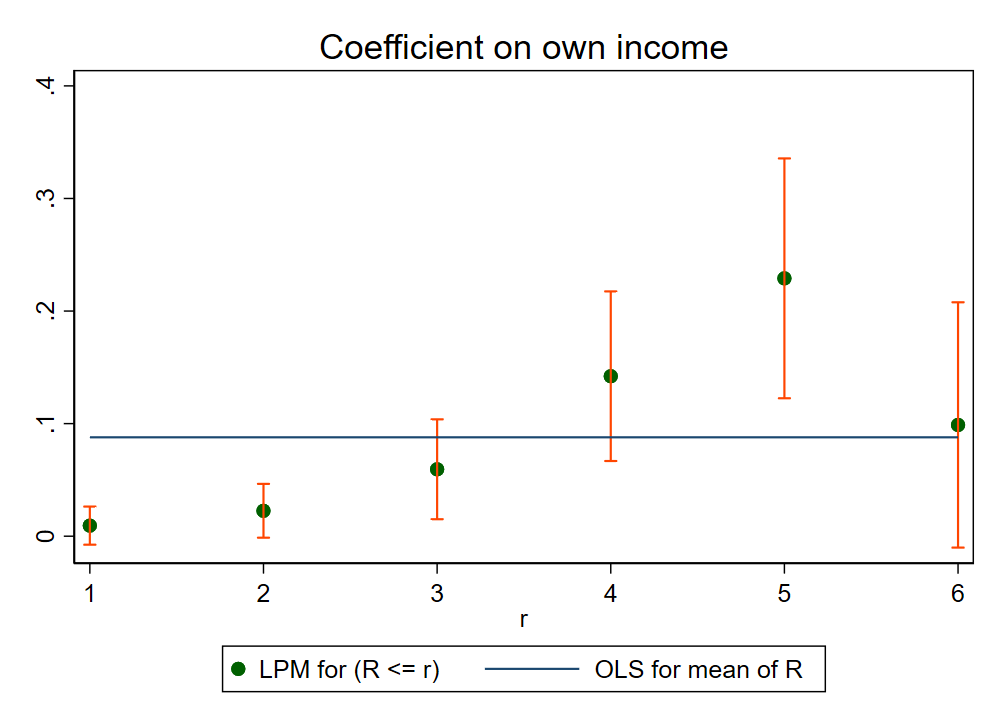}	\quad \quad \includegraphics[height=1.75in]{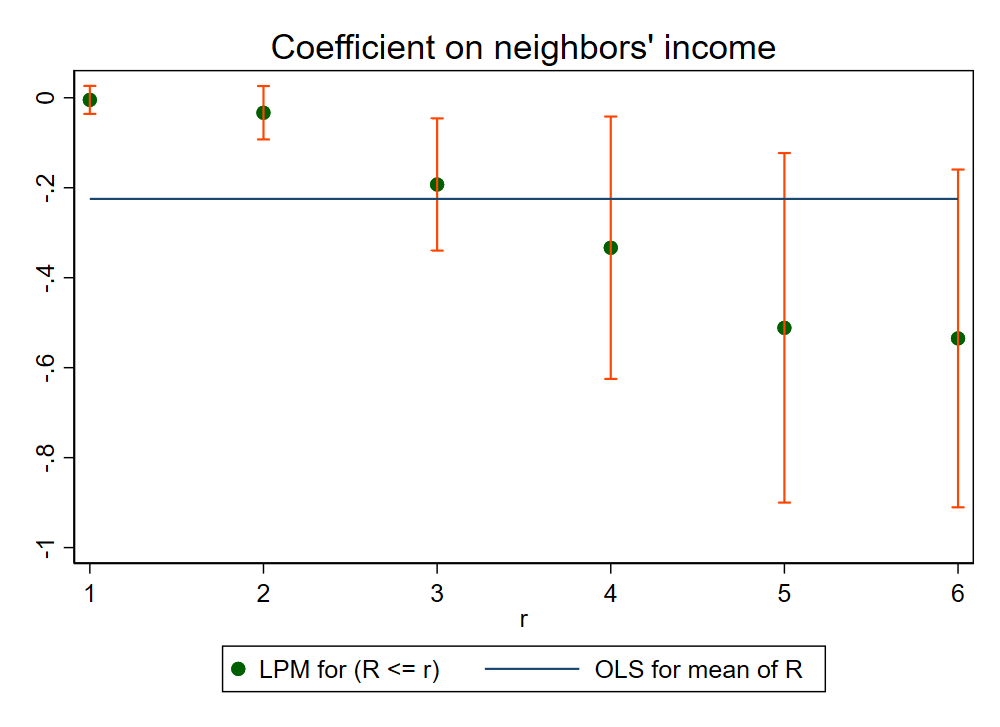} \\ \includegraphics[height=1.75in]{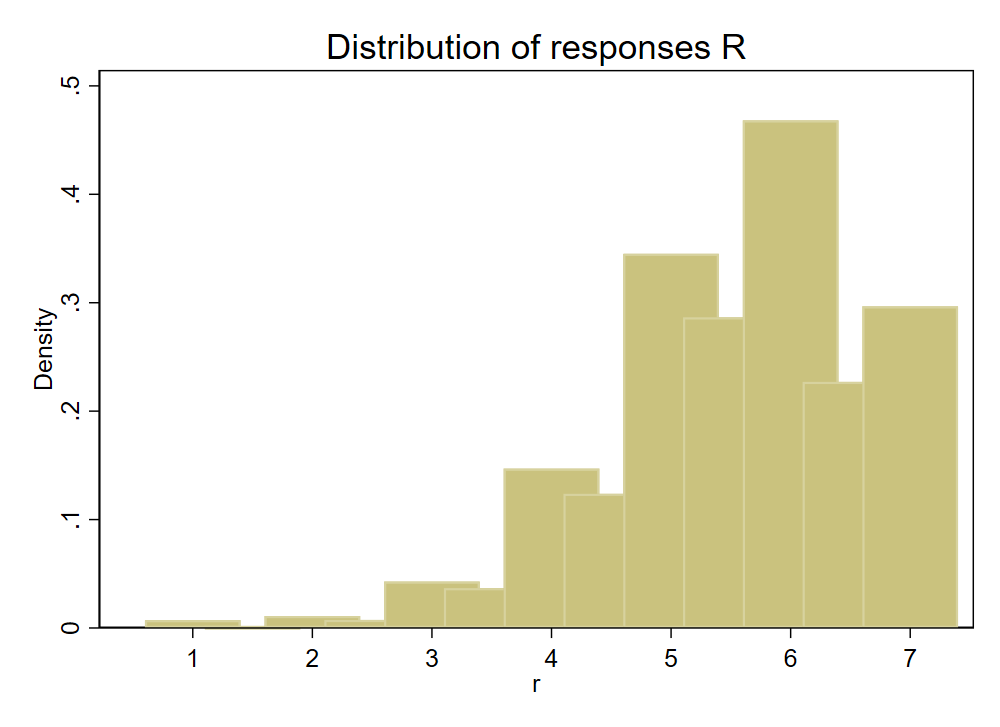}
\quad \quad \includegraphics[height=1.75in]{{"Results/mrsplot_3on_self"}.png}
\caption{Average of main respondent and spouse. Visualization of the OLS estimates of $\gamma_{1r}$ (top-left), $\gamma_{2r}$ (top-right), the ratio $\gamma_{2r}/\gamma_{1r}$ (bottom-right) from the regression $\mathbbm{1}(R_i \le r) = \gamma_{1r} X_{1i}+\gamma_{2r} X_{2i} + \lambda_r^TW_i+\epsilon_{ri}$ for $r \in \{1, 2, \dots 6\}$, along with a histogram of the response categories $r\in \{1, \dots 7\}$. The horizontal line in the upper panels and bottom right panel depicts the corresponding value from mean regression (see Table \ref{table:mainself}).} \label{fig:acrossr}
\end{center}
\end{figure}

\section{What would be identified with a smooth reporting function} \label{seccontinuousreporting}

This section first compares the results for regression derivatives with discrete response categories resulting from Theorem \ref{propflow} to a hypothetical case in which the space of responses $\mathcal{R}$ were instead a continuum. Then I consider such a continuum as a limit of richer and richer response spaces, which is necessary to develop some of the formal results in Section \ref{sec:comparing} of the main paper.

\subsection{Continuous regressors with a continuum of responses} \label{sec:continuum}

It is informative to compare the implications of Theorem \ref{propflow} to what would be identified if $H_i$ were itself directly observable in the data. As a benchmark, this section imagines an intermediate situation in which respondents can select a response from some bounded continuum in $\mathcal{R}$. This allows us to separate the effect of reporting heterogeneity from that of information loss due to discretization of the latent variable $H_i$ into categories.

Suppose $\mathcal{R}$ is a convex subset of $\mathbbm{R}$, for simplicity $\mathcal{R}=[0,\bar{R}]$ for some maximum response value $\bar{R}$. Figure \ref{fighonest3} depicts two examples of reporting functions on this continuum of responses.
\begin{figure}[H]
	\begin{center} \vspace{.2cm}
		\includegraphics[height=2in]{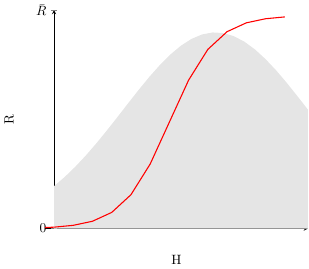} \quad \quad \quad \quad  \includegraphics[height=2in]{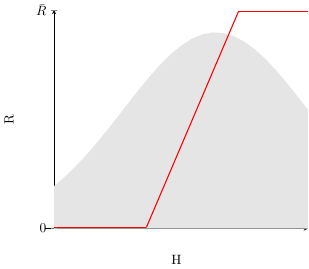}
		\caption{Example of two ``continuous'' reporting functions, with the density of $H$ depicted in gray. \label{fighonest3}}
	\end{center}
\end{figure}
While the example on the left side of Figure \ref{fighonest3} is a smooth sigmoid shape mapping $\mathbbm{R}$ to the interval $[0,\bar{R}]$, the piecewise-linear reporting function on the right has kinks at $\tau_{v}(0)$ and $\tau_{v}(\bar{R})$ beyond which the function is flat. Nevertheless, we may define a derivative function $r'(h, v)$ of any given $r(h,v)$ with respect to $h$, which by virtue of MONO can only fail to exist only at isolated points in $\mathcal{H}$ for a given $v$.\footnote{This is an application of ``Lebesque's theorem'' that monotone functions are differentiable almost everywhere.} Provided that $H_i$ is continuously distributed, it therefore does not affect results to treat $r'(h, v)$ as defined for all $h$. With ``smooth reporting'', we have the following analog of Theorem \ref{propflow}:
\begin{proposition}[] \label{propflowcont}
	Assume MONO, EXOG and REG for at least one $j$, with $\mathcal{R}$ a convex subset of $\mathbbm{R}$. Then: 
	\begin{equation}\label{compare:b}
		\nabla_x \mathbbm{E}[R_i|x,w] = \int dF_{V|W}(v|w) \int dh \cdot \textcolor{blue}{r'(h,v)}\cdot f_H(\textcolor{purple}{h}|x,v,w) \cdot \mathbbm{E}\left[\nabla_x h(x,U_i)|\textcolor{purple}{h},x,v,w\right]
	\end{equation}
	provided the ``boundary condition'': $\lim_{h \rightarrow \pm \infty}  f_H(h|x,v,w) \cdot \mathbbm{E}\left[\partial_{x_j} h(x,U_i)|H_i=h,x,v.w\right] = 0,$
	i.e. average partial effects do not explode for extreme values of $H_i$, any faster than the density of $H_i$ falls off in $h$, for each $v$ and $j$ satisfying REG.
\end{proposition}
\noindent The proof of Proposition \ref{propflowcont} makes use of a result of \citet{kasy2022} that relates derivatives of the density of an outcome with respect to policy variables, to the rate of change of the ``flow density'' quantity introduced in the discussion of Theorem \ref{propflow}.\\

\noindent We can compare this expression to what would be recovered by the infeasible regression of $H_i$ on $X_i$ and $W_i$ (i.e. if $H_i$ were observed):
\begin{subequations}\label{compare}
	\begin{equation}\label{compare:a}
		\nabla_x \mathbbm{E}[H_i|x,w] = \int dF_{V|W}(v|w) \int dh\cdot\textcolor{blue}{1}\cdot f_H(\textcolor{purple}{h}|x,v,w) \cdot \mathbbm{E}\left[\nabla_x h(x,U_i)|H_i=\textcolor{purple}{h},x,v,w\right]
	\end{equation}
	And with integer categories $\mathcal{R}$, using Theorem \ref{propflow}:
	\begin{equation}\label{compare:c}
		\nabla_x \mathbbm{E}[R_i|x,w] = \int dF_{V|W}(v|w) \textcolor{blue}{\sum_{r}} f_H(\textcolor{purple}{\tau_{v}(r)}|x,v,w) \cdot \mathbbm{E}\left[\nabla_x h(x,U_i)|H_i=\textcolor{purple}{\tau_{v}(r)},x,v,w\right]
	\end{equation}
\end{subequations}
These three expressions differ only in what multiplies $f_H(h|x,vmw)\cdot \mathbbm{E}\left[\nabla_x h(x,U_i)|h,x,v,w\right]$ for various values of $h$. Relative to (\ref{compare:a}), (\ref{compare:b}) introduces the derivative $r'(h,v)$ of the reporting function. Intuitively, $r'(h,v)$ corresponds to how closely spaced the thresholds are near a given value of $h$. If this spacing varies across the support of $h$, causal effects will be up-weighted for the $h$ where $r'(h,v)$ is largest, relative to the $h$ where the derivative is smaller. Comparing (\ref{compare:c}) to (\ref{compare:b}) shows that using subjective responses with discrete categories further involves information loss due to the discretization: the integral over all $h$ is replaced by a sum over the thresholds $\tau_v(r)$.\footnote{In the case of linear reporting functions with a continuous response space, Proposition \ref{propflowcont} generalizes a result of \citet{greenetext} for marginal effects in the double-censored Tobit model. The Tobit model takes a linear structural model $h(x,u) = x^T\beta+u$. Greene shows that if the error term $u$ has any continuous distribution, a marginal effect is equal to the true structural effect times the probability that an observation is not censored at either endpoint. (\ref{compare:b}) with no covariates $w$ reduces to	 $\partial_{x_1} \mathbbm{E}[R_i|x] =\beta_1 \cdot \int dF_{V}(v) \cdot \frac{\bar{R}}{\mu(v)-\ell(v)} \cdot P(0 < R_i < \bar{R}|x,v)$
	using that $r'(h,v) = \frac{R}{\mu(v)-\ell(v)}\cdot  \mathbbm{1}(\ell(v) < h < \mu(v))$. The traditional Tobit model further treats $V_i$ as degenerate with $\mu-\ell = R$, so the above recover's Greene's result that $\partial_{x_1} \mathbbm{E}[R_i|X_i=x] = \beta_1 \cdot P(0 < R_i < R|X_i=x)$.
}
\subsection{The ``dense response limit'' of many categories} \label{sec:dense}

In practice, survey questions do not typically allow individuals to give any real number (within a range) in response to subjective questions. However, results based on Proposition \ref{propflowcont} provide a more tractable setting to derive analytical results. If $\mathcal{R}$ is sufficiently rich, then this will provide a useful approximation to the actual properties of that setting (e.g. \cite{aermrs} elicits life-satisfaction data with 100 categories). Below, I give a formal definition of this ``dense response limit'' corresponding to an integer response space $\mathcal{R}=\left\{0,1,\dots,\bar{R}\right\}$, which proves useful in the analysis of Section \ref{sec:comparing}. Appeal to this limit is indicated by the symbol $\stackrel{R}{\rightarrow}$ in the results of Section \ref{sec:comparing}.

To define the dense response limit for a fixed $\bar{R}$, consider a sequence of response spaces $\mathcal{R}_n=\{0,1/n,2/n, \dots, (n \bar{R}) /n\}$ where note that $n \bar{R}$ has $n \bar{R}+1$ categories ranging from $0$ to $(n \bar{R}) /n=\bar{R}$. For a fixed value of reporting heterogeneity $v$, consider a sequence of reporting functions $r_n(\cdot, v)$ indexed by $n$, and let $\tau_{v,n}(\cdot)$ be a function from $\mathcal{R}_n$ to $\mathbbm{R}$ representing the thresholds corresponding to each function $r_n(\cdot, v)$ in the sequence.  
\begin{definition*}[(dense response limit)]
	Fix a $v \in \mathcal{V}$. Consider a sequence of reporting functions $r_n(\cdot, v)$ for $n \rightarrow \infty $. We say that the sequence converges to response function $r(\cdot,v)$ in the dense response limit, denoted as $r_n(\cdot, v) \stackrel{R}{\rightarrow} r(\cdot,v)$, if:
	$$ \lim_{n \rightarrow \infty} \tau_{v,n}(r_n) = \tau_v(r)$$
	for any sequence of $\{r_n\}_{n=1}^{\infty}$ where $r_n \in \mathcal{R}_n$ for each $n$, such that $\lim_{n \rightarrow \infty} r_n = r$ for some $r \in [0,\bar{R}]$ (according to the Euclidean metric on the reals). For any functional of all response functions $\theta(\{r_n(\cdot, v)\}_{v \in \mathcal{V}})$, let $\theta(r_n)\stackrel{R}{\rightarrow} \Theta$ denote that $\Theta$ evaluates the functional $\theta$ at the limiting family of response functions: $\Theta=\theta(r)$.
\end{definition*}
\noindent Intuitively, if the actual response scale is the integers $0$ to $\bar{R}$, the dense response limit instead approximates reports as taking on any real number in $[0,\bar{R}]$.

As a concrete example, consider linear response function $r(h,v)=\min\{\bar{R},\max\{0,h\}\}$ ranging from $0$ to $\bar{R}$ on the continuum $\mathcal{R} = [0,\bar{R}]$. Consider the sequence of reporting functions $r_n(h,v)=\max\limits_{r \in \mathcal{R}_n: h \le \tau_{v,n}(r)} r$, where we let the thresholds be $\tau_{v,n}(r) = r$ for each $r \in \mathcal{R}_n, r < \bar{R}$ (recall that $\tau_{v,n}(r)=\infty$ when $r$ is equal to it's highest value in the response space, in this case $\bar{R}$). The response function $r_n(h,v)$ then represents a ``staircase'' function that jumps from the $r^{th}$ category ($\frac{r-1}{n}$) to the $(r+1)^{th}$ category ($\frac{r}{n}$) at $\tau_{v,n}((r-1)/n) = (r-1)/n$. In this case $r_n(\cdot, v) \stackrel{R}{\rightarrow} r(\cdot,v)$ in the dense response limit, because for any sequence $\{r_n\}_{n=1}^{\infty}$ such that $\lim_{n \rightarrow \infty} r_n = r \in [0,\bar{R}]$ (for example $r_n = \max\limits_{r' \in \mathcal{R}_n: r' \le r} r'$) we have that $\lim_{n \rightarrow \infty} \tau_{v,n}(r_n) = \lim_{n \rightarrow \infty} r_n = r$.

In the dense response limit, discrete differences in the mean of $R_i$ depends upon the average slope $r'(h,V_i)$ of the response function $r(\cdot,V_i)$ for $h$ between $H_i$ and $H_i+\Delta_i$:
\begin{proposition} \label{propdiscretedense}
	Under MONO, EXOG, and REG, then in the dense response limit
	\begin{align*}
		\mathbbm{E}[R_i|x',w]-\mathbbm{E}[R_i|x,w]&\stackrel{R}{\rightarrow}\bar{R}\cdot  \mathbbm{E}[\Delta_i\cdot \bar{r}'(H_i,\Delta_i,V_i)|X_i=x,W_i=w]
	\end{align*}
	where $\bar{r}'(y,\Delta,v):= \frac{1}{\Delta} \int_{y}^{y+\Delta} r'(h,v) \cdot dh$.
\end{proposition}
\noindent Since $\bar{r}' \ge 0$, the weights on $\Delta_i$ in Proposition \ref{propdiscretedense} are positive and aggregate to\footnote{Note that if $\Delta_i$ and $\bar{r}'(H_i,\Delta_i,V_i)$ are uncorrelated conditional on $X_i=x,W_i=w$, then we can further write the RHS of Proposition \ref{propdiscretedense} as $\mathbbm{E}[\Delta_i|X_i=x]\cdot\bar{R}\cdot\mathbbm{E}[\bar{r}'(H_i,\Delta_i,V_i)|X_i=x]$.} $$\Pi_{x,x'} := \bar{R}\cdot\mathbbm{E}[\bar{r}'(H_i,\Delta_i,V_i)|X_i=x,W_i=w]$$
Proposition \ref{propflowcont} in Appendix \ref{seccontinuousreporting} derives an analogous result to Proposition \ref{propdiscretedense} for regression derivatives in the case of a continuous component of $x$. That result shows that the total weight on causal effects in a derivative $\partial_{x_j} \mathbbm{E}[R_i|x,w]$ is, by comparison:
\begin{align*}
	\Pi_{x}:=\bar{R} \cdot \mathbbm{E}[r'(H_i,V_i)|X_i=x,W_i=w]
\end{align*}
\noindent For ease of notation, I leave the dependence of quantities $\Pi_{x}$ and $\Pi_{x,x'}$ on the value of the control variables $W_i$ implicit.

A comparison of $\Pi_{x}$ and $\Pi_{x,x'}$ allows us to interpret the relative magnitudes of discrete and continuous differences in $\mathbbm{E}[R_i|X_i=x,W_i=w]$, as in Eq. (\ref{eq:genratio}). If we have, for example, a binary $X_1$ and continuous $X_2$, and we let
$x'=(1,x_2)$ and $x=(0,x_2)$ for some $x_2 \in \mathbbm{R}$, then:
\begin{equation} \label{eq:linearmodelcompare}
	\frac{\mathbbm{E}[R_i|X_i=x',W_i=w]-\mathbbm{E}[R_i|X_i=x,W_i=w]}{\partial_{x_1}\mathbbm{E}[R_i|X_i=x,W_i=w]} \stackrel{R}{\rightarrow} \frac{\tilde{\beta}_2(x,x',w)}{\tilde{\beta}_1(x'',w)} \cdot \frac{\Pi_{x,x'}}{\Pi_{x}}
\end{equation}
where $\tilde{\beta}_1(x'',w)$ is a convex weighted average over the (derivative) causal effect of $X_1$ on $H$ and $\tilde{\beta}_2(x,x',w)$ is a convex weighted average over causal effects of $X_2$ on $H$. If the aggregate weights are close in magnitude, i.e. $\Pi_{x,x'}/\Pi_{x} \approx 1$, then we can identify the relative magnitudes of these causal averages to a good approximation.

\subsection{Heterogeneous linear reporting in the dense response limit} \label{sec:hetero}
To assess whether the approximation that $\Pi_{x,x'}/\Pi_{x} \approx 1$ is plausible, I impose a further simplification. Let us say that \textit{heterogeneous linear reporting} holds with $\mathcal{R} = \{0,1,\dots \bar{R}\}$ if each individual spaces out the thresholds $\tau_v(r)$ evenly within some individual-specific range, i.e. $\tau_v(r) = \ell(v) + r\cdot \frac{\mu(v)-\ell(v)}{\bar{R}}$ where $\ell(v) = \tau_v(r)$ is the threshold between the two lowest categories for an individual with $V_i=v$, and $\mu(v)$ is the threshold between the top two categories.\footnote{Note that in the limit of many categories $\bar{R}$, this can be well approximated by the linear reporting function $\lim_{\bar{R} \rightarrow \infty} \frac{r(h,v)}{R} = \mathbbm{1}(\ell(v) \le h \le \mu(v))\cdot \frac{h-\ell(v)}{\mu(v)-\ell(v)}$.}

Heterogeneous linear reporting captures the idea that response functions are ``linear'', while still allowing them to vary by individual. Heterogeneous linear reporting may be a reasonable assumption if individuals aim to maximize the informativeness of their responses by equally spreading out the response categories \citep{vanpraag1991}, given their subjective definitions $\ell(v)$ and $\mu(v)$ of the minimum and maximum category thresholds.\footnote{Many studies justify the use of regression based approaches to studying subjective data $R_i$ by interpreting such data as a direct measurement of $H_i$. However, the function $r(\cdot, v)$ cannot literally be the identity function if $\mathcal{R}$ is a set of integers, unless we think that ``true'' happiness also only takes integer values. We might view the cardinality approach as instead supposing that $r(h)$ is homogeneous across individuals and that the thresholds $\tau(r)$ are equally spaced apart.} \citet{Kaiser_2019} summarize empirical evidence in support of linearity, for example from asking individuals directly about their response thresholds, or asking about verifiable outcomes such as an individual's height. 

With heterogeneous linear reporting, a partial identification result holds analytically in the dense response limit:
\begin{proposition} \label{propheterolinear}
	Suppose that the following hold in addition to MONO,EXOG,REG:
	\begin{enumerate}
		\item $r(h, v) \stackrel{R}{\rightarrow} \ell(v) + \frac{h - \ell(v)}{\mu(v)-\ell(v)}$, i.e. reporting is (heterogeneously) linear in the dense response limit; and
		\item For each $\Delta$ in the support of $\Delta_i$, $f_H(h|\Delta,x,v,w)$ is increasing on the interval $[\ell(v)-|\Delta|, \ell(v)+|\Delta|]$, and decreasing on the interval  $[\mu(v)-|\Delta|, \mu(v)+|\Delta|]$
	\end{enumerate}
	Then 
	$$ \frac{\Pi_{x,x'}}{\frac{1}{2}(\Pi_{x}+\Pi_{x'})} \in [1,2],$$	
	\noindent Furthermore, suppose that the lengths of reporting intervals are not too variable across individuals relative to variability in bunching at the endpoints $0$ and $\bar{R}$, in the sense that 
	$$Var\left[\left.\frac{1}{\mu(V_i)-\ell(V_i)} \right|x,w\right] \le Var\left[\left.\mathcal{B}_i \right|x,w\right]\cdot \mathbbm{E}\left\{\left.\frac{1}{\mu(V_i)-\ell(V_i)}\right|x,w\right\}^2,$$ where $\mathcal{B}_i:=P(R_i=0 \textrm{ or } R_i = \bar{R}|X_i,V_i,W_i)$, then
	$$ \frac{1}{2} \le \frac{\Pi_{x,x'}}{\Pi_{x}} \le \frac{1}{(1-\mathbbm{E}[\mathcal{B}_i|x,w])^2},$$
\end{proposition}
\noindent Proposition \ref{propheterolinear} provides two sets of bounds on the ratio of the total weight on causal effects in $\mathbbm{E}[R_i|x',w]-\mathbbm{E}[R_i|x,w]$, to the total weight on causal effects in a derivative $\partial_{x_j} \mathbbm{E}[R_i|x,w]$. The first bound, $\frac{\Pi_{x,x'}}{\frac{1}{2}(\Pi_{x}+\Pi_{x'})} \in [1,2]$ implies that, in the setup of Eq. (\ref{eq:linearmodelcompare}):
\begin{equation} \label{eq:linearmodelcompare2}
	\frac{\mathbbm{E}[R_i|X_i=x',W_i=w]-\mathbbm{E}[R_i|X_i=x,W_i=w]}{\frac{1}{2}\partial_{x_1} \mathbbm{E}[R_i|X_i=x',W_i=w]+\frac{1}{2}\partial_{x_1} \mathbbm{E}[R_i|X_i=x,W_i=w]} \stackrel{R}{\rightarrow}  \theta \cdot \frac{\tilde{\beta}_2(x,x',w)}{\tilde{\beta}_1(x,x',w)}
\end{equation}
where $\theta$ is some number between $1$ and $2$, and $\tilde{\beta}_1(x,x',w)$ is a convex combination of $\partial_{x_1}h(X_i,U_i)$. This bound requires no assumptions on how variable the happiness scale lengths $\mu(V_i)-\ell(V_i)$ can be across individuals with different $V_i$. By contrast, the second set of bounds requires us to assume that the coefficient of variation of $\frac{1}{\mu(V_i)-\ell(V_i)}$ is no greater than the standard deviation of $\mathcal{B}_i$, conditional on $X_i$ and $W_i$. Assuming homogeneity of reporting functions makes the coefficient of variation zero, trivially satisfying the assumption. More generally, the stringency of the assumption can be evaluated from the data by a nonparametric regression of observed bunching at the endpoints of the scale ($0$ and $\bar{R}$) on $X_i$ and $W_i$. 

Note that if the additional restriction justifying the second set of bounds holds, and $\mathbbm{E}[r'(H_i,V_i)|X_i=x,W_i=w]$ is roughly constant in $x$, then
$$\frac{\Pi_{x,x'}}{\Pi_{x}}\approx\frac{\Pi_{x,x'}}{\frac{1}{2}(\Pi_{x}+\Pi_{x'})}$$
and we can take the intersection of the two sets of bounds: $[1,1/(1-\mathbbm{E}[\mathcal{B}_i|x,w])^2]$. This bound will be very narrow if there are few endpoint bunchers when $X_i=x,W_i=w$.

\subsection{Simulation evidence on Proposition \ref{propheterolinear}} \label{sec:sim}
To gather some further suggestive evidence on the comparability of estimates that use discrete vs. continuous variation in $X$. I in this section simulate several data-generating-processes (DGPs) for $H_i$ and for the response functions $r(\cdot, V_i)$. Throughout, I take the response space $\mathcal{R}$ to be a set of integers $0,1,2 \dots \bar{R}$, where the value of $\bar{R}$ will be varied across DGPs. The DGPs are such that EXOG holds with no covariates $W_i$.

Consider a researcher comparing $\mathbbm{E}[R_i|X_i=x']-\mathbbm{E}[R_i|X_i=x]$ to $ \partial_{x_j}\mathbbm{E}[R_i|X_i=x]$ and $ \partial_{x_j}\mathbbm{E}[R_i|X_i=x']$ for some given values $x'$ and $x$, and regressors $X_j$. Given the results of the last section, we seek to compare $\Pi_{x,x'}$, $\Pi_{x}$ and $\Pi_{x'}$ to understand the relative weights each of these estimands place on causal effects.

For now, I suppose heterogeneous linear reporting, so that Proposition \ref{propheterolinear} holds in the dense-response limit $\bar{R} \rightarrow \infty$. Individual reporting functions can be characterized by $\ell(v)$, the value of happiness at which an individual with $V_i=v$ moves from response category $0$ to response category $1$, and $\mu(v)$, the value at which this individual would move from category $\bar{R}-1$ to the highest category $\bar{R}$. Response functions are sampled independently of everything else, which implies $U_i \indep V_i$.

In a first set of simulations, I take $H_i$ to have a standard normal distribution, conditional on $X_i=x$. Note that since the overall location and scale of the happiness distribution is not inherently meaningful, this choice of mean and variance is arbitrary. Next, I suppose that individuals' values of $\ell(V_i)$ are distributed uniformly between $-1$ and $-0.5$, and that $\mu(V_i)$ is independent of $\ell(V_i)$ and drawn uniformly from $[0.5, 1]$. The left panel of Figure \ref{fig:sim_normal} provides a visualization. These choices aim to reflect a world in which while individuals differ e.g. in the point $\mu(V_i)$ at which they would report $R=10$, this threshold for the highest possible category is for all individuals at least above the mean level of happiness in the population.

The simulations generally provide an optimistic picture that $\Pi_{x,x'}/\frac{1}{2}(\Pi_{x}+\Pi_{x'}) \approx 1$ across a wide variety of DGPs, and thus $\left\{\mathbbm{E}[R_i|x',w]-\mathbbm{E}[R_i|x,w]\right\}/\partial_{x_j}\mathbbm{E}[R_i|x,w]$ can be interpreted as close to a ratio of weighted averages of causal effects in those cases.

The table on the right side of Figure \ref{fig:sim_normal} reports $\frac{\Pi_{x,x'}}{\frac{1}{2}(\Pi_{x}+\Pi_{x'})}$ as a function of the number of response categories $\bar{R} \in [2,5,11,100]$, supposing a constant treatment effect $\Delta$ which is varied from $-0.5$ to $5$. Alternatively, the results can be interpreted as reporting conditional analogs of the quantity $\frac{\Pi_{x,x'}}{\frac{1}{2}(\Pi_{x}+\Pi_{x'})}$ among individuals sharing a value of $\Delta_i=h(x',U_i)-h(x,U_i)$, in a setting in which $H_i$ is independent of treatment effects $\Delta_i$, conditional on $X_i=x$.

\begin{figure}[H]
	\small
	\begin{center}
		\begin{tabular}{m{7.5cm} l}
			\includegraphics[width=3in]{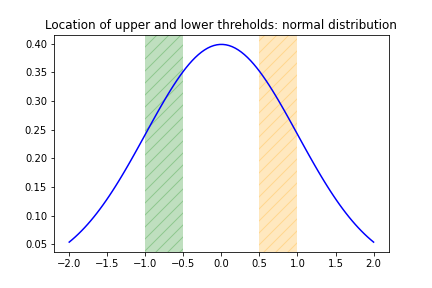}
			&\begin{tabular}{lllll}
\toprule
$\Delta$ &      $\bar{R}$=2 &      $\bar{R}$=5 &     $\bar{R}$=11 &    $\bar{R}$=100 \\
\midrule
    -0.5 & 1.017758 & 1.016489 & 1.018028 & 1.018884 \\
    -0.1 & 1.000335 & 1.000441 & 1.000664 & 1.000809 \\
     0.1 & 1.000837 & 1.000283 & 1.001079 & 1.001052 \\
    0.25 & 1.003905 & 1.005432 & 1.004132 & 1.003522 \\
     0.5 & 1.020549 & 1.019535 & 1.017904 & 1.014529 \\
       1 & 1.060440 & 1.062557 & 1.061607 & 1.051899 \\
       5 & 0.504738 & 0.531706 & 0.544236 & 0.549753 \\
\midrule
    1/NB & 1.867396 & 1.878186 & 1.874115 & 1.873973 \\
\bottomrule
\end{tabular}				
		\end{tabular}
		\vspace{.2cm}
		\caption{$H_i|X_i=x$ is standard normal, and 1000 reporting functions are drawn from $\ell(v) \sim U[-1,1/2]$, $\mu(v) \sim U[1/2,1]$. The left panel depicts the supports of $\ell(v)$ (green) and $\mu(v)$ (yellow) with the density of $H_i$. The right panel reports values of $\Pi_{x,x'}/\frac{1}{2}(\Pi_{x}+\Pi_{x'})$ as a function of $\Delta$ and the number of response categories $\bar{R}$. \label{fig:sim_normal}}
	\end{center}
\end{figure}

Proposition \ref{propheterolinear} implies that as $\bar{R}\rightarrow \infty$, $\Pi_{x,x'}/\frac{1}{2}(\Pi_{x}+\Pi_{x'})$ should lie between 1 and 2, for any values $\Delta$ such that $\ell(V_i) < -|\Delta|$ and $\mu(V_i)>|\Delta|$ for all $V_i$ (so that $f_H(h|x)$ is increasing on the interval $[\ell(V_i)-|\Delta|,\ell(V_i)+|\Delta|]$, and analogously for $\ell(v)$). This is true for all of the values reported in Figure \ref{fig:sim_normal}, aside from $\Delta=1$ and $\Delta=5$. In all but the case of $\Delta=5$, $\Pi_{x,x'}/\frac{1}{2}(\Pi_{x}+\Pi_{x'})$ is in fact quite close to unity, well within the refined bounds $[1,1/NB]$ which holds under the variance restriction in Proposition \ref{propheterolinear}, where $NB=P(0 < R_i < R|X=x)$ is the ``non-bunching'' probability.

With the exception of $\Delta=5$, the standard-normal DGP reported in Figure \ref{fig:sim_normal} provides an optimistic picture that $\left\{\mathbbm{E}[R_i|X_i=x']-\mathbbm{E}[R_i|X_i=x]\right\}/\partial_{x_j}\mathbbm{E}[R_i|X_i=x]$ uncovers something close to a ratio of weighted averages of causal effects, i.e. $\beta_1/\beta_2$ in the case described by Equation (\ref{eq:linearmodelcompare}). In this case, results do not differ substantially whether the number of response categories is small (e.g. $\bar{R}=2$, the case of binary response) or e.g. $\bar{R}=100$. Table \ref{table:sim_nrf_normal} shows that results also do not differ much whether there are few or many different reporting functions present in the population.

The $\Delta=5$ case nevertheless shows that the ratio in (\ref{eq:linearmodelcompare}) may be quite misleading \textit{in principle}, even with this distribution of $H_i$. The $\bar{R}=2$ value of $\Pi_{x,x'}/\frac{1}{2}(\Pi_{x}+\Pi_{x'}) \approx 0.5$ means that the magnitude of $\beta_1$ relative to that of $\beta_2$ would be under-estimated by a factor of 2, when using $x'=(1,x_2)$ and $x=(0,x_2)$ in a linear model $h(x,u) = \beta_1 x_1 +\beta_2 x_2$. On the other hand, it is implausible that binary treatment variable being analyzed would have an effect on happiness that is 5 times the variance of happiness in the population.

While the quantity $\Pi_{x,x'}/\frac{1}{2}(\Pi_{x}+\Pi_{x'})$ averages over the reporting heterogeneity in the population, Figure \ref{fig:sim_normal_dist} disaggregates this by $V_i$. Define $\delta_{\Delta,x,v}:=\frac{\sum_r \bar{f}_H(\Delta,\tau_v(r),x,v)-\sum_r f_H(\tau_v(r)|x,v)}{\sum_r f_H(\tau_v(r)|x,v)}$.
An individual with $X_i=x$ and $V_i=v$ will receive similar weights when using either discrete or continuous variation at $x$ if $\delta_{\Delta,x,v} \approx 0$. Write Eq. (\ref{eqdiscreteintermediate}) as:
\begin{align}
	\mathbbm{E}[R_i|x']-\mathbbm{E}[R_i|x]&
	=\int dF_{V|W}(v|w) \cdot \left(\sum_r f_H(\tau_v(r)|x,v) \right) \cdot \mathbbm{E}[\Delta_i|X_i=x,V_i=v] \nonumber\\
	&\hspace{1.25in} +\int dF_{V|W}(v|w) \cdot \int d\Delta \cdot f_H(\Delta|x,v) \cdot \Delta \cdot \delta_{\Delta,x,v} \nonumber 
\end{align} 

\noindent Figure \ref{fig:sim_normal_dist} reports the distributions of $\frac{\sum_r \bar{f}_H(\Delta,\tau_v(r),x,v)}{\sum_r f_H(\tau_v(r)|x,v)}=1+\delta_{\Delta,x,v}$, across 1000 reporting functions sampled the same as in Figure \ref{fig:sim_normal}. The distributions of $\delta_{\Delta,x,V_i}$ are approximately unimodal in each case, with a variance that tends to increase with the magnitude of $\Delta$.

\begin{figure}[H]
	\small
	\begin{center}
		\includegraphics[width=3.5in]{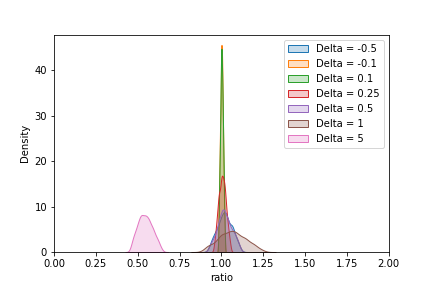}
		\vspace{.2cm}
		\caption{The distribution of $1+\delta_{\Delta,x,V_i}$ across $V_i$ is depicted across alternative values of $\Delta_i$, with $H_i|X_i=x$ standard normal, $\bar{R}=100$, and 1000 reporting functions are drawn from $\ell(v) \sim U[-1,1/2]$, $\mu(v) \sim U[1/2,1]$. \label{fig:sim_normal_dist}}
	\end{center}
\end{figure}

Figures \ref{fig:sim_dblnormal},\ref{fig:sim_uniform} and \ref{fig:sim_lognormal} repeat the exercise of Figure \ref{fig:sim_normal} with alternative distributions assumed for $H_i|X_i=x$. Figure \ref{fig:sim_dblnormal} first relaxes unimodality of the normal distribution by letting $H_i$ be distributed as a mixture of two normals, leading to a ``double-peaked'' shape. Upper and lower thresholds $\mu$ and $\ell$ are sampled from the decreasing and increasing (respectively) portions of this distribution's density. The table shows that $\Pi_{x,x'}/\frac{1}{2}(\Pi_{x}+\Pi_{x'})$ is again close to unity across a wide range of treatment effect sizes, with $\beta_1/\beta_2$ now being over-estimated in the case of an extremely large treatment effect $\Delta=5$. Figure \ref{fig:sim_dblnormal_dist} reports the distributions of $\delta_{\Delta,x,V_i}$, as in Figure \ref{fig:sim_normal_dist}. 

\begin{figure}[H]
	\small
	\begin{center}
		\begin{tabular}{m{7.5cm} l}
			\includegraphics[width=3in]{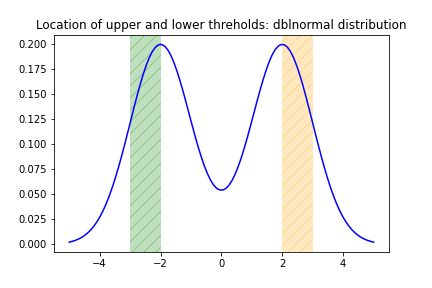}
			&\begin{tabular}{lllll}
\toprule
$\Delta$ &      $\bar{R}$=2 &      $\bar{R}$=5 &     $\bar{R}$=11 &    $\bar{R}$=100 \\
\midrule
    -0.5 & 0.945228 & 0.999726 & 1.002822 & 1.002350 \\
    -0.1 & 0.996596 & 1.000027 & 1.000005 & 1.000072 \\
     0.1 & 0.998344 & 1.000034 & 1.000106 & 1.000211 \\
    0.25 & 0.989510 & 0.999942 & 1.000275 & 1.001007 \\
     0.5 & 0.958623 & 1.000036 & 1.003705 & 1.004092 \\
       1 & 0.901081 & 0.999569 & 1.009504 & 1.013429 \\
       5 & 3.335567 & 1.361256 & 1.225963 & 1.154723 \\
\midrule
    1/NB & 1.470781 & 1.481656 & 1.480439 & 1.483049 \\
\bottomrule
\end{tabular}			
		\end{tabular} \vspace{.2cm}
		\caption{$H_i|X_i=x$ is an equal mixture of $\mathcal{N}(-2,1)$ and $\mathcal{N}(2,1)$, and 1000 reporting functions are drawn from $\ell(v) \sim U[-3,-2]$, $\mu(v) \sim U[2,3]$. The left panel depicts the supports of $\ell(v)$ (green) and $\mu(v)$ (yellow) with the density of $H_i$. The right panel reports values of $\Pi_{x,x'}/\frac{1}{2}(\Pi_{x}+\Pi_{x'})$ as a function of $\Delta$ and the number of response categories $\bar{R}$.\label{fig:sim_dblnormal}}
	\end{center}
\end{figure}

Figure \ref{fig:sim_uniform} instead uses a uniform distribution for $H_i$. This allows us to sample the thresholds $\mu$ and $\ell$ from regions that abut the extremes of the population happiness distribution. Results here are encouraging, except in the cases where $\Delta$ moves a significant portion of the population outside of $[0,1]$ (e.g. $|\Delta| \ge 0.5$. In such cases, there is significant non-overlap between the distributions of $H_i|X_i=x'$ and $H_i|X_i=x$). Notably, $\Pi_{x,x'}/\frac{1}{2}(\Pi_{x}+\Pi_{x'})$ is non-monotonic in the magnitude of $\Delta$, first increasing above unity and then falling much below it, with opposing effects canceling out when $\Delta=1$. Figure \ref{fig:sim_uniform_dist} reports the distributions of $\delta_{\Delta,x,V_i}$, as in Figure \ref{fig:sim_normal_dist}. 

\begin{figure}[H]
	\small
	\begin{center}
		\begin{tabular}{m{7.5cm} l}
			\includegraphics[width=3in]{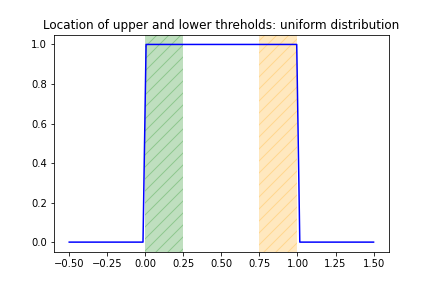}
			&\begin{tabular}{lllll}
\toprule
$\Delta$ &      $\bar{R}$=2 &      $\bar{R}$=5 &     $\bar{R}$=11 &    $\bar{R}$=100 \\
\midrule
    -0.5 & 1.348617 & 1.334445 & 1.332001 & 1.334231 \\
    -0.1 &      1.0 &      1.0 & 1.005277 & 1.010828 \\
     0.1 &      1.0 &      1.0 & 1.004095 & 1.002599 \\
    0.25 &      1.0 & 1.031953 & 1.031220 & 1.035741 \\
     0.5 & 1.320108 & 1.135912 & 1.101178 & 1.081512 \\
       1 & 0.999805 & 0.998796 & 0.995304 & 0.995515 \\
       5 & 0.199587 & 0.199893 & 0.199174 & 0.200451 \\
\midrule
    1/NB & 1.350946 & 1.361273 & 1.355604 & 1.357831 \\
\bottomrule
\end{tabular}				
		\end{tabular} \vspace{.2cm}
		\caption{$H_i|X_i=x$ uniform $[0,1]$, and 1000 reporting functions are drawn from $\ell(v) \sim U[0,1/4]$, $\mu(v) \sim U[3/4,1]$. The left panel depicts the supports of $\ell(v)$ (green) and $\mu(v)$ (yellow) with the density of $H_i$. The right panel reports values of $\Pi_{x,x'}/\frac{1}{2}(\Pi_{x}+\Pi_{x'})$ as a function of $\Delta$ and the number of response categories $\bar{R}$.\label{fig:sim_uniform}}
	\end{center}
\end{figure}

Finally, Figure \ref{fig:sim_lognormal} introduces skewness by letting happiness have a standard log-normal distribution. Corresponding to the long right-tail in the happiness distribution, I take $\mu(V_i)$ to have support over a large range of values relative to $\ell(V_i)$. The results are less optimistic, as compared with the normally distributed case. For $|\Delta| > 0.1$, $\Pi_{x,x'}/\frac{1}{2}(\Pi_{x}+\Pi_{x'})$ differs from unity by more than $10\%$. However, the worst-case $\Delta=5$ is not much worse than in the normally-distributed DGP, with $\Pi_{x,x'}/\frac{1}{2}(\Pi_{x}+\Pi_{x'})$ at least about $0.45$ for all $\bar{R}$. Figure \ref{fig:sim_lognormal_dist} reports the distributions of $\delta_{\Delta,x,V_i}$, as in Figure \ref{fig:sim_normal_dist}. 

\begin{figure}[H]
	\small
	\begin{center}
		\begin{tabular}{m{7.5cm} l}
			\includegraphics[width=3in]{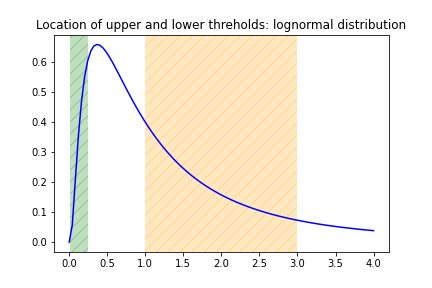}
			&\begin{tabular}{lllll}
\toprule
$\Delta$ &      $\bar{R}$=2 &      $\bar{R}$=5 &     $\bar{R}$=11 &    $\bar{R}$=100 \\
\midrule
    -0.5 & 0.684805 & 0.840033 & 0.891200 & 0.928110 \\
    -0.1 & 0.909362 & 0.931933 & 0.954190 & 0.979041 \\
     0.1 & 1.098319 & 1.075785 & 1.052181 & 1.026204 \\
    0.25 & 1.256482 & 1.182417 & 1.118010 & 1.066896 \\
     0.5 & 1.505123 & 1.278124 & 1.185238 & 1.120666 \\
       1 & 1.643653 & 1.258917 & 1.196697 & 1.137716 \\
       5 & 0.494329 & 0.453087 & 0.439618 & 0.452964 \\
\midrule
    1/NB & 1.467527 & 1.447327 & 1.460737 & 1.445131 \\
\bottomrule
\end{tabular}		
		\end{tabular} \vspace{.2cm}
		\caption{$H_i|X_i=x$ is standard log-normal, and 1000 reporting functions are drawn from $\ell(v) \sim U[1/100,1/4]$, $\mu(v) \sim U[1,3]$. The left panel depicts the supports of $\ell(v)$ (green) and $\mu(v)$ (yellow) with the density of $H_i$. The right panel reports values of $\Pi_{x,x'}/\frac{1}{2}(\Pi_{x}+\Pi_{x'})$ as a function of $\Delta$ and the number of response categories $\bar{R}$. \label{fig:sim_lognormal}}
	\end{center}
\end{figure}

Below I report further results and variations on the DGPs discussed above. Tables \ref{table:sim_nrf_dblnormal}, \ref{table:sim_nrf_uniform} and \ref{table:sim_nrf_lognormal} show that as with the normal DGP, results also do not differ much whether there are few or many different reporting functions present in the population. Taking the lognormal distribution of $H$ as representing the worst-case among the distributions considered, I also consider some variations on the reporting-function DGP used above. Figure \ref{fig:sim_lognormaloverlap} allows the support of $\ell$ and $\mu$ to ``overlap" so that the minimum threshold $\ell$ for some individuals is higher than that \textit{maximum} threshold $\mu$ is for others. Figure \ref{fig:sim_lognormal_fixedends} eliminates all heterogeneity in reporting functions. Figures \ref{fig:sim_lognormal_uniform} and \ref{fig:sim_lognormal_normal} dispense with (heterogeneously) linear reporting, instead sampling the thresholds for a given individual from a specified distribution and sorting them in ascending order to define that individual's reporting function. In all cases, results fall within the range of those presented above.

\begin{table}[H]
	\small
	\begin{center}
		{
			\begin{tabular}{llllll}
\toprule
$\Delta$ & $\#$ r's &        1 &       10 &       11 &     1000 \\
\midrule
    -0.5 & -1.77878 & 1.019952 & 1.006977 & 0.997881 & 1.018028 \\
    -0.1 & -0.36785 & 1.016885 & 0.998345 & 0.996135 & 1.000664 \\
     0.1 & 0.367972 & 1.008979 & 1.004651 & 0.997364 & 1.001079 \\
    0.25 & 0.912569 & 1.012629 & 1.010883 & 1.009681 & 1.004132 \\
     0.5 & 1.779339 & 1.009398 & 1.020753 & 1.035096 & 1.017904 \\
       1 & 3.230443 & 1.136782 & 1.086727 & 1.036715 & 1.061607 \\
       5 & 5.013323 & 0.493566 & 0.579100 & 0.526786 & 0.544236 \\
\midrule
    1/NB &          & 2.369596 & 1.828073 & 1.841236 & 1.873973 \\
\bottomrule
\end{tabular}	
		} \vspace{.2cm}
		\caption{$\Pi_{x,x'}/\frac{1}{2}(\Pi_{x}+\Pi_{x'})$ as a function of $\Delta$ and the number $|\mathcal{V}| \in \{1,10,11,1000\}$ of reporting functions. All cells take $\bar{R}=11$ response categories, and the column labeled $\# r's$ reports the average number of these 11 thresholds crossed by the value of $\Delta$ corresponding to that row, averaged over the distribution of $H_i|X_i=x$. $H_i|X_i=x$ is standard normal, and in all cases thresholds are sampled as depicted in Figure \ref{fig:sim_normal}.\label{table:sim_nrf_normal}}
	\end{center}
\end{table}

\begin{figure}[H]
	\small
	\begin{center}
		\includegraphics[width=3.5in]{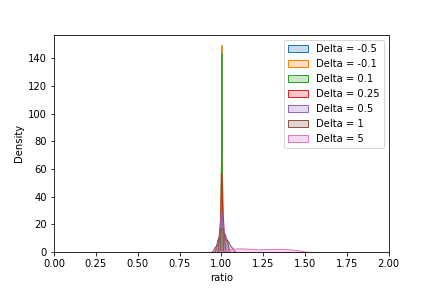}
		\vspace{.2cm}
		\caption{The distribution of $1+\delta_{\Delta,x,V_i}$ across $V_i$ is depicted across alternative values of $\Delta_i$, with $H_i|X_i=x$ an equal mixture of $\mathcal{N}(-2,1)$ and $\mathcal{N}(2,1)$, $\bar{R}=100$, and 1000 reporting functions with thresholds sampled as depicted in Figure \ref{fig:sim_dblnormal}. \label{fig:sim_dblnormal_dist}}
	\end{center}
\end{figure}

\begin{table}[H]
	\small
	\begin{center}
		{
			\begin{tabular}{llllll}
\toprule
$\Delta$ & $\#$ r's &        1 &       10 &       11 &     1000 \\
\midrule
    -0.5 & -0.65973 & 0.987486 & 1.008402 & 1.002409 & 1.002822 \\
    -0.1 & -0.13228 & 0.998206 & 0.997224 & 1.000729 & 1.000005 \\
     0.1 & 0.132622 & 0.998407 & 1.000369 & 0.998732 & 1.000106 \\
    0.25 & 0.330621 & 0.999557 & 1.000884 & 1.001690 & 1.000275 \\
     0.5 & 0.660180 & 0.998541 & 1.001786 & 1.003425 & 1.003705 \\
       1 & 1.296189 & 1.000055 & 1.006245 & 1.020919 & 1.009504 \\
       5 & 4.851162 & 1.408659 & 1.182369 & 1.230598 & 1.225963 \\
\midrule
    1/NB &          & 1.785322 & 1.425248 & 1.507230 & 1.483049 \\
\bottomrule
\end{tabular}	
		} \vspace{.2cm}
		\caption{$\Pi_{x,x'}/\frac{1}{2}(\Pi_{x}+\Pi_{x'})$ as a function of $\Delta$ and the number $|\mathcal{V}| \in \{1,10,11,1000\}$ of reporting functions. All cells take $\bar{R}=11$ response categories, and the column labeled $\# r's$ reports the average number of these 11 thresholds crossed by the value of $\Delta$ corresponding to that row, averaged over the distribution of $H_i|X_i=x$.  $H_i|X_i=x$ is an equal mixture of $\mathcal{N}(-1/2,1)$ and $\mathcal{N}(1/2,1)$, and in all cases thresholds are sampled as depicted in Figure \ref{fig:sim_dblnormal}. \label{table:sim_nrf_dblnormal}}
	\end{center}
\end{table}

\begin{figure}[H]
	\small
	\begin{center}
		\includegraphics[width=3.5in]{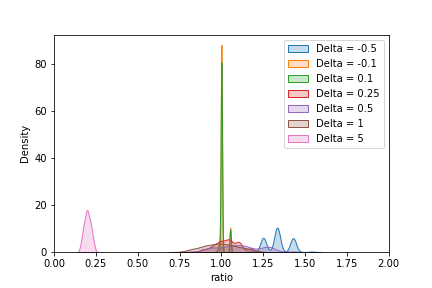}
		\vspace{.2cm}
		\caption{The distribution of $1+\delta_{\Delta,x,V_i}$ across $V_i$ is depicted across alternative values of $\Delta_i$, with $H_i|X_i=x$ uniform on $[0,1]$, $\bar{R}=100$, and 1000 reporting functions with thresholds sampled as depicted in Figure \ref{fig:sim_dblnormal}. \label{fig:sim_uniform_dist}}
	\end{center}
\end{figure}

\begin{table}[H]
	\small
	\begin{center}
		{
			\begin{tabular}{llllll}
\toprule
$\Delta$ & $\#$ r's &        1 &       10 &       11 &     1000 \\
\midrule
    -0.5 &     -5.0 & 1.333333 & 1.351351 & 1.301775 & 1.332001 \\
    -0.1 &     -1.0 & 0.999999 & 1.005025 &      1.0 & 1.005277 \\
     0.1 & 0.999075 & 0.999999 & 1.005025 & 1.004566 & 1.004095 \\
    0.25 & 2.412154 & 1.003553 & 1.015483 & 1.036788 & 1.031220 \\
     0.5 & 4.131344 & 1.063924 & 1.021872 & 1.078504 & 1.101178 \\
       1 & 4.976521 & 1.074331 & 1.067193 & 0.989457 & 0.995304 \\
       5 & 4.979371 & 0.195783 & 0.198180 & 0.212324 & 0.199174 \\
\midrule
    1/NB &          & 1.219642 & 1.400211 & 1.309034 & 1.357831 \\
\bottomrule
\end{tabular}	
		} \vspace{.2cm}
		\caption{$\Pi_{x,x'}/\frac{1}{2}(\Pi_{x}+\Pi_{x'})$ as a function of $\Delta$ and the number $|\mathcal{V}| \in \{1,10,11,1000\}$ of reporting functions. All cells take $\bar{R}=11$ response categories, and the column labeled $\# r's$ reports the average number of these 11 thresholds crossed by the value of $\Delta$ corresponding to that row, averaged over the distribution of $H_i|X_i=x$.  $H_i|X_i=x$ is uniform $[0,1]$, and in all cases thresholds are sampled as depicted in Figure \ref{fig:sim_uniform}.\label{table:sim_nrf_uniform}}
	\end{center}
\end{table}

\begin{figure}[H]
	\small
	\begin{center}
		\includegraphics[width=3.5in]{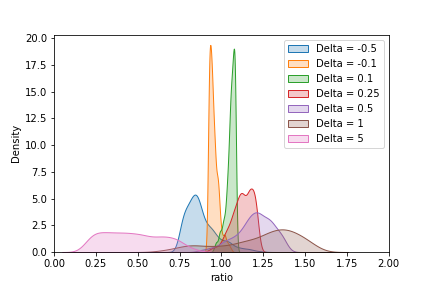}
		\vspace{.2cm}
		\caption{The distribution of $1+\delta_{\Delta,x,V_i}$ across $V_i$ is depicted across alternative values of $\Delta_i$, with $H_i|X_i=x$ a standard log-normal, $\bar{R}=100$, and 1000 reporting functions with thresholds sampled as depicted in Figure \ref{fig:sim_dblnormal}. \label{fig:sim_lognormal_dist}}
	\end{center}
\end{figure}

\begin{table}[H]
	\small
	\begin{center}
		{
			\begin{tabular}{llllll}
\toprule
$\Delta$ & $\#$ r's &        1 &       10 &       11 &     1000 \\
\midrule
    -0.5 & -1.70479 & 0.875985 & 0.850173 & 0.919761 & 0.891200 \\
    -0.1 & -0.39632 & 0.938366 & 0.951939 & 0.949562 & 0.954190 \\
     0.1 & 0.412033 & 1.073444 & 1.045219 & 1.062470 & 1.052181 \\
    0.25 & 1.033962 & 1.120209 & 1.077792 & 1.136499 & 1.118010 \\
     0.5 & 1.980888 & 1.311599 & 1.181132 & 1.178890 & 1.185238 \\
       1 & 3.438382 & 1.122504 & 1.187908 & 1.187129 & 1.196697 \\
       5 & 4.659613 & 0.190009 & 0.427811 & 0.435993 & 0.439618 \\
\midrule
    1/NB &          & 1.799321 & 1.459198 & 1.573983 & 1.445131 \\
\bottomrule
\end{tabular}	
		} \vspace{.2cm}
		\caption{$\Pi_{x,x'}/\frac{1}{2}(\Pi_{x}+\Pi_{x'})$ as a function of $\Delta$ and the number $|\mathcal{V}| \in \{1,10,11,1000\}$ of reporting functions. All cells take $\bar{R}=11$ response categories, and the column labeled $\# r's$ reports the average number of these 11 thresholds crossed by the value of $\Delta$ corresponding to that row, averaged over the distribution of $H_i|X_i=x$. $H_i|X_i=x$ is standard log-normal, and in all cases thresholds are sampled as depicted in Figure \ref{fig:sim_lognormal}.\label{table:sim_nrf_lognormal}}
	\end{center}
\end{table}

\begin{figure}[H]
	\small
	\begin{center}
		\begin{tabular}{m{7.5cm} l}
			\includegraphics[width=3in]{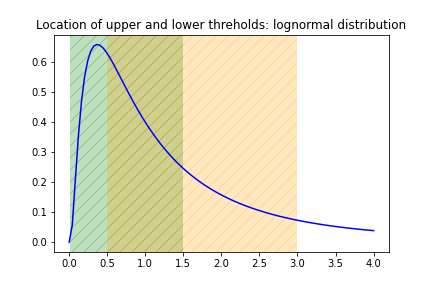}
			&\begin{tabular}{lllll}
\toprule
$\Delta$ &      $\bar{R}$=2 &      $\bar{R}$=5 &     $\bar{R}$=11 &    $\bar{R}$=100 \\
\midrule
    -0.5 & 0.697930 & 0.730712 & 0.746278 & 0.755180 \\
    -0.1 & 0.914523 & 0.921389 & 0.924831 & 0.929371 \\
     0.1 & 1.092894 & 1.085627 & 1.079682 & 1.072299 \\
    0.25 & 1.242410 & 1.222273 & 1.206747 & 1.194887 \\
     0.5 & 1.461151 & 1.413333 & 1.394014 & 1.366348 \\
       1 & 1.729324 & 1.596417 & 1.559468 & 1.534750 \\
       5 & 0.630123 & 0.627392 & 0.608452 & 0.583946 \\
\midrule
    1/NB & 2.333187 & 0.562761 & -235.774 & -1.75283 \\
\bottomrule
\end{tabular}				
		\end{tabular}
		\vspace{.2cm}
		\caption{$H_i|X_i=x$ is standard lognormal, and 1000 reporting functions are drawn from $\ell(v) \sim U[.01,1.5]$, $\mu(v) \sim U[0.5,3]$. Thus, the highest threshold $\mu$ for some individuals is lower than the lowest threshold $\ell$ is for other individuals. The left panel depicts the supports of $\ell(v)$ (green) and $\mu(v)$ (yellow) with the density of $H_i$. The right panel reports values of $\Pi_{x,x'}/\frac{1}{2}(\Pi_{x}+\Pi_{x'})$ as a function of $\Delta$ and the number of response categories $\bar{R}$. \label{fig:sim_lognormaloverlap}}
	\end{center}
\end{figure}

\begin{figure}[H]
	\small
	\begin{center}
		\begin{tabular}{lllll}
\toprule
$\Delta$ &      $\bar{R}$=2 &      $\bar{R}$=5 &     $\bar{R}$=11 &    $\bar{R}$=100 \\
\midrule
    -0.5 & 0.608198 & 0.830388 & 0.862295 & 0.908220 \\
    -0.1 & 0.904246 & 0.924110 & 0.952734 & 0.984422 \\
     0.1 & 1.103991 & 1.083320 & 1.054811 & 1.022938 \\
    0.25 & 1.275106 & 1.205307 & 1.116844 & 1.061898 \\
     0.5 & 1.594517 & 1.299465 & 1.193268 & 1.132321 \\
       1 & 1.955425 & 1.359822 & 1.262187 & 1.213427 \\
       5 & 0.529326 & 0.462219 & 0.454500 & 0.458494 \\
\midrule
    1/NB & 1.341849 & 1.341849 & 1.341849 & 1.341849 \\
\bottomrule
\end{tabular} \vspace{.2cm}
		\caption{$H_i|X_i=x$ is standard log-normal, and all individuals have the same linear reporting function with $\mu(v) = 0.1$ and $\ell(v)=0.2$. Table reports values of $\Pi_{x,x'}/\frac{1}{2}(\Pi_{x}+\Pi_{x'})$ as a function of $\Delta$ and the number of response categories $\bar{R}$. \label{fig:sim_lognormal_fixedends}}
	\end{center}
\end{figure}

\begin{figure}[H]
	\small
	\begin{center}
		\begin{tabular}{m{7.5cm} l}
			\includegraphics[width=3in]{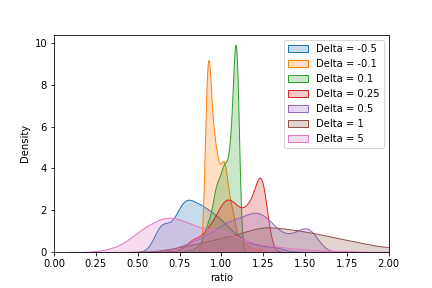}
			&\begin{tabular}{lllll}
\toprule
$\Delta$ &      $\bar{R}$=2 &      $\bar{R}$=5 &     $\bar{R}$=11 &    $\bar{R}$=100 \\
\midrule
    -0.5 & 0.888923 & 0.878441 & 0.873254 & 0.874889 \\
    -0.1 & 0.963337 & 0.981338 & 0.976410 & 0.979380 \\
     0.1 & 1.030927 & 1.026978 & 1.029188 & 1.026451 \\
    0.25 & 1.047773 & 1.071765 & 1.082974 & 1.074864 \\
     0.5 & 1.126755 & 1.175856 & 1.166967 & 1.167957 \\
       1 & 1.364791 & 1.335619 & 1.334358 & 1.325175 \\
       5 & 0.756363 & 0.773198 & 0.776693 & 0.767673 \\
\midrule
    1/NB & 1.341849 & 1.341849 & 1.341849 & 1.341849 \\
\bottomrule
\end{tabular}				
		\end{tabular}
		\vspace{.2cm}
		\caption{ The distribution of $1+\delta_{\Delta,x,V_i}$ across 1000 reporting functions (left), and values of $\Pi_{x,x'}/\frac{1}{2}(\Pi_{x}+\Pi_{x'})$ as a function of $\Delta$ and the number of response categories $\bar{R}$ (right), for $H_i|X_i=x$ following a log normal distribution with all thresholds sampled individually from a uniform distribution on $[0.1,3]$. Thus, thresholds are not equally spaced within individual reporting functions. \label{fig:sim_lognormal_uniform}}
	\end{center}
\end{figure}

\begin{figure}[H]
	\small
	\begin{center}
		\begin{tabular}{m{7.5cm} l}
			\includegraphics[width=3in]{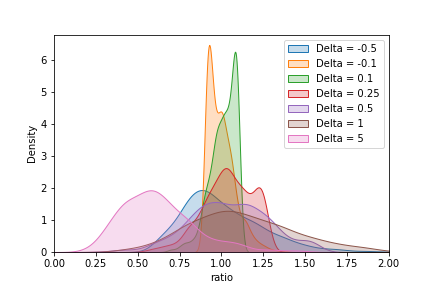}
			&\begin{tabular}{lllll}
\toprule
$\Delta$ &      $\bar{R}$=2 &      $\bar{R}$=5 &     $\bar{R}$=11 &    $\bar{R}$=100 \\
\midrule
    -0.5 & 0.981985 & 0.976757 & 0.980952 & 0.986851 \\
    -0.1 & 0.985272 & 0.990361 & 0.994975 & 0.993520 \\
     0.1 & 0.996470 & 1.006833 & 1.007525 & 1.007262 \\
    0.25 & 1.011348 & 1.028414 & 1.029626 & 1.017586 \\
     0.5 & 1.043327 & 1.044149 & 1.030733 & 1.043710 \\
       1 & 1.062959 & 1.096241 & 1.084608 & 1.075300 \\
       5 & 0.606502 & 0.594470 & 0.585655 & 0.588043 \\
\midrule
    1/NB & 1.341849 & 1.341849 & 1.341849 & 1.341849 \\
\bottomrule
\end{tabular}				
		\end{tabular}
		\vspace{.2cm}
		\caption{ The distribution of $1+\delta_{\Delta,x,V_i}$ across 1000 reporting functions (left), and values of $\Pi_{x,x'}/\frac{1}{2}(\Pi_{x}+\Pi_{x'})$ as a function of $\Delta$ and the number of response categories $\bar{R}$ (right), for $H_i|X_i=x$ following a log normal distribution with all thresholds sampled individually  from a normal distribution with mean $2$ and variance $1$. Thus, thresholds are not equally spaced within individual reporting functions. \label{fig:sim_lognormal_normal}}
	\end{center}
\end{figure}

\section{Implications of the identification results for practice} \label{sec:regression}
From Theorems \ref{propflow} and \ref{propdiscrete}, it is clear that learning from the conditional distribution of responses $R_i$ given variation in $X_i$, one can uncover positive linear combinations of causal effects, but with weights that are not  under the researcher's control. Rather, they depend on individuals' unobserved and heterogeneous reporting functions, and the distribution of underlying happiness $H_i$ near the thresholds at which those individuals move between successive response categories. 

One immediate implication is that if causal effects have the same sign for all individuals, this sign can be identified empirically by mean regression of responses $R_i$ on variation in $X_i$, whether that variation is continuous or discrete.\footnote{When the goal is not causal inference but understanding the joint distribution of $H_i$ and $X_i$, we have from Appendix Eq. \eqref{propquantileidr} that if the sign of $\partial_{x_j} Q_{H|XW}(\alpha|x,w)$ is the same for all $\alpha$, this sign will be reflected in $\partial_{x_j} \mathbbm{E}[R_i|x,w]$. Analogously, with discrete variation in $X_i$, if the conditional distribution of $H_i$ given $X_i=x'$ stochastically dominates that of $X_i=x$ (all conditional on $W_i$), this will be reflected in the sign of the observable conditional mean difference $\mathbbm{E}[R_i|x',w]-\mathbbm{E}[R_i|x,w]$.} The same-sign assumption in fact leads to over-identification restrictions, because $\partial_{x_j}P(R_i \le r|x,w)$ or $P(R_i \le r|x',w)-P(R_i \le r|x,w)$ must have the same sign for all $r$. 

However, researchers often want to be more ambitious and compare the magnitudes of the effects of multiple explanatory variables on $H_i$. The results of the preceding sections show that if $\mathbbm{E}[R_i|X_i=x,W_i=w] = m(x, w)$ is modeled as a fully flexible function of the regressors and estimated nonparametrically, features of the function $m$ can be interpreted causally: derivatives of $m$ uncover positive weighted combinations of partial effects (Section \ref{seccontiniousvariation}) and discrete differences uncover positively-weighted combinations of treatment effects (Section \ref{secdiscrete}). In general, these weights vary not only with regressor $x_j$ but by value of the entire vector $x$, making interpretation somewhat tedious.

Although nonparametric approaches allow one to estimate the entire function $m(x,w)$ consistently, it is difficult to report and interpret an infinite-dimensional object, and the curse of dimensionality looms large with several $X$. One path forward for a continuous $X_1$ is to estimate and report the average of $\partial_y \mathbbm{E}[R_i|X_{1i}=y,X_{-1,i}]$ over the distributions of $y=X_{1i}$ and of the other regressors $X_{-1,i}$. For e.g. a binary regressor $X_2$, one could instead report the average difference  $\mathbbm{E}[R_i|X_{-2,i},X_{2i}=1]-\mathbbm{E}[R_i|X_{-2,i},X_{2i}=0]$ over the distribution of the other regressors $X_{-2,i}$. Such averages can be estimated at the $\sqrt{n}$ rate \citep{ichimuratodd}, and their ratios can still be interpreted in terms of ratios of convex averages of causal effects as in Section \ref{sec:comparing}---the averaging is  now over $x$ as well. In Appendix \ref{sec:ilustrativeexample}, I follow this approach using the estimator of \citet{liracine} (which is implemented in the Stata command \texttt{npregress kernel}) to synthetic data with two explanatory variables. This estimator applies kernel regression techniques to setups in which there may both be continuous and discrete regressors. 

Notwithstanding the above, in practice researchers often instead estimate parsimonious specifications of the function $m$, most frequently applying OLS to linear models of the form of \eqref{eq:ols}. The remainder of this section studies the interpretation of the estimands $\gamma_j$ in Eq. (\ref{eq:ols}) in light of the results of this paper. I focus on the case with no control variables $W_i$ for ease of exposition.

\subsection{Case 1: linear model is correctly specified}
The most straightforward case arises when Eq. (\ref{eq:ols}) is correctly specified in the sense that the conditional expectation function is in fact linear in the $x$, i.e.
\begin{equation} \label{eq:olscorrect}
	\mathbbm{E}[R_i|X_i=x,W_i=w] = \gamma_1 x_1 + \dots + \gamma_J x_J + \lambda^T w
\end{equation}
or equivalently that $\mathbbm{E}[\epsilon_i|X_i=x,W_i=w] = 0$ in (\ref{eq:ols}). It should be emphasized that a linear model for causal effects: $h(x,u) = x^T\beta+u$, does \textit{not} imply that a linear relationship holds between $R_i$ and $X_i$ (conditional on $W_i$), given non-linearity in the response functions. However, whether or not $\mathbbm{E}[R_i|X_i=x,W_i=w]$ exhibits a linear functional form can be examined empirically, given that $(R_i,X_i,W_i)$ are all observable. 

In order to keep notation to a minimum, I for the remainder of this section assume that no control variables $W_i$ are needed for EXOG to hold. In the context of Eq. (\ref{eq:olscorrect}), consider comparing the regression coefficient $\gamma_1$ with $\gamma_2$, if $X_1$ and $X_2$ are both continuously distributed. Since each $\gamma_j$ is then equal to
$ \partial_{x_j}\mathbbm{E}[R_i|X_i=x]$, the ratio $\gamma_2/\gamma_1$ recovers a ratio of two convex averages of causal effects by Eq. (\ref{eqexpratio}). If in addition to Eq. (\ref{eq:olscorrect}), the structural function is linear with $h(x,u)=x^T\beta+u$, then $\gamma_2/\gamma_1 = \beta_2/\beta_1$. 

Now consider comparing the coefficients for a continuously distributed $X_1$ (e.g. income) and a binary $X_2$ (e.g. an indicator for being married). For any values $x_2, x_3 \dots x_{J}$, note that: \normalsize
$$\frac{\gamma_2}{\gamma_1}=\frac{\mathbbm{E}[R_i|X_{1i}=x_1,X_{2i}=1,\dots X_{Ji}=x_{J}]-\mathbbm{E}[R_i|X_{1i}=x_1,X_{2i}=0,\dots X_{Ji}=x_{J}]}{\frac{1}{2}\partial_{x_1}\mathbbm{E}[R_i|X_{1i}=x_1,X_{2i}=1,\dots X_{Ji}=x_{J}] + \frac{1}{2}\partial_{x_1}\mathbbm{E}[R_i|X_{1i}=x_1,X_{2i}=0,\dots X_{Ji}=x_{J}]}$$
\large
If a linear model again holds \textit{both} for $\mathbbm{E}[R_i|X_i=x]$ and for the structural function $h(x,u)=x^T\beta+u$, then under the assumptions of Proposition \ref{propheterolinear}, Eq. (\ref{eq:linearmodelcompare2}) with $x'=(x_1,1, \dots)$ and $x=(x_1,0, \dots)$ implies that $\frac{\gamma_2}{\gamma_1}$ estimates $\frac{\beta_2}{\beta_1}$ up to a scaling factor $\theta$ that lies between 1 and 2. More generally, if linearity holds only for $\mathbbm{E}[R_i|X_i=x]$ but not necessarily for $h(x,u)$, and the assumptions of Proposition \ref{propheterolinear} are still satisfied, then $\gamma_2/\gamma_1$ identifies a ratio of two weighted averages of causal effects, again up to a factor $\theta \in [1,2]$, where the weights aggregate to one both the numerator and the denominator. In the numerator, the averaging is over $\Delta_i = h(x',U_i)-h(x,U_i)$ among units with $X_i=x$ while in the denominator it is over both $\mathbbm{E}[\partial_{x_1} h(x,U_i)|H_i \in \tau_{V_i},X_i=x]$ and $\mathbbm{E}[\partial_{x_1} h(x',U_i)|H_i \in \tau_{V_i},X_i=x']$, cf. Eq. (\ref{eqtauV}).

Finally, suppose that we wish to compare regression coefficients for two discrete variables $X_1$ and $X_2$. For simplicity, suppose that they are both binary. Then, :
$$\frac{\gamma_2}{\gamma_1}=\frac{\mathbbm{E}[R_i|X_{1i}=x_1,X_{2i}=1,\dots X_{Ji}=x_{J}]-\mathbbm{E}[R_i|X_{1i}=x_1,X_{2i}=0,\dots X_{Ji}=x_{J}]}{\mathbbm{E}[R_i|X_{1i}=1,X_{2i}=x_2,\dots X_{Ji}=x_{J}]-\mathbbm{E}[R_i|X_{1i}=0,X_{2i}=x_2,\dots X_{Ji}=x_{J}]}$$
for any $x=(x_1,x_2, \dots, x_{dx})$. To analyze this case we can apply Proposition \ref{propheterolinear} twice while using a continuously distributed third variable $X_3$ as a common comparison. This implies that under the assumptions of Proposition \ref{propheterolinear}, in the dense response limit $\gamma_2/\gamma_1$ identifies a ratio of two weighted averages of causal effects (with respect to $x_2$ in the numerator, and $x_1$ in the denominator) up to a factor that lies between $1/2$ and $2$.\footnote{To see this, let $\gamma_3 = \partial_{x_3}\mathbbm{E}[R_i|x]$, and write $\gamma_2/\gamma_1 = \gamma_2/\gamma_3 \cdot \gamma_3/\gamma_1$. Let $\tilde{\beta}_1$, $\tilde{\beta}_2$, and $\tilde{\beta}_3$ denote the convex combinations of causal effects associated with $\gamma_1$, $\gamma_2$ and $\gamma_3$ (cf Propositions \ref{propflow} and \ref{propdiscrete} after normalizing the weights). By Proposition \ref{propheterolinear} $\gamma_1/\gamma_3 \stackrel{R}{\rightarrow} \theta_1 \cdot \tilde{\beta_1}/\tilde{\beta_3}$ and $\gamma_2/\gamma_3 \stackrel{R}{\rightarrow} \theta_2 \cdot \tilde{\beta_2}/\tilde{\beta_3}$, where $\theta_1, \theta_2 \in [1,2]$. Thus,  $\gamma_2/\gamma_1 \stackrel{R}{\rightarrow} \theta_2/\theta_1 \cdot \tilde{\beta_2}/\tilde{\beta_1}$.}

\subsection{Case 2: misspecified regression function}

When Eq. (\ref{eq:ols}) is misspecified, the estimands of $\gamma_j$ in (\ref{eq:ols}) remain well-defined as population linear projection coefficients, but these do not always bear a straightforward relationship to the features of the conditional expectation function $m(x,w) = \mathbbm{E}[R_i|X_i=x,W_i=w]$ of interest. Nevertheless, some existing results on linear regression are useful to gain some intuition.

One case in which the estimand of Eq. (\ref{eq:ols}) remains causally interpretable without assuming linearity of the expectation (\ref{eq:olscorrect}) occurs when we have a single continuously distributed $X_i$ and no control variables. In this case, Eq. (\ref{eq:ols}) amounts to simple linear regression: $R_i = \gamma_0 + \gamma_1 X_i + \epsilon_i$. \citet{Yitzhaki1996} shows that the regression coefficient $\gamma_1 = \frac{Cov(R,X)}{Var(X)}$ can then be written as a weighted average over the local derivative of $\mathbbm{E}[R_i|X_i=x]$ even if it is non-linear:
$$ \gamma_1 = \int w(x) \cdot \frac{d}{dx} \mathbbm{E}[R_i|X_i=x] \cdot dx$$
where $w(x):= \frac{1}{Var(X)}\int_{-\infty}^x f_X(t)(t-\mathbbm{E}[X_i])dt$ is a positive function that integrates to unity, with $f_X$ denoting the density of $X_i$. By Theorem \ref{propflow}, $\gamma_1$ thus still captures a positively weighted combination of causal effects $\partial_{x} h(x,U_i)$, where the averaging is now also over $x$. If all units in the population have the same sign of $\partial_{x} h(x,U_i)$, then this sign can be recovered as that of $\gamma_1$. \citet{MHE} extend the above expression to a case with covariates: if $\mathbbm{E}[X_1|W_i]$ is linear in $W_i$, then $\gamma_1$ can be written as $\mathbbm{E}[\gamma_1(W_i)]$, where the quantities that define $\gamma_1(w)$ are analogous to the above but condition on $W_i=w$, with weights again integrating to unity.\footnote{The linearity assumption is not restrictive if $W_i$ consists of indicators for an exhaustive set of covariate cells, a so-called ``fully-saturated'' regression.} An analogous expression can also be derived for a setting with a binary $X_1$ \citep{Angrist1998,MHE} or an ordered $X_1$ \citep{angristkrreuger}. Thus with a single treatment variable $X_1$ of any type, a linear regression equation (\ref{eq:ols}) with fully saturated controls simply re-averages the causal effects of $X_1$ on $H_i$ derived in this paper over a second set of positive weights.

Unfortunately, the results mentioned above do not carry over to the general setting with multiple treatment variables $X_1 \dots X_J$ and controls estimated by Eq. (\ref{eq:ols}). \citet{contaminationbias} show that regressions like (\ref{eq:ols}) with controls $W$ can be subject to ``contamination bias'', in which the coefficient $\gamma_1$ on $X_1$ includes not only effects from $X_1$, but also effects from the other treatments $X_2 \dots X_{J}$. In other words, $\gamma_1$ does not cleanly separate variation in $X_1$ from variation in the other treatments.

Since \citet{contaminationbias} consider a standard setup in which the outcome variable of interest is directly observed, we can facilitate the connection by phrasing our examination of regression (\ref{eq:ols}) in terms of the causal effects of $X_i$ on $R_i$ (rather than on $H_i$).\footnote{Note that given any consistent estimator for the average treatment effect of some covariate contrast $x,x'$ (differing only in the first $J$ components) on $R$, one can interpret this using the methods of the present paper by translating it back into a statement about conditional means, since $\mathbbm{E}[R_i(x')-R_i(x)] = \mathbbm{E}[R_i|X_i=x']-\mathbbm{E}[R_i|X_i=x]$.} This interpretation is justified under assumption EXOG, because we can define potential outcomes $R(x)$ with respect to $x \in \mathcal{X}$ as $R_i(x) = R(H(x,U_i),V_i)$ with $\{R_i(x) \indep X_{ji}\} | W_i$ for $j=1\dots J$. Unfortunately, contamination bias is possible even under fairly optimistic linearity assumptions, for example that $\mathbbm{E}[R_i(x)-R_i(x_0)|W_i=w] = x'\beta(w) + \lambda'w$ for some fixed reference treatment $x_0 \in \mathcal{X}$, and vectors $\beta(w)$ and $\lambda$, i.e. conditional-on-$W$ average treatment effects are linear in all treatment variables. If the per-unit conditional effects $\beta(w)$ vary with $w$, then e.g. the estimand $\gamma_1$ may not capture a clean average of the $\beta_1(w)$ but instead include a second term that depends on the $\beta_2 \dots \beta_J$. The threat of contamination bias is not in any way specific to the use of subjective outcome variables, but may be particularly pernicious in this context given the motivation to compare magnitudes across regressors. \citet{contaminationbias} provide detail on possible solutions.

\subsection{An illustrative simulation example} \label{sec:ilustrativeexample}

\subsubsection{Setup}
Consider a population in which happiness is determined as by three things: i) one's income $X_{1i}$ (measured in thousands of dollars); ii) whether they are married $X_{2i} \in \{0,1\}$; and iii) an idiosyncratic error term $U_i$, according to a linear causal relationship:
\begin{equation} \label{eq:happinesseq}
	H_i = \beta_1 \ln(X_{1i}) + \beta_2 X_{2i}+U_i
\end{equation}
For the sake of illustration, we suppose that in this world money does not buy happiness: in fact, it has a slight negative effect with $\beta_1 = -0.1$. However, marriage does come with a substantial benefit to happiness: $\beta_2=1$.

First, suppose that life satisfaction is measured by a binary question in which $R_i=1$ indicates that individual $i$ responded ``yes'' and $R_i=0$ that they responded ``no'' to the question: ``All things considered, are you satisfied with your life at present?'' Given Assumption MONO, we know by Lemma \ref{prophonest} that individual reporting functions can be written
$$ R_i = \mathbbm{1}(H_i > \tau_{V_i})$$
Suppose that $V_i$ takes two values in the population. Optimistic Reporters, indicated by $V_i=1$ have a threshold $\tau_1 = -1$, and Pessimistic Reporters, indicated by $V_i=0$, have $\tau_0 = 0$. While a Pessimistic Reporter requires $H_i$ to be positive to indicate they are satisfied with life, Optimistic Reporters only require $H_i > -1$ to report that they are satisfied with life.

In line with Assumption EXOG, we will eventually assume that $(U_i,V_i) \indep (X_{1i},X_{2i})$, i.e. income and marital status are as good as randomly assigned. However, to investigate departures from this assumption, I introduce a parameter $\rho$ that governs the correlation between income and ``reporting optimism'' $V_i$. In particular, the probability of being an Optimistic Reporter as a function of income is: $\mathbbm{E}[V_i|X_{1i}=y]=\Phi(\rho\cdot \ln(y/50))$, where $\Phi$ is the normal CDF function. Thus if $\rho>0$ the proportion of Optimistic Reporters is increasing with income: all among the richest are are, while none among the poorest are Optimistic Reporters.

I round out the remaining aspects of the DGP as follows:
\begin{itemize}
	\item The distribution of income is log-normal: $ln(X_{1i}/50) \sim \mathcal{N}(0,1)$ trimmed to the range $[20,200]$, with $X_{1i}$ in thousands of dollars. Trimming incomes below 20 avoids $H_i$ tending towards infinity as $X_1 \downarrow 0$.
	\item Half of all individuals are married $\mathbbm{E}[X_{2i}]=0.5$, with $X_{2i} \indep (U_i,V_i,X_{1i})$
	\item $U_i \sim \mathcal{N}(0,1)$, and $U_{i} \indep (V_i,X_{i})$
\end{itemize}
Figure \ref{fig:hdist} shows kernel density estimates of the resulting distribution of $H_i$ in the population, computed from a sample of $N=10,000$. The threshold for Pessimistic Reporters $\tau_0=0$ (blue, dashed vertical line) and for Optimistic reporters $\tau_1=-1$ (orange, dash-dot vertical line) fall close to the center of the distribution.

\begin{figure}[H]
	\small
	\begin{center}
		\includegraphics[width=3in]{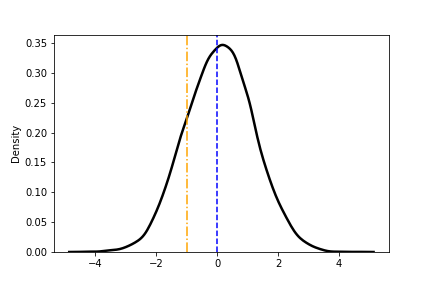} \vspace{.2cm}
		\caption{Distribution of $H_i$ in the illustrative example. Vertical lines indicate the life satisfaction thresholds $\tau_0 = 0$ (blue, dashed) and $\tau_1=-1$ (orange, dash-dot) for Pessimistic and Optimistic Reporters, respectively. \label{fig:hdist}}
	\end{center}
\end{figure}

\subsubsection{The importance of exogeneity} \label{sec:importanceexog}
To illustrate the importance of Assumption EXOG, let us first consider ourselves in the shoes of an econometrician facing a DGP with $\rho =1$. Since Optimistic Reporters tend to have higher incomes, this introduces a mechanical positive correlation between income and reported well-being, depicted in the left panel of Figure \ref{fig:illustrative_rho1}.\footnote{\label{fn:cef}This correlation can be computed explicitly: by the law of iterated expectations, we have that
	\begin{align*}
		\mathbbm{E}[R_i|X_{1i}=y, X_{2i}=m] &= P(V_i = 0|X_{1i}=y)\cdot P(\beta_1ln(y)+\beta_2m+U_i > \tau_0)+P(V_i = 1|X_{1i}=y)\cdot P(\beta_1ln(y)+\beta_2m+U_i > \tau_1)\\
		&= \Phi(\rho\cdot \ln(y/50))\cdot \Phi(\beta_1ln(y)+\beta_2m-\tau_1)+(1-\Phi(\rho\cdot \ln(y/50)))\cdot \Phi(\beta_1ln(y)+\beta_2m-\tau_0)
\end{align*}} A regression of $R$ on $\ln(X_1)$ and $X_2$ picks up this spurious correlation between $R$ and $X_1$ that arises from reporting heterogeneity: the coefficient on log income reported in Column (1) is positive, despite $\beta < 0$. The ratio of estimates $\hat{\beta}_{Married}/\hat{\beta}_{LogIncome}$ evaluates to $4.04$, having opposite sign as the true value of $\beta_2/\beta_1 = -10$. Column (3) shows that if the econometrician did have access to direct observations of $H$, a simple OLS regression estimates $\beta_1$ and $\beta_2$, and hence their ratio, well--in line with Eq. (\ref{eq:happinesseq}).

\begin{figure}[H]
	\small
	\begin{center}
		\begin{tabular}{m{6cm} l}
			\includegraphics[width=7cm]{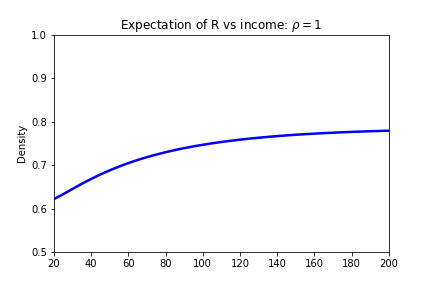}
			&{
\def\sym#1{\ifmmode^{#1}\else\(^{#1}\)\fi}
\begin{tabular}{l*{3}{c}}
\hline\hline
            &\multicolumn{1}{c}{(1)}&\multicolumn{1}{c}{(2)}&\multicolumn{1}{c}{(3)}\\
            &\multicolumn{1}{c}{R}&\multicolumn{1}{c}{R}&\multicolumn{1}{c}{H}\\
\hline
Log Income  &      0.0765\sym{***}&                     &      -0.106\sym{***}\\
            &   (0.00699)         &                     &    (0.0161)         \\
[1em]
Income      &                     &     0.00155\sym{***}&                     \\
            &                     &  (0.000191)         &                     \\
[1em]
Married     &       0.309\sym{***}&       0.310\sym{***}&       1.015\sym{***}\\
            &   (0.00861)         &   (0.00864)         &    (0.0199)         \\
[1em]
Constant    &       0.232\sym{***}&       0.698\sym{***}&      0.0162         \\
            &    (0.0295)         &   (0.00482)         &    (0.0672)         \\
\hline
$\hat{\beta}_{Married}/\hat{\beta}_{LogIncome}$&        4.04         &        3.69         &       -9.55         \\
Estimator   &         OLS         &Local Linear         &         OLS         \\
N           &       10000         &       10000         &       10000         \\
\hline\hline
\multicolumn{4}{l}{\footnotesize Standard errors in parentheses}\\
\multicolumn{4}{l}{\footnotesize \sym{*} \(p<0.05\), \sym{**} \(p<0.01\), \sym{***} \(p<0.001\)}\\
\end{tabular}
}
			
		\end{tabular}
		\vspace{.2cm}
		\caption{$\rho=1$ case. Left panel depicts the conditional expectation function $\mathbbm{E}\left\{\mathbbm{E}[R_i|X_{1i}=y,X_{2i}]\right\}$ as a function of $y$ (calculated from the known DGP as described in footnote \ref{fn:cef}), when $\rho=1$. Regression results (right panel) of $R_i$ on $X_i$ reflect this spurious positive association between income and reported satisfaction, estimated on a simulated dataset of $10,000$ observations. Column (1) uses OLS of $R$ on log-income and marriage, while Column (2) nonparametrically estimates the mean marginal effect of income and the mean effect of Marital (see text for details). Column (3) reports an (infeasible) direct regression of $H_i$ on log-income and marriage, recovering consistent estimates of the true parameters $\beta_1=-0.1$ and $\beta_2=1$. \label{fig:illustrative_rho1}}
	\end{center}
\end{figure}

Column (2) of Figure \ref{fig:illustrative_rho1} shows that $\hat{\beta}_{Married}/\hat{\beta}_{LogIncome}$ getting the wrong sign in Column (1) is not due to functional form misspecification in the OLS estimates. A nonparametric regression of $R$ on income and marital status again captures a positive ratio, and of similar magnitude. Specifically, Column (2) reports the average derivative of $\mathbbm{E}[R_i|X_{1i}=y,X_{2i}]$ with respect to $y$ over the distribution of $X_{1i}$ as the ``coefficient'' for income, and estimates the average difference  $\mathbbm{E}[R_i|X_{1i},X_{2i}=1]-\mathbbm{E}[R_i|X_{1i},X_{2i}=0]$ as the ``coefficient'' for marital status. This is implemented using the kernel estimator of \citet{liracine}, with bandwidth chosen by cross-validation. Standard errors are calculated using 500 bootstrap replications. I report $\hat{\beta}_{Married}/\hat{\beta}_{LogIncome}$ computed by averaging the local ratio of effects across the empirical distribution of $X_i$: $\frac{1}{N} \sum_{i=1}^N \frac{\mathbbm{E}[R_i|X_{1i},X_{2i}=1]-\mathbbm{E}[R_i|X_{1i},X_{2i}=0]}{\left.\partial_{y}\mathbbm{E}[R_i|y,X_{2i}]\right|_{y=X_{1i}}}$.

\subsubsection{The importance of correct functional form}
Figure \ref{fig:illustrative_rho0} turns to the case of $\rho=0$, in which Assumption EXOG is satisfied and hence the results of this paper apply. In the left panel, we see that the conditional expectation of $R$ with respect to income is now decreasing in income, in line with the negative sign on $\beta_1$. The OLS regression of reported satisfaction on log income and marriage now yields $\hat{\beta}_{Married}/\hat{\beta}_{LogIncome}=-7.91$, which has the correct sign but undershoots the true value of $-10$. This could be due to $\Pi_{x,x'}/\Pi_{x} < 1$, in the parlance of Section \ref{sec:comparing}, but also could arise from misspecification of the functional form of $\mathbbm{E}[R|X_{1i},X_{2i}]$. Column (2) again implements the local linear regression method described above, returning estimates $\hat{\beta}_{Married}/\hat{\beta}_{LogIncome}=-12.13$. These estimates are very similar to those obtained by an OLS regression of $R$ on income (not in logs) as well as marital status. This underscores the fact that functional form assumptions regarding the effects of $X$ on $H$ to not translate unchanged into features of the observable correlations between $X$ and $R$. In this case the causal relationship is linear in log income, while the observable one is linear in income. This distinction matters quantitatively in this case for assessing the relative contributions of income and marriage to well-being. 

\begin{figure}[H]
	\small
	\begin{center}
		\begin{tabular}{m{6cm} l}
			\includegraphics[width=7cm]{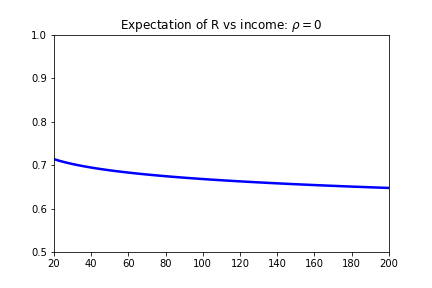}
			&{
\def\sym#1{\ifmmode^{#1}\else\(^{#1}\)\fi}
\begin{tabular}{l*{3}{c}}
\hline\hline
            &\multicolumn{1}{c}{(1)}&\multicolumn{1}{c}{(2)}&\multicolumn{1}{c}{(3)}\\
            &\multicolumn{1}{c}{R}&\multicolumn{1}{c}{R}&\multicolumn{1}{c}{R}\\
\hline
Log Income  &     -0.0386\sym{***}&                     &                     \\
            &   (0.00726)         &                     &                     \\
[1em]
Income      &                     &   -0.000522\sym{***}&   -0.000520\sym{***}\\
            &                     &  (0.000103)         &  (0.000102)         \\
[1em]
Married     &       0.305\sym{***}&       0.306\sym{***}&       0.306\sym{***}\\
            &   (0.00879)         &   (0.00892)         &   (0.00879)         \\
[1em]
Constant    &       0.688\sym{***}&       0.684\sym{***}&       0.567\sym{***}\\
            &    (0.0304)         &   (0.00483)         &    (0.0100)         \\
\hline
$\hat{\beta}_{Married}/\hat{\beta}_{LogIncome}$&       -7.91         &      -12.13         &      -11.87         \\
Estimator   &         OLS         &Local Linear         &         OLS         \\
N           &       10000         &       10000         &       10000         \\
\hline\hline
\multicolumn{4}{l}{\footnotesize Standard errors in parentheses}\\
\multicolumn{4}{l}{\footnotesize \sym{*} \(p<0.05\), \sym{**} \(p<0.01\), \sym{***} \(p<0.001\)}\\
\end{tabular}
}
			
		\end{tabular}
		\vspace{.2cm}
		\caption{$\rho=0$ case. Left panel depicts the conditional expectation function $\mathbbm{E}\left\{\mathbbm{E}[R_i|X_{1i}=y,X_{2i}]\right\}$ as a function of $y$ as a function of $y$ (calculated from the known DGP as described in footnote \ref{fn:cef}), when $\rho=0$. Now assumption EXOG is satisfied and the observable relationship between $R$ and income reflects sign of the true negative effect $\beta_1$. Right panel reports regression results of $R_i$ on $X_i$ on a simulated dataset of $10,000$ observations. Column (1) uses OLS on log income and married, and Column (2) again uses the nonparametric estimator described in the text for Figure \ref{fig:illustrative_rho1}. Column (3) compares this against OLS using income rather than the log of income as a regressor. For Column (3) $\hat{\beta}_{LogIncome}$ is computed as $\hat{\beta}_{Income}\cdot \hat{\mathbbm{E}}[1/X_{1i}]$). \label{fig:illustrative_rho0}}
	\end{center}
\end{figure}

\subsubsection{The effect of the number of response categories}
While the two DGPs reported above consider a binary $R$ for simplicity, Figure \ref{fig:illustrative_rho011} reports the $\rho=0$ case with an 11-point scale for $R$. The DGP is unchanged from above except that now Pessimistic Reporters have linear reporting functions with 
$$ \tau_0(r) = -5 + r$$
while Optimistic Reporters have all thresholds shifted down by $1$ relative to the Pessimists:
$$ \tau_1(r) = -6 + r$$
Figure \ref{fig:illustrative_rho011} again compares a linear regression of $R$ on log-income and marital status (1) to a nonparametric (2) and linear regression (3) of $R$ on income and marital status. In Column (1), the estimated ratio $\hat{\beta}_{Married}/\hat{\beta}_{LogIncome}$ is close in magnitude to $-10$ while the estimated ratios in Columns (2) and (3) are somewhat larger. This suggests that the Column (1) estimate being close to the truth is a coincidence of functional-form misspecification offsetting $\Pi_{x,x'}/\Pi_{x} > 1$ in line with Theorem \ref{propheterolinear}. Indeed, comparing Columns (2) and (3) the CEF of $R$ appears to again be approximately linear in income.

\begin{figure}[H]
	\small
	\begin{center}
		\begin{tabular}{m{6cm} l}
			\includegraphics[width=7cm]{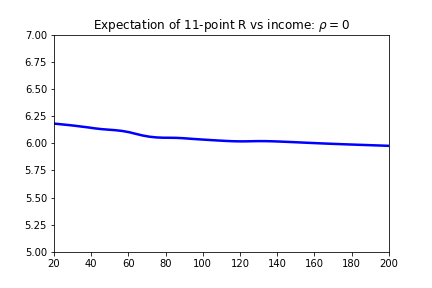}
			&{
\def\sym#1{\ifmmode^{#1}\else\(^{#1}\)\fi}
\begin{tabular}{l*{3}{c}}
\hline\hline
            &\multicolumn{1}{c}{(1)}&\multicolumn{1}{c}{(2)}&\multicolumn{1}{c}{(3)}\\
            &\multicolumn{1}{c}{R}&\multicolumn{1}{c}{R}&\multicolumn{1}{c}{R}\\
\hline
Log Income  &      -0.108\sym{***}&                     &                     \\
            &    (0.0194)         &                     &                     \\
[1em]
Income      &                     &    -0.00144\sym{***}&    -0.00145\sym{***}\\
            &                     &  (0.000264)         &  (0.000271)         \\
[1em]
Married     &       1.021\sym{***}&       1.021\sym{***}&       1.021\sym{***}\\
            &    (0.0230)         &    (0.0221)         &    (0.0230)         \\
[1em]
Constant    &       6.018\sym{***}&       6.094\sym{***}&       5.682\sym{***}\\
            &    (0.0808)         &    (0.0129)         &    (0.0252)         \\
\hline
$\hat{\beta}_{Married}/\hat{\beta}_{LogIncome}$&       -9.50         &      -14.35         &      -14.25         \\
Estimator   &         OLS         &Local Linear         &         OLS         \\
N           &       10000         &       10000         &       10000         \\
\hline\hline
\multicolumn{4}{l}{\footnotesize Standard errors in parentheses}\\
\multicolumn{4}{l}{\footnotesize \sym{*} \(p<0.05\), \sym{**} \(p<0.01\), \sym{***} \(p<0.001\)}\\
\end{tabular}
}
			
		\end{tabular}
		\vspace{.2cm}
		\caption{$\rho=0$ case with an eleven point response scale. Left panel depicts a lowess regression of $R$ on income, in the simulated dataset of $10,000$ observations.  \label{fig:illustrative_rho011}}
	\end{center}
\end{figure}

\section{Proofs} \label{sec:proofs}

\subsection{Proof of Lemma \ref{prophonest}}
Fix any $v$. First we show that if (\ref{eq:iff}) holds for all $r$ then Assumption MONO holds. Indeed, suppose that for some $h' > h$ we had $r(h',v) < r(h,v)$. Substituting $r=r(h',v)$ into (\ref{eq:iff}), we would then have that $r(h,v)> r(h',v) \implies h > \tau_v(r(h',v))$ and hence that $h' > \tau_v(r(h',v))$ since $h'>h$. But $h' > \tau_v(r(h',v))$ violates the definition of $\tau_v$, since then $h' > \sup\{h \in \mathcal{H}: r(h,v) \le r(h',v)\} \ge h'$. 

Left-continuity of $r$ holds by considering any increasing sequence of $h$ converging to $\tau_v(r)$, i.e. I show that $\lim_{h \uparrow \tau_v(r)} r(h,v) = r(\tau_v(r),v)$. First, note that $\lim_{h \uparrow \tau_v(r)} r(h,v) > r(\tau_v(r),v)$ would violate weak monotonicity of $r$. Suppose instead that $\lim_{h \uparrow \tau_v(r)} r(h,v) = r^*$ where $r^* < r(\tau_v(r),v)$. This limit exists by the increasing property of $r$. It must then be the case that $\tau_v(r^*) = \tau_v(r)$. To see this, consider the two alternatives. For $\tau_v(r^*) < \tau_v(r)$, there would need to exist an $h^*$ such that $r(h^*,v)>r^*$ but $h^* < \tau_v(r)$. This would violate $\lim_{h \uparrow \tau_v(r)} r(h,v) = r^*$ given that $r$ is increasing. Suppose instead that $\tau_v(r^*) > \tau_v(r)$. Then there would need to exist an $h^*$ such that $r(h^*,v)>r$ but $h^* < \tau_v(r)$. But $h^* < \tau_v(r)$ implies that $r(h^*,v) \le r$ given that $r$ is increasing. Now, given that $\tau_v(r^*) = \tau_v(r)$, $r^* < r(\tau_v(r),v)$ would violate (\ref{eq:iff}) for $h=\tau_v(r^*)$, because $r(h,v)>r \implies h > \tau_v(r)$. 

Now we will show that if Assumption MONO holds then (\ref{eq:iff}) is satisfied for all $v,r$. First, note that $\tau_v(r)$ is weakly increasing in $r$, and thus $r(h,v) \le r \implies \tau_v(r(h,v)) \le \tau_v(r) \implies  h \le \tau_{v}(r)$ since by the definition of $\tau_v(r)$: $h \le \tau_v(r(h,v))$ for any $h$. Thus we can establish the $\implies$ direction of (\ref{eq:iff}), without even invoking Assumption MONO. In the other direction, assume that for some $r$ and $h$, $h \le \tau_v(r)$ but $r(h,v) > r$. By the increasing property of MONO: $h \le \tau_{v}(r) \implies r(h,v) \le r(\tau_{v}(r),v)$. Thus $r < r(h,v) \le r(\tau_{v}(r),v)$ and thus $r(\tau_{v}(r),v)>r$, so $r(\cdot, v)$ must have a left discontinuity at $\tau_v(r)$.

\subsection{Proof of Theorem \ref{propflow}}
Along the way to proving Theorem \ref{propflow}, we will first establish the following result:
\begin{lemma} \label{lemma:hdist}
	Assume MONO holds and that REG$_j$ holds for some $j \in \{1,\dots, J\}$. Assume further that $\{X_{i} \indep V_i \} \textrm{ }| \textrm{ } W_{i}$ (the first part of EXOG). Then:
	\begin{align*} 
		\partial_{x_j} &P(R_{i} \le r|x,w) \\ 
		&= -\mathbbm{E}\left\{\left.{f_H(\tau_{V_i}(r)|x,V_i,w)} \cdot \partial_{x_j} \left.Q_{H|XVW}(\alpha|x,V_i,w)\right|_{\alpha = F_{H|XVW}(\tau_{V_i}(r)|x,V_i,w)]} \right|W_i=w\right\}
	\end{align*}
\end{lemma}
\noindent Lemma \ref{lemma:hdist} is of independent interest, because it shows that under MONO, a regression of the distribution of $R$ on a component of $X$ can be decomposed into a linear combination of quantile regressions of $H$ on $X$ (conditional on $V$ and $W$). Beyond regularity conditions, this result only requires reporting heterogeneity $V$ to be conditionally independent of variation in $X_j$, and no causal assumptions. To interpret this result causally, we add the second part of EXOG and establish Theorem \ref{propflow}.

To establish both results, note that by the law of iterated expectation and Lemma \ref{prophonest}:
\begin{align*}
P(R_{i} \le r|X_i=x,W_i=w) &= \int dF_{UV|XW}(u,v|x,w) \cdot \mathbbm{1}(r(h(x,u),v) \le r)\\
&= \int dF_{UV|XW}(u,v|x,w) \cdot \mathbbm{1}(h(x,u) \le \tau_v(r))\\
&=  \int dF_{V|XW}(v|x,w) \int dF_{U|XVW}(u|x,v,w) \cdot \mathbbm{1}(h(x,u) \le \tau_v(r))\\
&= \int dF_{V|W}(v|w) \cdot \mathbbm{E}\left[\left.\mathbbm{1}(h(x,U_i) \le \tau_v(r))\right|X_i=x, V_i=v,W_i=w\right]\\
&= \int dF_{V|W}(v|w) \cdot P(H_i \le \tau_v(r)|X_i=x, V_i=v,W_i=w)
\end{align*}
where I have used $\left\{X_{i} \indep V_i\right\}|W_{i}$ in the second to last step to replace $F_{V|XW}$ with $F_{V|W}$.

By differentiating the equation $Q_{H|XVW}(F_{H|XVW}(h|x,v)|x,v,w) = h$ with respect to $x_j$, we have:
$$ \partial_{x_j} P(H_{i} \le h|X_i=x,V_i=v,W_i=w) = -f_H(h|x,v,w) \cdot \left.\partial_{x_j} Q_{H|XVW}(\alpha|x,v,w)\right|_{\alpha = F_{H|XVW}(h|x,v,w)}$$
By dominated convergence (using Assumption REG) we can move the derivative inside the expectation, and thus:
$$\partial_{x_j} P(R_{i} \le r|x) = -\int dF_{V|W}(v|w) \cdot f_H(\tau_v(r)|x,v,w)\cdot \partial_{x_j}  \left.Q_{H|XVW}(\alpha|x,v,w)\right|_{\alpha = F_{H|XVW}(\tau_v(r)|x,v,w)}$$
establishing Lemma \ref{lemma:hdist}.

Note that EXOG implies that $\{X_{ji} \indep U_i\}|(X_{-j,i},V_i,W_{i})$, where $X_{-j,i}$ denotes all of the components of $X_i$ aside from the $j^{th}$. The theorem of \citet{hoderleinmammen2007} implies that given this and REG:
$$\partial_{x_j} \left.Q_{H|XVW}(\alpha|x,v,w)\right|_{\alpha = F_{H|XVW}(\tau_v(r)|x,v,w)}=\mathbbm{E}\left[\partial_{x_j} h(x,U_i)|H_i=\tau_v(r),x,v,w\right]$$
Therefore:
\begin{align*}
\partial_{x_j} &P(R_{i} \le  r|x) = -\int dF_{V|W}(v|w) \cdot f_H(\tau_v(r)|x,v)\cdot \mathbbm{E}\left[\partial_{x_j} h(x,U_i)|H_i=\tau_v(r),x,v\right]
\end{align*}
In the case where $V$ is degenerate, a similar proof to the above is used in \citet{chernozukovetal2019} to study derivatives of conditional choice probabilities in multinomial choice models (under somewhat different regularity conditions).

In the proof of Theorem 1 in \citet{hoderleinmammen2007}, the conditional expectation function analogous to $\mathbbm{E}\left[\partial_{x_j} h(x,U_i)|H_i=h,X_i=x,V_i=v,W_i=w\right]$ appearing in the expression for $\partial_{x_j} P(R_{i} \le  r|X_i=x)$ is defined to be the following integral:
\begin{align}
	\int dt \cdot t \cdot &\frac{f_{H,\partial_{x_j}h(x,U)|XVW}(h,t|x,v,w)}{f_{H|XVW}(h|x,v,w)}= \int dt \cdot t \cdot \frac{\partial_{t} \partial_h P(H_i \le h,\partial_{x_j}h(x,U_i) \le t|x,v,w)}{\partial_h P(H_i \le h|x,v,w)} \nonumber\\
	&\int dt \cdot t \cdot \partial_{t} \left\{\frac{\lim_{\epsilon \downarrow 0} P(H_i \in [h,h+\epsilon],\partial_{x_j}h(x,U_i) \le t|x,v,w)/\epsilon}{\lim_{\epsilon \downarrow 0} P(H_i \in [h,h+\epsilon]|x,v,w)/\epsilon}\right\} \nonumber \\
	&\int dt \cdot t \cdot \partial_{t} \left\{\lim_{\epsilon \downarrow 0} \frac{P(\partial_{x_j}h(x,U_i) \le t, h(x,U_i) \in [h,h+\epsilon]|x,v,w)}{P(h(x,U_i) \in [h,h+\epsilon]|x,v,w)}\right\} \label{eq:cefdef}
\end{align}
where $f_{H|XVW}$ and $f_{H,\partial_{x_j}h(x,U)|XVW}$ exist and have a ratio that is dominated by an absolutely integrable function $M \cdot c(t)$, by Assumption $REG_j$. Given that $U_i$ is a random vector in $\mathbbm{R}^{d_U}$ with a well-defined probability distribution conditional on $X_i=x,V_i=v,W_i=w$, the limit
$$\lim_{\epsilon \downarrow 0} \frac{P(\partial_{x_j}h(x,U_i) \le t, h(x,U_i) \in (h,h+\epsilon]|x,v,w)}{P(h(x,U_i) \in (h,h+\epsilon]|x,v,w)}$$
yields a regular conditional probability distribution of $\partial_{x_j}h(x,U_i)$ given $H_i=h(x,U_i)=h$ (and $X_i=x,V_i=v,W_i=w$). See the result of \citet{asderivatives} for details.

Under an interchange of the limit and the integral in \eqref{eq:cefdef}, we could also write the quantity $
\mathbbm{E}[\partial_{x_j} h(x,U_i)|H_i=h,X_i=x,V_i=v,W_i=w]$ as $\lim_{\epsilon \downarrow 0} \hspace{.cm} \mathbbm{E}[\partial_{x_j} h(x,U_i)|H_i \in [h,h+\epsilon],X_i=x,V_i=v,W_i=w]$. I employ this limit representation to offer an intuitive description of the Theorem \ref{propflow} estimand as a weighted average over $\partial_{x_j} h(x,U_i)$, among individuals having $H_i$ ``near'' $\tau_{V_i}(r)$, in the limit that $\epsilon \downarrow 0$.\footnote{One can also establish Theorem \ref{propflow} intuitively by applying Theorem \ref{propdiscrete} and letting $x' \rightarrow x$ (see Footnote \ref{fn:lemmasmall}).} \citet{sasaki_2015} shows how such outcome-conditioned average derivatives can be written as an explicit integral over the distribution of heterogeneity values $U_i\in \mathbbm{R}^{d_U}$ such that $h(x,U_i)=\tau_v(r)$.

\subsection{Proof of Corollary \ref{corr:avgderivative}}
Let $\beta_r(x,v,w):=\mathbbm{E}\left[\partial_{x_j} h(x,U_i)|H_i=\tau_{v}(r),x,v,w\right]$. Averaging Eq \eqref{eq:propflow} over $X_i,W_i$ yields:
\begin{align*}
	\mathbbm{E}[&\partial_{x_j} P(R_i \le r|X_i,W_i)] \\
	&=- \int dF_{XW}(x,w) \cdot \int dF_{V|W}(v|w) \cdot f_{H|XVW}(\tau_{v}(r)|x,v,w) \cdot \beta_r(x,v,w)\\
	&=- \int dF_{XW}(x,w) \cdot \int dF_{V|XW}(v|x,w) \cdot f_{H|XVW}(\tau_{v}(r)|x,v,w) \cdot \beta_r(x,v,w)\\
	&=- \int dF_{XVW}(x,v,w)\cdot f_{H|XVW}(\tau_{v}(r)|x,v,w) \cdot \beta_r(x,v,w)\\
	&=- \int dF_{XVW|H}(x,v,w|\tau_v(r))\cdot f_{H}(\tau_{v}(r)) \cdot \beta_r(x,v,w)\\
	&=- \int dF_{V|H}(v|\tau_{v}(r)) \cdot f_{H}(\tau_{v}(r)) \int dF_{XW|VH}(x,w|v,\tau_v(r)) \cdot \mathbbm{E}\left[\partial_{x_j} h(X_i,U_i)|H_i=\tau_{v}(r),x,v,w\right]\\
	&=- \int dF_{V|H}(v|\tau_{v}(r)) \cdot f_{H}(\tau_v(r)) \cdot  \mathbbm{E}\left[\partial_{x_j} h(X_i,U_i)|H_i=\tau_{V_i}(r),V_i=v\right]\\
	&=- \int dF_{V}(v) \cdot f_{H|V}(\tau_v(r)|v) \cdot  \mathbbm{E}\left[\partial_{x_j} h(X_i,U_i)|H_i=\tau_{V_i}(r),V_i=v\right]
\end{align*}
using EXOG in the second step. Note that by Bayes' rule:
$$dF_{V|H-\tau_V(r)}(v|0) = f_{H-\tau_V(r)|V}(0|v)\cdot \frac{dF_V(v)}{f_{H-\tau_V(r)}(0)} = f_{H|V}(\tau_v(r)|v)\cdot \frac{dF_V(v)}{f_{H-\tau_V(r)}(0)}$$
and thus $dF_V(v)\cdot f_{H|V}(v|\tau_v(r)) = f_{H-\tau_V(r)}(0) \cdot dF_{V|H-\tau_V(r)}(v|0)$, where $f_{H|V}(\tau_v(r)|v)=f_{H-\tau_V(r)|V}(0|v)$ given that  $\tau_V(r)$ is a constant given $V=v$. Note  that existence of $f_{H-\tau_V(r)}(0)$ is guaranteed by Assumption $REG_j$, since by integrating item four of $REG_j$ we know that the density $f_{H|V}$ exists and thus the density $f_{H-\tau_V(r)|V}$ exists as well. Thus:
\begin{align*}
	\mathbbm{E}[&\partial_{x_j} P(R_i \le r|X_i,W_i)] \\
	&=- \int dF_{V}(v) \cdot f_{H|V}(\tau_v(r)|v) \cdot  \mathbbm{E}\left[\partial_{x_j} h(X_i,U_i)|H_i=\tau_{V_i}(r),V_i=v\right]\\
	&=- f_{H-\tau_V(r)}(0) \cdot \int dF_{V|H-\tau_V(r)}(v|0)\cdot  \mathbbm{E}\left[\partial_{x_j} h(X_i,U_i)|H_i-\tau_{V_i}(r)=0,V_i=v\right]\\
	&=- f_{H-\tau_V(r)}(0) \cdot \mathbbm{E}\left[\partial_{x_j} h(X_i,U_i)|H_i-\tau_{V_i}(r)=0\right]=- f_{H-\tau_V(r)}(0) \cdot \mathbbm{E}\left[\partial_{x_j} h(X_i,U_i)|H_i=\tau_{V_i}(r)\right]
\end{align*}

\subsection{Proof of Theorem \ref{propdiscrete}}
I begin with a heuristic overview: the detailed proof is below. The logic of the result is as follows: for a given individual having $V_i=v$, $R_i$ will be less than or equal to $r$ when $X_i=x'$, but not when $X_i=x$, if $\Delta_i< 0$ and $h(x,U_i) \in (\tau_v(r),\tau_v(r)+|\Delta_i|]$. This event increases the value of $P(R_{i} \le r|x,w')-P(R_{i} \le r|x,w)$. On the other hand, $R_i$ will be less than or equal to $r$ when $X_i=x$ but not when $X_i=x'$ when $\Delta_i > 0$ and $h(x,U_i) \in (\tau_v(r)-\Delta_i,\tau_v(r)]$. This event instead decreases the value of $P(R_{i} \le r|x',w)-P(R_{i} \le r|x,w)$. The RHS of Theorem \ref{propdiscrete} can be written as
$$\mathbbm{E}\left\{\left.\int_{\tau_{V_i}(r)-\Delta_i}^{\tau_{V_i}(r)} dy \cdot f_H(y|\Delta_i,X_i=x, V_i)\right|W_i=w\right\},$$
which averages over both positive and negative $\Delta_i$, covering both cases.

Now let us prove the result of Theorem \ref{propdiscrete}. By the law of iterated expectations, Lemma \ref{prophonest}, and then EXOG
\begin{align*}
	&P(R_{i} \le r|X_i=x',W_i=w)-P(R_{i} \le r|X_i=x,W_i=w)\\
	&= \int dF_{UV|XW}(u,v|x',w) \cdot \mathbbm{1}(r(h(x',u),v) \le r)-\int dF_{UV|XW}(u,v|x,w) \cdot \mathbbm{1}(r(h(x,u),v) \le r)\\
	&= \int dF_{UV|XW}(u,v|x',w) \cdot \mathbbm{1}(h(x',u) \le \tau_v(r))- \int dF_{UV|XW}(u,v|x,w) \cdot \mathbbm{1}(h(x,u) \le \tau_v(r))\\
	&= \int dF_{V|W}(v|w) \cdot \left\{P(h(x',U_i) \le \tau_v(r)|X_i=x',v,w)-P(h(x,U_i) \le \tau_v(r)|X_i=x,v,w)\right\}\\
	&= \int dF_{V|W}(v|w) \cdot \left\{P(h(x',U_i) \le \tau_v(r)|X_i=x,v,w)-P(h(x,U_i) \le \tau_v(r)|X_i=x, v,w)\right\}
\end{align*}
using that $X_i \indep U_i |W_i,V_i$ by EXOG in the last step. Thus:
\begin{align*}
	&P(R_{i} \le r|X_i=x',W_i=w)-P(R_{i} \le r|X_i=x,W_i=w)\\
	&= \int dF_{V|W}(v|w) \cdot \left\{P(h(x',U_i) \le \tau_v(r) \textrm{ but not } h(x,U_i) \le \tau_v(r)|x,v,w)\right.\\
	&\left. \hspace{2in}-P(h(x,U_i) \le \tau_v(r) \textrm{ but not } h(x',U_i) \le \tau_v(r)|x,v,w)\right\}\\
	&=\int dF_{V|W}(v|w) \cdot \left\{P(h(x',U_i) \le \tau_v(r) < h(x,U_i)|x,v,w)-P(h(x,U_i) \le \tau_v(r) < h(x',U_i) |x,v,w)\right\}\\
	&=\int dF_{V|W}(v|w) \cdot \left\{P(h(x,U_i) \in (\tau_v(r),\tau_v(r)-\Delta_i]|x,v,w)-P(h(x,U_i) \in (\tau_v(r)-\Delta_i,\tau_v(r)]|x,v,w)\right\}\\
	&=\int dF_{V|W}(v|w) \cdot \left\{P(H_i \in (\tau_v(r),\tau_v(r)-\Delta_i]|x,v,w)-P(H_i \in (\tau_v(r)-\Delta_i,\tau_v(r)]|x,v,w)\right\}\\
	&=-\int dF_{V|W}(v|w) \cdot \int dF_{\Delta|XVW}(\Delta|x,v,w)\cdot \left\{P(H_i \in (\tau_v(r),\tau_v(r)-\Delta]|\Delta,x,v,w)\right.\\
	&\hspace{4in} \left.-P(H_i \in (\tau_v(r)-\Delta,\tau_v(r)]|\Delta,x,v,w)\right\}\\
	&=-\int dF_{V|W}(v|w) \cdot \int dF_{\Delta|XVW}(\Delta|x,v,w) \int_{\tau_v(r)-\Delta}^{\tau_v(r)} dy \cdot f_H(h|\Delta,x,v,w)\\
	&=-\int dF_{V|W}(v|w) \cdot \int dF_{\Delta|VW}(\Delta|x,v,w)\cdot \bar{f}_H(\tau_v(r)|\Delta,x,v,w)\cdot \Delta\\
	&=-\int dF_{V|W}(v|w) \cdot \mathbbm{E}[\bar{f}_H(\tau_v(r)|\Delta_i,x,v,w)\cdot \Delta_i|V_i=v,W_i=w]\\
	&=-\mathbbm{E}[\bar{f}_H(\tau_{V_i}(r)|\Delta_i,x,V_i,w)\cdot \Delta_i|W_i=w]
\end{align*}
using EXOG and with the definition $\bar{f}_H(y|\Delta,x,v,w):= \frac{1}{\Delta} \int_{y-\Delta}^{y} f_H(h|\Delta,x,v,w) dh$.

\subsection{Proof of Proposition \ref{prop:matzkin}}
To fix the scale normalization, suppose that $g(x^*)=1$ for some $x^* \in \mathcal{X}$. Then, note that by the fundamental theorem of calculus, we may write
\begin{align*}
	\log g(x) = \int_{x^*}^{x} \nabla \log  g(x) \circ dv = \sum_{j=1}^J \int_{x^*_j}^{x_j} \partial_{x_j} \log g(x_1, \dots x_{j-1},t,0,\dots,0) dt
\end{align*}
where $\circ$ denotes a dot product and $dv$ traces any continuous path in $\mathcal{X}$ from $x^*$ to $x$, for example the one given after the second equality that integrates over each $x_j$ in turn.

If all components of $X$ are continuous and there are no controls, then note that for any $x \in \mathcal{X}$ we can identify $\partial_{x_j} g(x)/\partial_{x_k} g(x) = \partial_{x_j} \mathbbm{E}[R_i|x]/\partial_{x_k} \mathbbm{E}[R_i|x]$ for any $j, k \in 1 \dots J$ by Eq. (\ref{eqexpratio2}). By assumption that $g(x)$ is homogeneous of degree one, we have that $g(\lambda x) = \lambda g(x)$. ``Euler's theorem'' of homogeneous functions then implies that $ g(x) = \sum_{j=1}^J \partial_{x_j} g(x)\cdot x_j$ (this result can be obtained by differentiating $g(\lambda x) = \lambda g(x)$ with respect to $\lambda$ and evaluating at $\lambda=1$). Thus $ \left(\partial_{x_k} \log g(x)\right)^{-1}=\frac{g(x)}{\partial_{x_k} g(x)} = 1+\sum_{j \ne k} \frac{\partial_{x_j} g(x)}{\partial_{x_k} g(x)}\cdot x_j$. We now arrive at a constructive expression for $g(x)$ in terms of observables
\begin{equation} \label{eq:g}
	g(x)=e^{\int_{x^*_j}^{x_j}  \left(1+\sum_{j \ne k} \frac{\partial_{x_j} \mathbbm{E}[R_i|(x_1, \dots x_{j-1},t,0,\dots,0)]}{\partial_{x_k} \mathbbm{E}[R_i|(x_1, \dots x_{j-1},t,0,\dots,0) ]}\cdot x_j\right)^{-1}dt}
\end{equation}

\subsection{Proof of Proposition \ref{prop:chars}}
For any $x$ and $x'$ that differ in component $X_j$ only:
\begin{align}
&P(R_i \le r|X_i=x',A_i)-P(R_i \le r|X_i=x,A_i)\nonumber \\
&\hspace{.1cm}=\mathbbm{E}[\mathbbm{1}(H_i \le \tau_{V_i}(r))|x',A_i]-\mathbbm{E}[\mathbbm{1}(H_i \le  \tau_{V_i}(r))|x,A_i]\nonumber \\
&\hspace{.1cm}= \int dF_{V|XA}(v|x',A_i)\cdot  \mathbbm{E}[\mathbbm{1}(H_i \le \tau_{v}(r))|x',v,A_i]-\int dF_{V|XA}(v|x,A_i)\cdot  \mathbbm{E}[\mathbbm{1}(H_i \le \tau_{v}(r))|x,v,A_i]\nonumber \\
&\hspace{.1cm}= \int dF_{V|XA}(v|x,A_i)\cdot \{ \mathbbm{E}[\mathbbm{1}(h(x',U_i) \le \tau_{v}(r))|x',v,A_i]- \mathbbm{E}[\mathbbm{1}(h(x,U_i) \le \tau_{v}(r))|x,v,A_i]\}\nonumber \\
&\hspace{.1cm}= \int dF_{V|XA}(v|x,A_i)\cdot \{ \mathbbm{E}[\mathbbm{1}(h(x',U_i) \le \tau_{v}(r))- \mathbbm{1}(h(x,U_i) \le \tau_{v}(r))|x,v,A_i]\}\nonumber \\
&\hspace{.1cm}= \int dF_{V|XA}(v|x,A_i)\cdot \{ \mathbbm{E}[\mathbbm{1}(h(x',U_i) \le \tau_{v}(r) < h(x,U_i))|x,v,A_i]\} \nonumber \\
&\hspace{2cm} - \int dF_{V|XA}(v|x,A_i)\cdot \{ \mathbbm{E}[\mathbbm{1}(h(x,U_i) \le \tau_{v}(r) < h(x',U_i))|x,v,A_i]\} \label{eq:diffcondA}
\end{align}
where in the second equality I have used that $\{X_{ji} \indep V_i\}|(A_i,W_i)$ so that $F_{V|XA}(v|x',a)=F_{V|XA}(v|x,a)$ for all $a$, and in the fourth equality that $\{X_{ji} \indep U_i\}|(A_i,W_i,V_i)$ so that
\begin{align*}
\mathbbm{E}[ \mathbbm{1}(h(x',U_i) \le  \tau_v(r))|X_i=x',V_i=v,A_i]=\mathbbm{E}[ \mathbbm{1}(h(x',U_i) \le  \tau_v(r))|X_i=x,V_i=v,A_i]
\end{align*}
Given \eqref{eq:diffcondA}, we have that
\begin{align*}
&\mathbbm{E}[A_i \cdot \{P(R_i \le r|x',A_i)-P(R_i \le r|x,A_i)\}|X_i=x]\\
&\hspace{1cm}= \int dF_{A|X}(a|x)\cdot a \cdot \int dF_{V|XA}(v|x,a)\cdot \{ \mathbbm{E}[\mathbbm{1}(h(x',U_i) \le \tau_{v}(r) < h(x,U_i)|x,v,a]\} \nonumber \\
&\hspace{2cm} - \int dF_{A|X}(a|x)\cdot a \cdot \int dF_{V|XA}(v|x,a)\cdot \{ \mathbbm{E}[\mathbbm{1}(h(x,U_i) \le \tau_{v}(r) < h(x',U_i)|x,v,a]\}\\
&=\mathbbm{E}[A_i \cdot \mathbbm{1}(h(x',U_i) \le \tau_{V_i}(r) < h(x,U_i))|X_i=x]-\mathbbm{E}[A_i \cdot \mathbbm{1}(h(x,U_i) \le \tau_{V_i}(r) < h(x',U_i))|X_i=x]
\end{align*}
and similarly 
\begin{align*}
&\mathbbm{E}[P(R_i \le r|x',A_i)-P(R_i \le r|x,A_i)|X_i=x]\\
&\hspace{.5cm}=\mathbbm{E}[\mathbbm{1}(h(x',U_i) \le \tau_{V_i}(r) < h(x,U_i))|X_i=x]-\mathbbm{E}[\mathbbm{1}(h(x,U_i) \le \tau_{V_i}(r) < h(x',U_i))|X_i=x]
\end{align*}
Note that assuming the numerator and denominator below both exist, we can write
\begin{align*}
&\frac{\mathbbm{E}[A_i \cdot \partial_{x_{j}}P(R_i \le r|X_i=x,A_i)|X_i=x]}{\mathbbm{E}[\partial_{x_{j}}P(R_i \le r|X_i=x,A_i)|X_i=x]}\\
&\hspace{1cm}= \frac{\lim_{x' \downarrow x} \frac{1}{||x'-x||}\cdot \mathbbm{E}[A_i \cdot \{P(R_i \le r|x',A_i)-P(R_i \le r|x,A_i)\}|X_i=x]}{\lim_{x' \downarrow x} \frac{1}{||x'-x||} \cdot \mathbbm{E}[P(R_i \le r|x',A_i)-P(R_i \le r|x,A_i)|X_i=x]}\\
&\hspace{1cm}=\lim_{x' \downarrow x}\frac{\cancel{\frac{1}{||x'-x||}} \cdot \mathbbm{E}[A_i \cdot \{P(R_i \le r|x',A_i)-P(R_i \le r|x,A_i)\}|X_i=x]}{\cancel{\frac{1}{||x'-x||}} \cdot \mathbbm{E}[P(R_i \le r|x',A_i)-P(R_i \le r|x,A_i)|X_i=x]}\\
&\hspace{1cm} = \lim_{x' \downarrow x}\frac{\mathbbm{E}[A_i \cdot \{P(R_i \le r|x',A_i)-P(R_i \le r|x,A_i)\}|X_i=x]}{\mathbbm{E}[P(R_i \le r|x',A_i)-P(R_i \le r|x,A_i)|X_i=x]}
\end{align*}
where $x'$ is a sequence of vectors that differ from $x$ only in the $j^{th}$ component, and I've assumed dominated convergence so that we can interchange the limits and expectations.

Thus, by the above:
\begin{align}
&\frac{\mathbbm{E}[A_i \cdot \partial_{x_{j}}P(R_i \le r|X_i=x,A_i)|X_i=x]}{\mathbbm{E}[\partial_{x_{j}}P(R_i \le r|X_i=x,A_i)|X_i=x]} \nonumber \\
&=\lim_{x' \downarrow x} \quad \frac{\mathbbm{E}[A_i \cdot \mathbbm{1}(h(x',U_i) \le \tau_{V_i}(r) < h(x,U_i))|X_i=x]}{\mathbbm{E}[\mathbbm{1}(h(x',U_i) \le \tau_{V_i}(r) < h(x,U_i))|X_i=x]-\mathbbm{E}[\mathbbm{1}(h(x,U_i) \le \tau_{V_i}(r) < h(x',U_i))|X_i=x]} \nonumber \\
&\hspace{.0cm}-\lim_{x' \downarrow x} \quad\frac{\mathbbm{E}[A_i \cdot \mathbbm{1}(h(x,U_i) \le \tau_{V_i}(r) < h(x',U_i))|X_i=x]}{\mathbbm{E}[\mathbbm{1}(h(x',U_i) \le \tau_{V_i}(r) < h(x,U_i))|X_i=x]-\mathbbm{E}[\mathbbm{1}(h(x,U_i) \le \tau_{V_i}(r) < h(x',U_i))|X_i=x]} \label{eq:charestimand}
\end{align}	
Assume that $\partial_{x_j} h(x,U_i)$ has the same sign for all $i$ ``uniformly'' in the sense that
either the first or the second term above is zero. Suppose for example that $\partial_{x_j} h(x,U_i) \ge 0$ with probability one. Then:
\begin{align*}
&\frac{\mathbbm{E}[A_i \cdot \partial_{x_{j}}P(R_i \le r|X_i=x,A_i)|X_i=x]}{\mathbbm{E}[\partial_{x_{j}}P(R_i \le r|X_i=x,A_i)|X_i=x]}\\
&=\lim_{x' \downarrow x} \quad \frac{\mathbbm{E}[A_i \cdot \mathbbm{1}(h(x',U_i) \le \tau_{V_i}(r) < h(x,U_i))|X_i=x]}{\mathbbm{E}[\mathbbm{1}(h(x',U_i) \le \tau_{V_i}(r) < h(x,U_i))|X_i=x]}\\
&=\lim_{x' \downarrow x} \quad \frac{P(h(x',U_i) \le \tau_{V_i}(r) < h(x,U_i)|X_i=x) \cdot \mathbbm{E}[A_i|h(x',U_i) \le \tau_{V_i}(r) < h(x,U_i),X_i=x]}{P(h(x',U_i) \le \tau_{V_i}(r) < h(x,U_i)|X_i=x)}\\
&=\lim_{x' \downarrow x} \quad \mathbbm{E}[A_i|h(x',U_i) \le \tau_{V_i}(r) < h(x,U_i),X_i=x] = \mathbbm{E}[A_i|h(x,U_i) = \tau_{V_i}(r),X_i=x]\\
&=\mathbbm{E}[A_i|H_i = \tau_{V_i}(r),X_i=x]
\end{align*}
provided that the RHS of the last line is well-defined. Similarly, if $\partial_{x_j} h(x,U_i) \le 0$ with probability one, then the LHS above evaluates to $\lim_{x' \downarrow x} \quad \mathbbm{E}[A_i|h(x,U_i) \le \tau_{V_i}(r) < h(x',U_i),X_i=x] = \mathbbm{E}[A_i|H_i = \tau_{V_i}(r),X_i=x]$ and we thus obtain the same expression.

More generally, if the sign of treatment effects vary by unit:
\begin{align}
&\frac{\mathbbm{E}[A_i \cdot \partial_{x_{j}}P(R_i \le r|X_i=x,A_i)|X_i=x]}{\mathbbm{E}[\partial_{x_{j}}P(R_i \le r|X_i=x,A_i)|X_i=x]} \nonumber \\
&=\lim_{x' \downarrow x}  \quad \frac{P(h(x',U_i) \le \tau_{V_i}(r) < h(x,U_i)|X_i=x)\cdot \mathbbm{E}[A_i|h(x',U_i) \le \tau_{V_i}(r) < h(x,U_i),X_i=x]}{P(h(x',U_i) \le \tau_{V_i}(r) < h(x,U_i)|X_i=x)-P(h(x,U_i) \le \tau_{V_i}(r) < h(x',U_i)|X_i=x)} \nonumber \\
&\hspace{.5cm}-\lim_{x' \downarrow x}  \quad\frac{P(h(x,U_i) \le \tau_{V_i}(r) < h(x',U_i)|X_i=x)\cdot \mathbbm{E}[A_i|h(x,U_i) \le \tau_{V_i}(r) < h(x',U_i),X_i=x]}{P(h(x',U_i) \le \tau_{V_i}(r) < h(x,U_i)|X_i=x)-P(h(x,U_i) \le \tau_{V_i}(r) < h(x',U_i)|X_i=x)} \label{eq:chargenera}
\end{align} and the estimand $\frac{\mathbbm{E}[A_i \cdot \partial_{x_{j}}P(R_i \le r|X_i=x,A_i)|X_i=x]}{\mathbbm{E}[\partial_{x_{j}}P(R_i \le r|X_i=x,A_i)|X_i=x]}$ yields a non-convex combination of $\mathbbm{E}[A_i|\mathbbm{1}(h(x,U_i) \le \tau_{V_i}(r)<h(x',U_i),X_i=x]$ and $\mathbbm{E}[A_i|\mathbbm{1}(h(x',U_i) \le \tau_{V_i}(r)<h(x,U_i),X_i=x]$.

A sufficient condition for $\partial_{x_j} h(x,U_i)$ to have the same sign for all $i$ ``uniformly'' in the above sense is that for all $x'$ within some neighborhood of $x$, $h(x',U_i)$ is either strictly increasing or strictly decreasing in component $j$ of $x'$, for all $U_i$. Specifically, let $x'(\delta)$ be the vector $x$ but with $\delta$ added to the $j^{th}$ component. Then, for some $\bar{\delta}>0$, we have we have that $P(h(x,U_i) \ge h(x'(\delta),U_i))=1$ or $P(h(x,U_i) \ge h(x'(\delta),U_i))=0$ for any $\delta \le \bar{\delta}$ (i.e. $x'$ and $x$ are sufficiently close). Then given that $\mathbbm{E}[A_i|H_i = \tau_{V_i}(r),X_i=x]$ is well-defined we have either that $\lim_{x' \downarrow x} \frac{P(h(x',U_i) \le \tau_{V_i}(r) < h(x,U_i)|X_i=x)}{P(h(x,U_i) \le \tau_{V_i}(r) < h(x',U_i)|X_i=x)} = 0$ if $\partial_{x_j} h(x,U_i) \ge 0$ with probability one, or
that $\lim_{x' \downarrow x} \frac{P(h(x,U_i) \le \tau_{V_i}(r) < h(x',U_i)|X_i=x)}{P(h(x',U_i) \le \tau_{V_i}(r) < h(x,U_i)|X_i=x)} = 0$ if $\partial_{x_j} h(x,U_i) \le 0$ with probability one. In either case one term of \eqref{eq:chargenera} evaluates to zero and the other to $\mathbbm{E}[A_i|H_i = \tau_{V_i}(r),X_i=x]$.

To see that \eqref{eq:Amean2} holds under the stronger condition that $\{X_{ji} \indep (A_i,U_i,V_i)\}|W_i$, we have in this case by similar steps as above:
\begin{align}
&\mathbbm{E}[A_i\cdot \mathbbm{1}(R_i \le r)|X_i=x']-\mathbbm{E}[A_i\cdot \mathbbm{1}(R_i \le r)|X_i=x]\nonumber \\
&\hspace{.1cm}=\mathbbm{E}[A_i \cdot \mathbbm{1}(H_i \le \tau_{V_i}(r))|x']-\mathbbm{E}[A_i \cdot \mathbbm{1}(H_i \le  \tau_{V_i}(r))|x]\nonumber \\
&\hspace{.1cm}= \int dF_{VA|X}(v,a|x')\cdot a \cdot  \mathbbm{E}[\mathbbm{1}(H_i \le \tau_{v}(r))|x']-\int dF_{VA|X}(v,a|x)\cdot a \cdot  \mathbbm{E}[\mathbbm{1}(H_i \le \tau_{v}(r))|x]\nonumber \\
&\hspace{.1cm}= \int dF_{VA|X}(v,a|x)\cdot a \cdot \{ \mathbbm{E}[\mathbbm{1}(h(x',U_i) \le \tau_{v}(r))|x']- \mathbbm{E}[\mathbbm{1}(h(x,U_i) \le \tau_{v}(r))|x]\}\nonumber \\
&\hspace{.1cm}= \int dF_{VA|X}(v,a|x)\cdot a \cdot \{ \mathbbm{E}[\mathbbm{1}(h(x',U_i) \le \tau_{v}(r))- \mathbbm{1}(h(x,U_i) \le \tau_{v}(r))|x]\}\nonumber \\
&\hspace{.1cm}= \int dF_{VA|X}(v,a|x)\cdot a \cdot \{ \mathbbm{E}[\mathbbm{1}(h(x',U_i) \le \tau_{v}(r) < h(x,U_i))|x]\} \nonumber \\
&\hspace{2cm} - \int dF_{VA|X}(v,a|x)\cdot a \cdot \{ \mathbbm{E}[\mathbbm{1}(h(x,U_i) \le \tau_{v}(r) < h(x',U_i))|x]\} \nonumber \\
&\hspace{.1cm}= \mathbbm{E}[A_i \cdot \{\mathbbm{1}(h(x',U_i) \le \tau_{V_i}(r) < h(x,U_i))-\mathbbm{1}(h(x,U_i) \le \tau_{v}(r) < h(x',U_i))\}|X_i=x]
\label{eq:diffcondAindep}
\end{align}
using that $\{X_{ji} \indep (V_i,A_i)\}|W_i$ so that $F_{VA|X}(v,a|x')=F_{VA|X}(v,a|x)$ for all $a$ in the third equality. Similarly:
\begin{align}
&P(R_i \le r)|X_i=x')-P(R_i \le r)|X_i=x)\nonumber \\
&\hspace{.1cm}= \mathbbm{E}[\mathbbm{1}(h(x',U_i) \le \tau_{V_i}(r) < h(x,U_i))-\mathbbm{1}(h(x,U_i) \le \tau_{v}(r) < h(x',U_i))|X_i=x] \label{eq:diffcondAindep2}
\end{align}
And thus
$$\frac{\partial_{x_j} \mathbbm{E}[A_i \cdot \mathbbm{1}(R_i \le r)|X_i=x]}{\partial_{x_j} P(R_i \le r|X_i=x)} = \lim_{x' \downarrow x} \frac{\mathbbm{E}[A_i\cdot \mathbbm{1}(R_i \le r)|X_i=x']-\mathbbm{E}[A_i\cdot \mathbbm{1}(R_i \le r)|X_i=x]}{P(R_i \le r)|X_i=x')-P(R_i \le r)|X_i=x)}
$$
under suitable regularity conditions to take the derivative outside of the expectation. Given \eqref{eq:diffcondAindep} and \eqref{eq:diffcondAindep2}, the above yields the same estimand as \eqref{eq:charestimand}, again simplifying to $\mathbbm{E}[A_i|h(x,U_i) = \tau_{V_i}(r), X_i=x]$ given the common sign of derivative $\partial_{x_j}h(x,U_i)$ across all individuals $i$.

\subsection{Proof of Corollary \ref{corr:uniformv}}
The proof will make use of the following lemma:
\begin{lemma*} 
	Let $A$ and $B$ be two independent random variables defined on a common probability space, and $g$ a measurable function of $A$. Let $C=A-B$. Then $A \indep C$ iff $B$ is distributed uniformly on its support.
\end{lemma*}
\begin{proof}
	Let the supports of $A$, $B$ and $C$ be $\mathcal{A}$, $\mathcal{B}$, and $\mathcal{C}$. Note that by the law of iterated expectations
	$$P(C \le c|A=a) = P(g(a)-B \le c|A=a) = P(B \ge g(a)-c)$$
	using independence between $A$ and $B$.
	
	Now consider the case that $B$ is uniformly distributed on $\mathcal{B}$. $P(C \le c|A=a)=1-F_{B}(g(a)-c)$ and $C$ thus has a density conditional on $A$:
	\begin{align} \label{eq:densc}
		f_{C|A=a}(c) = \partial_c P(C \le c|A=a) = f_{B}(g(a)-c)
	\end{align}
	Since $f_{B}(g(a)-c)$ is constant on $\mathcal{B}$, $f_{C|A=a}(c)$ does not depend on $a$ for any $c \in \mathcal{C}$, and hence $C$ and $A$ are independent.
	
	In the other direction, note that the RHS of \eqref{eq:densc} \textit{will} depend on $a$ for some $c$, provided that $f_{B}(b) \ne f_{B}(b')$ for some $b,b' \in \phi(\mathcal{B})$, where we say that $f_{B}(b) \ne f_{B}(b')$ if $f_{B}(t)$ and $f_{B}(t')$ both exist or have different values, or if the derivative of $F_{B}(\cdot)$ does not exist at one of the two points. Suppose that $F_{B}(\cdot)$ is differentiable at $b$ but not at $b'$. Then for any $a \in \mathcal{A}$, $P(C \le c|A=a)$ is differentiable at $c=g(a)-b$ but not at $c=g(a)-b'$, and hence $C$ and $A$ are not independent. Suppose instead that $F_{B}(\cdot)$ is differentiable at $b$ and $b'$ but $f_{B}(b) \ne f_{B}(b')$. Note that for any $b \in \mathcal{B}$ at which $F_B$ is differentiable, and any $c \in \mathcal{C}$, the conditional density $f_{C|g(A)=c+b}(c)$ exists and is equal to $f_{B}(b)$. Thus $f_{C|g(A)=c+b}(c) \ne f_{C|g(A)=c+b'}(c)$, and again $C$ and $A$ are not independent.
\end{proof}

Let $\tau_{ri} = \tau_{V_i}(r)$ and define $\tau_i$ to be a vector of $\tau_{ri}$ across $r \in \mathcal{R}$. Given $(U_i,V_i) \indep X_i |W_i$. If this independence assumption holds with $V_i=\tau_i$, then it also holds with $\tau_{i}(r)$, and we can rewrite Eq. \eqref{eq:propflow} as a one-dimensional integral over $\tau_{ri}$:
\begin{align*}
	&\partial_{x_j} P(R_i \le r|x,w)\\
	& =-\int dF_{\tau_r|W}(t|w) \cdot f_{h(x,U)}(t|\tau_{ri}=t,x,w) \cdot \mathbbm{E}\left[\partial_{x_j} h(x,U_i)|h(x,U_i)=t,\tau_{ri}=t,x,w\right]\\
	&=-\int dF_{\tau_r|W}(t|w) \cdot f_{h(x,U)}(t|\tau_{ri}=t,x,w) \cdot \mathbbm{E}\left[\partial_{x_j} h(x,U_i)|h(x,U_i)=t,h(x,U_i)-\tau_{ri}=0,x,w\right]
\end{align*}
Under the assumption that $\tau_{ri} \sim Unif[\ell_w,u_w]$ and $V_i \indep U_i|X_i,W_i$, we can replace $dF_{\tau_r|W}(t|w)$ with $\frac{dt}{\mu_w-\ell_w} \cdot \mathbbm{1}(\ell_w \le t \le \mu_w)$ and we have by the Lemma above that $U_i \indep \{h(x,U_i)-\tau_{ri}\}|X_i=x,W_i=w$. Thus:
\begin{align*}
	\partial_{x_j} P(R_i \le r|x,w) &=-\frac{1}{\mu_w-\ell_w} \cdot\int_{\ell_w}^{u_w} dt  \cdot f_{h(x,U)}(t|x,w) \cdot \mathbbm{E}\left[\partial_{x_j} h(x,U_i)|h(x,U_i)=t,x,w\right]
\end{align*}
since provided that $t \in \textrm{supp}\{\tau_{ri}|W_i=w\} = [\ell_w,\mu_w]$, $f_{h(x,U)}(t|\tau_{ri}=t,x,w) =f_{h(x,U)}(t|x,w)$ and $\mathbbm{E}\left[\partial_{x_j} h(x,U_i)|h(x,U_i)=t,h(x,U_i)-\tau_{ri}=0,x,w\right]=\mathbbm{E}\left[\partial_{x_j} h(x,U_i)|h(x,U_i)=t,x,w\right]$. Meanwhile:
\begin{align*}
	\mathbbm{E}\left[\partial_{x_j} h(x,U_i)|x,w\right] &=\int dt \cdot f_{h(x,U)}(t|x,w) \cdot \mathbbm{E}\left[\partial_{x_j} h(x,U_i)|h(x,U_i)=t,x,w\right]\\
	&=\int_{\ell_w}^{u_w} dt \cdot f_{h(x,U)}(t|x,w) \cdot \mathbbm{E}\left[\partial_{x_j} h(x,U_i)|h(x,U_i)=t,x,w\right]
\end{align*}
using that $\textrm{supp}\{h(x,U_i)\} \subseteq [\mu_w,\ell_w]$ in the second equality.

Combining, we have that $\partial_{x_j} P(R_i \le r|x,w) =-\frac{1}{\mu_w-\ell_w} \cdot \mathbbm{E}\left[\partial_{x_j} h(x,U_i)|x,w\right]$.

\subsection{Proof of Proposition \ref{propidr}}
Using integration by parts and IDR:
\begin{align*}
	&P(R_i \le r|X_i=x,W_i=w) = \int_h P(r(h,V_i) \le r|H_i=h, X_i=x,W_i=w) \cdot dF_{H|XW}(h|x,w)\\
	&=\int_h P(r(h,V_i) \le r|W_i=w) \cdot dF_{H|XW}(h|x,w)\\
	&=\left.F_{H|XW}(h|x,w)P(r(h,V_i) \le r)\right|_h -\int_h F_{H|XW}(h|x,w) \cdot \frac{d}{dh}P(r(h,V_i) \le r|W_i=w) \cdot dh
\end{align*}
This implies that
\begin{align*}
	P(R_i \le r|x',w)&-P(R_i \le r|x,w)\\& = -\int_h \left\{F_{H|XW}(h|x',w)-F_{H|X}(h|x,w)\right\} \cdot \frac{d}{dh}P(r(h,V_i) \le r|W_i=w)\\
	&= -\int_h \left\{F_{H|XW}(h|x',w)-F_{H|XW}(h|x,w)\right\} \cdot \frac{d}{dh}P(h \le \tau_{V_i}(r)|W_i=w)\\
	&= \int_h \left\{F_{H|X}(h|x')-F_{H|X}(h|x)\right\} \cdot f_{\tau_V|W}(h|w)
\end{align*} 
since the first term does not depend on $x$. 

\subsection{Proof of Proposition \ref{propidr2}}
The following sequence of steps uses the law of iterated expectations, then IDR, then $\mathbbm{E}[A_i|X_i=x,W_i=w]=\int_0^1 Q_{A|X=x,W=w}(u)\cdot du$ for any random variable $A$, and finally that $Q_{r(H,v)|X=x',W=w}(u)=r(Q_{H|X=x',W=w}(u),v)$ since $r(\cdot, v)$ is weakly increasing and left-continuous for all $v$ \citep{quantileequivariance}: 
\begin{align*}
	\mathbbm{E}[R_i|x',w]-\mathbbm{E}[R_i|x,w] &= \int dF_{V|W}(v|w)\cdot \mathbbm{E}[r(H_i,v)|x',v,w]-\mathbbm{E}[r(H_i,v)|x,v,w]\\
	&= \int dF_{V|W}(v|w)\cdot \left\{\mathbbm{E}[r(H_i,v)|x',w]-\mathbbm{E}[r(H_i,v)|x,w]\right\}\\
	&= \int dF_{V|W}(v|w)\cdot \int_{0}^1 \left\{Q_{r(H,v)|X=x',W=w}(u)-Q_{r(H,v)|X=x,W=w}(u)\right\} du\\
	&= \int dF_{V|W}(v|w)\cdot \int_{0}^1 \left\{r(Q_{H|X=x',W=w}(u),v)-r(Q_{H|X=x,W=w}(u),v)\right\} du\\
	&= \int_{0}^1 \left[ \int dF_{V|W}(v|w) \left\{r(Q_{H|X=x',W=w}(u),v)-r(Q_{H|X=x,W=w}(u),v)\right\}\right] du\\
	&= \int_{0}^1 \bar{r}'_{x',x}(u) \cdot \left\{Q_{H|X=x',W=w}(u)-Q_{H|X=x,W=w}(u)\right\}du
\end{align*}
where the interchange of integrals is warranted provided that each of $\mathbbm{E}[R_i|x',w]$ and $\mathbbm{E}[R|x,w]$ are finite, because
\begin{align*}
	\int dF_{V|W}(v|w)&\cdot \int_{0}^1 \left|r(Q_{H|X=x',W=w}(u),v)-r(Q_{H|X=x,W=w}(u),v)\right| du\\
	& \hspace{.3in} \le \int dF_{V|W}(v|w)\cdot \int_{0}^1 \left|r(Q_{H|X=x',W=w}(u),v)|+|r(Q_{H|X=x,W=w}(u),v)\right| du\\
	& \hspace{1in} = \mathbbm{E}[|R_i||x',w]-\mathbbm{E}[|R_i||x,w] < \infty
\end{align*}
Note as well that $\bar{r}'_{x',x}(u) \cdot \left\{Q_{H|X=x',W=w}(u)-Q_{H|X=x,W=w}(u)\right\}$ is always well-defined and equal to $r(Q_{H|X=x',W=w}(u),v)- r(Q_{H|X=x,W=w}(u),v)$, because $r(Q_{H|X=x',W=w}(u),v)\ne r(Q_{H|X=x,W=w}(u),v)$ implies that $Q_{H|X=x',W=w}(u)\ne Q_{H|X=x,W=w}(u)$.

\subsection{Proof of Proposition \ref{propflowcont}}
By the law of iterated expectations: $ \mathbbm{E}[R_i|X_i=x,W_i=w] = \int dF_{V|W}(v|w) \cdot \int dh \cdot r(h,v)\cdot f_H(h|x,v,w)$. Now use REG to move the derivative inside the integral:
$$ \partial_{x_j} \mathbbm{E}[R_i|X_i=x,W_i=w] = \int dF_{V|W}(v|w) \cdot \int dh \cdot r(h,v)\cdot \partial_{x_j} f_H(h|x,v,w) $$
Theorem 1 of \citet{kasy2022} (for a one-dimensional outcome) implies that
$\partial_{x_j} f_H(h|x,v,w) = -\frac{\partial}{\partial h} \left\{ f_H(h|x,v,w) \cdot \mathbbm{E}\left[\partial_{x_j} h(x,U_i)|H_i=h,x,v,w\right]\right\}$. Thus
\begin{align*}
	\partial_{x_j} \mathbbm{E}[R_i|x,w] = -\int dF_{V|W}(v|w) \int dh \cdot r(h,v) \cdot \frac{\partial}{\partial h} \left\{ f_H(h|x,v,w) \cdot  \mathbbm{E}\left[\partial_{x_j} h(x,U_i)|H_i=h,x,v,w\right]\right\}
\end{align*}
Now use integration by parts, applying the assumed boundary condition eliminates the first term, establishing the result:
\begin{align*}
	\partial_{x_j} \mathbbm{E}[R_i|x,w]
	&=0+\int dF_{V|W}(v|w) \int dh \cdot r'(h,v)\cdot f_h(h|x,v,w) \cdot \mathbbm{E}\left[\partial_{x_j} h(x,U_i)|H_i=h,x,v,w\right]
\end{align*}

\subsection{Proof of Proposition \ref{propdiscretedense}}
With the substitution $h=\tau_v(r)$, $dr = r'(h,v)\cdot dh$:
\begin{align*}
&\sum_r \int_{\tau_v(r)-\Delta}^{\tau_v(r)} dy \cdot f_H(y|\Delta,x, v,w) \stackrel{R}{\rightarrow}\bar{R}\cdot \int dr \int_{\tau_v(r)-\Delta}^{\tau_v(r)} dy \cdot f_H(y|\Delta,x,v,w)\\
&= \bar{R}\cdot \int dh\cdot  r'(h,v) \int_{h-\Delta}^{h} dy \cdot f_H(y|\Delta,x,v,w)=\bar{R}\cdot\int dy \int_{y}^{y+\Delta} dh \cdot r'(h,v)\cdot f_H(y|\Delta,x,v,w)\\
&=\bar{R}\cdot \int dy \cdot \Delta \cdot  \bar{r}'(y,\Delta,v)\cdot f_H(y|\Delta,x,v,w)=\Delta \cdot \bar{R}\cdot \mathbbm{E}[\bar{r}'(H_i,\Delta,v)|\Delta_i=\Delta,X_i=x, V_i=v,W_i=w]
\end{align*}
where $\bar{r}'(y,\Delta,v):= \frac{1}{\Delta} \int_{y}^{y+\Delta} r'(h,v) dh$. Thus:
\begin{align*}
\mathbbm{E}&[R_i|X_i=x',W_i=w]-\mathbbm{E}[R_i|X_i=x,W_i=w]\\
&=\bar{R}\cdot  \int dF_{V|W}(v|w) \cdot \int dF_{\Delta|XVW}(\Delta|x,v,w) \cdot \Delta \cdot\mathbbm{E}[\bar{r}'(H_i,\Delta,v)|\Delta_i=\Delta,X_i=x, V_i=v,W_i=w]\\
&=\bar{R}\cdot  \int dF_{V|W}(v|w) \cdot \int dF_{\Delta|XVW}(\Delta|x,v,w) \cdot \mathbbm{E}[\Delta \cdot\bar{r}'(H_i,\Delta,v)|\Delta_i=\Delta,X_i=x, V_i=v,W_i=w]\\
&=\bar{R}\cdot  \int dF_{V|W}(v|w) \cdot \mathbbm{E}\left[\left. \mathbbm{E}[\Delta_i\cdot\bar{r}'(H_i,\Delta_i,V_i)|\Delta_i=\Delta,X_i=x, V_i=v]\right|X_i=x,V_i=v\right]\\
&=\bar{R}\cdot  \int dF_{V|W}(v|w) \cdot \mathbbm{E}[\Delta_i\cdot\bar{r}'(H_i,\Delta_i,V_i)|X_i=x, V_i=v,W_i=w]\\
&=\bar{R}\cdot  \mathbbm{E}[\Delta_i\cdot\bar{r}'(H_i,\Delta_i,V_i)|X_i=x,W_i=w]
\end{align*}
Note that if we assume that $\Delta_i$ and $\bar{r}'(H_i,\Delta_i,V_i)$ are uncorrelated conditional on $X_i=x,W_i=w$, this reduces to
$$\mathbbm{E}[R_i|X_i=x']-\mathbbm{E}[R_i|X_i=x]=\bar{R}\cdot  \mathbbm{E}[\Delta_i|X_i=x]\cdot\mathbbm{E}[\bar{r}'(H_i,\Delta_i,V_i)|X_i=x]
$$

\subsection{Proof of Proposition \ref{propheterolinear}}
Starting with Proposition \ref{propdiscretedense}, observe that $\bar{r}'(y,\Delta,v):= \frac{1}{\Delta} \int_{y}^{y+\Delta} r'(h,v) dh$ is equal to
$$r'(v) \cdot \begin{cases}\frac{y-(\ell(v)-\Delta)}{|\Delta|}\cdot \mathbbm{1}(y \in [\ell(v)-\Delta, \ell(v)])+ \mathbbm{1}(y \in [\ell(v),\mu(v) - \Delta])\\
\hspace{2in} + \frac{\mu(v) - y}{\Delta}\cdot \mathbbm{1}(y \in [\mu(v)-\Delta, \mu(v)])& \textrm{ if } \Delta> 0\\
\frac{y-\ell(v)}{\Delta}\cdot \mathbbm{1}(y \in [\ell(v), \ell(v)+|\Delta|])+ \mathbbm{1}(y \in [\ell(v)+|\Delta|,\mu(v)])\\
\hspace{2in} + \frac{\mu(v)+|\Delta| - y}{|\Delta|}\cdot \mathbbm{1}(y \in [\mu(v), \mu(v)+|\Delta|]) & \textrm{ if } \Delta< 0
\end{cases}$$
where $r'(v) = \frac{|\mathcal{R}|}{\ell(v)-\mu(v)}$. To ease notation, let us for the moment make the conditioning implicit and let $f(y)$ denote $f_H(y|\Delta,x,v,w)$ and $F(y)$ the corresponding conditional CDF. Let us keep $v$ also implicit in both $\ell$ and $\mu$.  If we let $\theta$ denote the quantity $\frac{1}{r'(v)}\int dy \cdot \bar{r}'(y,\Delta,v)$ for a fixed $\Delta$, then:
\begin{align}
\theta = \begin{cases} [F(\ell)-F(\ell-\Delta)] \mathbbm{E}\left[\left.\frac{H_i-(\ell-\Delta)}{\Delta}\right|H_i \in [\ell-\Delta, \ell]\right] + F(\mu-\Delta)\\
\hspace{1in}-F(\ell)+ [F(\mu)-F(\mu-\Delta)] \mathbbm{E}\left[\left.\frac{\mu - H_i}{\Delta}\right|H_i \in [\mu-\Delta, \mu]\right] & \textrm{ if } \Delta > 0\\
[F(\ell+|\Delta|)-F(\ell)] \mathbbm{E}\left[\left.\frac{H_i-\ell}{|\Delta|}\right|H_i \in [\ell, \ell+|\Delta|]\right]+ F(\mu)\\
\hspace{1in}-F(\ell+|\Delta|)+ [F(\mu+|\Delta|)-F(\mu)] \mathbbm{E}\left[\left.\frac{\mu+\Delta - H_i}{|\Delta|}\right|H_i \in [\mu,\mu+|\Delta|] \right] & \textrm{ if } \Delta < 0 \end{cases}\label{eq:theta}
\end{align}
To get a lower bound on $\theta$, we use the assumption that $f(y)$ is increasing on the interval $[\ell-|\Delta|, \ell+|\Delta|]$, as well as decreasing on the interval  $[\mu-|\Delta|, \mu+|\Delta|]$:
\begin{align*} \theta \ge &\begin{cases} \frac{1}{2}[F(\ell)-F(\ell-\Delta)] + F(\mu-\Delta)-F(\ell)+ \frac{1}{2}[F(\mu)-F(\mu-\Delta)] & \textrm{ if } \Delta > 0\\
\frac{1}{2}[F(\ell+|\Delta|)-F(\ell)] + F(\mu)-F(\ell+|\Delta|) + \frac{1}{2}[F(\mu+|\Delta|)-F(\mu)] & \textrm{ if } \Delta < 0 \end{cases}\\
&\hspace{1.4in}=\begin{cases} \frac{1}{2}[F(\mu-\Delta)-F(\ell-\Delta)] + \frac{1}{2}[F(\mu)-F(\ell)] & \textrm{ if } \Delta > 0\\
\frac{1}{2}[F(\mu+|\Delta|)-F(\ell+|\Delta|)] + \frac{1}{2}[F(\mu)-F(\ell)] & \textrm{ if } \Delta < 0 \end{cases}\\
&\hspace{1.4in}=\frac{1}{2}[F(\mu-\Delta)-F(\ell-\Delta)] + \frac{1}{2}[F(\mu)-F(\ell)]\\
& =\frac{1}{2}[F(\mu(v)|\Delta,x',v)-F(\ell(v)|\Delta,x',v)]+\frac{1}{2}[F(\mu(v)|\Delta,x,v,w)-F(\ell(v)|\Delta,x,v,w)],
\end{align*}
reintroducing conditioning values with the notation $F(\cdot|\Delta,x,v,w):=F_{H|\Delta XVW}(\cdot|\Delta,x,v,w)$. A lower bound on the weight $\Pi_{x,x'}$ on causal effects in $\mathbbm{E}[R_i|x',w]-\mathbbm{E}[R_i|x,w]$ can thus given by averaging over $V_i$ (c.f. Proposition \ref{propdiscretedense}):
\begin{align*}
\Pi_{x,x'}& \ge \int dF_{V|W}(v|w) \cdot \int dF(\Delta|x,v,w)  \cdot \left\{\frac{1}{2}[F(\mu(v)|\Delta,x',v,w)-F(\ell(v)|\Delta,x',v,w)]\right.\\
&\hspace{3in}\left.+\frac{1}{2}[F(\mu(v)|\Delta,x,v,w)-F(\ell(v)|\Delta,x,v,w)]\right\}
\end{align*}
Note that this exactly the same as the average between the weights $\Pi_{x}$ and $\Pi_{x'}$ corresponding to using continuous variation at $X_i=x$ and $X_i=x'$, respectively. For example (c.f. Eq. \ref{compare:b}):
\begin{align*}
\Pi_{x}& = \int dF_{V|W}(v|w) \cdot \int dF(\Delta|x,v,w)   \cdot [F(\mu(v)|\Delta,x,v,w)-F(\ell(v)|\Delta,x,v,w)]
\end{align*}
This leads to the lower bound of $\Pi_{x,x'}/(\frac{1}{2}\Pi_{x}+\frac{1}{2}\Pi_{x'}) \ge 1$ in Proposition \ref{propdiscretedense}.

Now, to obtain an upper bound, notice that an upper bound on $\theta$ occurs if we imagine putting all of the mass in each of the interval conditional expectations in (\ref{eq:theta}) to the right in the intervals that depend on $\ell$, and at the left end for the intervals that depend on $\mu$. Then:
\begin{align*}
\theta &\le \begin{cases} \cancel{F(\ell)}-F(\ell-\Delta) + \cancel{F(\mu-\Delta)}-\cancel{F(\ell)}+ F(\mu)-\cancel{F(\mu-\Delta)}& \textrm{ if } \Delta > 0\\
\cancel{F(\ell+|\Delta|)}-F(\ell) + \cancel{F(\mu)}-\cancel{F(\ell+|\Delta|)} + F(\mu+|\Delta|)-\cancel{F(\mu)}& \textrm{ if } \Delta < 0 \end{cases}\\
&=\begin{cases} F(\mu)-F(\ell-\Delta) & \textrm{ if } \Delta > 0\\
=F(\mu+|\Delta|)-F(\ell) & \textrm{ if } \Delta < 0 \end{cases}=\begin{cases} F(\mu(v)|\Delta,x,v,w)-F(\ell(v)|\Delta,x',v,w) & \textrm{ if } \Delta > 0\\
F(\mu(v)|\Delta,x',v,w)-F(\ell(v)|\Delta,x,v,w) & \textrm{ if } \Delta < 0 \end{cases}
\end{align*}
where I've used that $F(y|\Delta,x',v,w) = F(y-\Delta|\Delta,x,v,w)$ in the last step. An upper bound for $\theta$ that applies to both cases can be obtained by adding them together:
\begin{equation}
\theta \le F(\mu(v)|\Delta,x,v,w)-F(\ell(v)|\Delta,x,v,w) + F(\mu(v)|\Delta,x',v,w)-F(\ell(v)|\Delta,x',v,w)
\end{equation}
where I've used that $F(\mu) \ge F(\ell - \Delta)$ and $F(\mu + |\Delta|) \ge F(\ell)$ are implied by the assumption that $f(y)$ is increasing on the interval $[\ell-|\Delta|, \ell+|\Delta|]$, while decreasing on the interval  $[\mu-|\Delta|, \mu+|\Delta|]$, which implies that $\mu - |\Delta| \ge \ell + |\Delta|$.

Thus, an upper bound on the weight $\Pi_{x,x'}$ on causal effects in $\mathbbm{E}[R_i|x',w]-\mathbbm{E}[R_i|x,w]$ is:
\begin{align*}
\Pi_{x,x'}& \ge \int dF_{V|W}(v|w) \cdot \int dF_{\Delta|XVW}(\Delta|x,v,w)  \cdot \left\{F(\mu(v)|\Delta,x',v,w)-F(\ell(v)|\Delta,x',v,w)\right.\\
&\hspace{3in}\left.+F(\mu(v)|\Delta,x,v,w)-F(\ell(v)|\Delta,x,v,w)\right\}
\end{align*}
leading to the upper bound of $\Pi_{x,x'}/(\frac{1}{2}\Pi_{x}+\frac{1}{2}\Pi_{x'}) \le 2$ in Proposition \ref{propdiscretedense}.

Now consider the final condition in Proposition \ref{propdiscretedense}. That $\Pi_{x,x'}/\Pi_{x} \ge 1/2$ follows from the above since $F(\mu(v)|\Delta,x',v,w)-F(\ell(v)|\Delta,x',v,w) \ge 0$ for all $\Delta,x,v,w$. For the upper bound we have
\begin{align*}
\frac{\Pi_{x}}{\Pi_{x,x'}} &\ge \frac{\mathbbm{E}\left\{\left.\frac{NB(x,V_i,w)}{\mu(V_i)-\ell(V_i)}\right|X_i=x,W_i=w\right\}}{\mathbbm{E}\left[\left.\frac{1}{\mu(V_i)-\ell(V_i)} \right|\right]}\\
& = \frac{\mathbbm{E}\left[\left.\frac{1}{\mu(V_i)-\ell(V_i)} \right|X_i=x,W_i=w\right]\cdot NB(x,w) - Cov\left[\left.\frac{1}{\mu(V_i)-\ell(V_i)}, NB(x,V_i,w) \right|X_i=x,W_i=w\right]}{\mathbbm{E}\left\{\left.\frac{1}{\mu(V_i)-\ell(V_i)}\right|X_i=x,W_i=w\right\}}\\
&\ge NB(x,w)-\sqrt{\frac{ Var\left[\left.\frac{1}{\mu(V_i)-\ell(V_i)} \right|X_i=x,W_i=w\right]}{\mathbbm{E}\left\{\left.\frac{1}{\mu(V_i)-\ell(V_i)}\right|X_i=x,W_i=w\right\}^2}\cdot Var\left[\left.NB(x,V_i,w) \right|X_i=x,W_i=w\right]}\\
&\ge NB(x,w)-Var\left[\left.NB(x,V_i,w) \right|X_i=x,W_i=w\right]\\
&\ge NB(x,w) - NB(x,w)\cdot(1-NB(x,w)) = NB(x,w)^2
\end{align*}
where $NB(x,v,w):=P(0 < R_i < \bar{R}|x,v,w) =P(\ell(V_i) \le H_i \le \mu(V_i)|x,v,w)$ and $NB(x,w) = \mathbbm{E}[NB(x,V_i,w)|x,w]$ is the observable probability of not bunching given $(X_i,W_i)=(x,w)$. The third inequality uses the assumption that $\frac{ Var\left[\left.\frac{1}{\mu(V_i)-\ell(V_i)} \right|x,w\right]}{\mathbbm{E}\left\{\left.\frac{1}{\mu(V_i)-\ell(V_i)}\right|x,w\right\}^2}\le Var\left[\left.NB(x,V_i,w) \right|x,w\right]$ and the last one that $Var\left[\left.NB(x,V_i,w) \right|x,w\right] \le NB(x,w)\cdot (1-NB(x,w))$ since $NB(x,v,w) \in [0,1]$ for all $x,v,w$.

From this notation, we obtain the form written in Proposition \ref{propheterolinear} by noting that  $NB(X_i,V_i,W_i)=1-\mathcal{B}_i$. Note that $Var\left[\left.NB(x,V_i,w) \right|x,w\right]=Var\left[\left.\mathcal{B}_i \right|x,w\right]$.

\end{appendices}

\end{document}